\documentclass[11pt]{article}
\usepackage[letterpaper, left=1in, right=1in, top=1in,bottom=1in]{geometry}
\usepackage{amsfonts}       % blackboard math symbols
\usepackage{nicefrac}       % compact symbols for 1/2, etc.
\usepackage{bbm}
\usepackage[ruled,noend]{algorithm2e}
\usepackage{bm}
\usepackage{tikz}
\usetikzlibrary{positioning,chains,fit,shapes,calc,arrows}
\usetikzlibrary{arrows.meta}
\usepackage{graphicx}
\usepackage{pgf}
\usepackage{mathtools}
\usepackage{mathrsfs}
\usepackage{amsmath,amsthm,amssymb}
\usepackage{comment}
\usepackage{xfrac}
\usepackage{accents}

\usepackage{eurosym}             % <---
\usepackage{caption}
\usepackage[export]{adjustbox}   % <---
\usepackage{booktabs, tabularx} 

\usepackage{auxiliary}

\usepackage{float}

\usepackage{url}      

\usepackage[pagebackref=true]{hyperref} \PassOptionsToPackage{pagebackref=true}{hyperref}   

\usepackage{cleveref}

\usepackage{color}              % Need the color package
\usepackage{color-edits}[suppress]
\addauthor{tl}{cyan}
\addauthor{ww}{orange}
\addauthor{df}{blue}
\addauthor{ep}{violet}
\addauthor{bs}{purple}

\title{Group fairness in dynamic refugee assignment}

\author{
Daniel Freund\thanks{Massachusetts Institute of Technology, \texttt{dfreund@mit.edu}}
\and Thodoris Lykouris\thanks{Massachusetts Institute of Technology, \texttt{lykouris@mit.edu}} 
\and Elisabeth Paulson\thanks{Harvard University and Stanford University, \texttt{epaulson@stanford.edu}}
\and Bradley Sturt\thanks{
University of Illinois at Chicago, \texttt{bsturt@uic.edu}
}
\and Wentao Weng\thanks{Massachusetts Institute of Technology, \texttt{wweng@mit.edu}}}

\date{First version: January 2023\\
Current version: January 2025\footnote{A preliminary version of this work was accepted for presentation at the Conference on Economics and Computation (EC 2023).}}

\begin{document}

\maketitle

\begin{abstract}
% !TEX root = main.tex
Ensuring that refugees and asylum seekers thrive (\emph{e.g.}, find employment) in their host countries is a profound humanitarian goal, and a primary driver of employment is the geographic location within a host country to which the refugee or asylum seeker is assigned. Recent research has proposed and implemented algorithms that assign refugees and asylum seekers to geographic locations in a manner that maximizes the average employment across all arriving refugees. While these algorithms can have substantial overall positive impact, using data from two industry collaborators we show that the impact of these algorithms can vary widely across key subgroups based on country of origin, age, or educational background. Thus motivated, we develop a simple and interpretable framework for incorporating \emph{group fairness} into the  dynamic refugee assignment problem. In particular, the framework can flexibly incorporate many existing and future  definitions of group fairness from the literature (\emph{e.g.}, Random, Proportionally Optimized within-group, and MaxMin). Equipped with our framework, we propose two bid-price algorithms that maximize overall employment while simultaneously yielding provable group fairness guarantees. Through extensive numerical experiments using various definitions of group fairness and real-world data from the U.S. and the Netherlands, we show that our algorithms can yield substantial improvements in group fairness  compared to an offline benchmark fairness constraints, with only small relative decreases ($\approx$ 1\%-5\%) in global performance. 
\end{abstract}
\newpage
\section{Introduction} \label{sec:intro}
% !TEX root = main.tex
Over two million new refugees and asylum seekers are projected to require resettlement in 2024  to escape  violence and persecution in their origin countries  \cite{unhcr2023}, and ensuring the successful integration of these displaced people is a profound humanitarian and policy goal. Successful integration (typically measured by gainful employment after resettlement\footnote{For example, in the US, the Refugee Act of 1980 mandates annual audits of proxies for this metric.}) is crucial both to enhance the quality of life of refugees and asylum seekers and to invigorate the local economy of host countries. The challenge of facilitating the integration of refugees and asylum seekers into society is typically the purview  of non-profit and governmental {resettlement agencies} such as the Swiss State Secretariat of Migration (SEM)  in Switzerland, the Central Agency for the Reception of Asylum Seekers (COA) in the Netherlands, and Global Refuge in the United States.

A primary driver for successful integration is the geographic location in a host country to which a refugee or asylum seeker is assigned. Traditionally, case officers within resettlement agencies have made assignments based on local quotas and their own judgment. In the past few years, however, innovations in analytics have given rise to machine learning (ML) models that predict integration outcomes using personal characteristics~\cite{bansak2018improving}. The advent of ML in this context opens up new possibilities for using analytics to improve resettlement decisions.

Equipped with ML models, the {dynamic refugee assignment problem} faced by resettlement agencies can be formulated as follows. Over the course of an extended period of time (\emph{e.g.}, a fiscal year), the resettlement agency  sequentially receives new cases $t=1,\ldots,T$, where each case consists of a family or individual refugee/asylum seeker. 
The agency uses ML to compute an estimate of the probability $w_{t j}$ that case $t$ would  successfully find employment in each geographic location $j$. Each location $j$ {has a capacity}~$s_j$ on the number of refugees or asylum seekers that it is expected to resettle over that 
period. For some cases, the assignment is pre-determined by existing family ties or other constraints such as medical or educational considerations. For the \emph{free} cases---those that can be resettled to any location---the resettlement agency makes an irrevocable decision of which location to assign the case. This decision must not violate the capacity constraints at any geographic location, and the objective of the resettlement agency in this problem is to maximize the average employment {rate} across~all~cases.

Several recent works have tackled the dynamic refugee assignment problem via optimization. The first papers to propose the use of optimization and ML in this context (\cite{bansak2018improving} and later \cite{ahani2021placement}) study a static variant in which all refugees arrive simultaneously and develop exact algorithms based on network flows and integer programming. Subsequent works \cite{ahani2021dynamic,elisabethpaper} formulate a dynamic version of the problem where refugees arrive sequentially from a probability distribution and design efficient heuristics to overcome computational intractability. These algorithms yield significant predicted gains (up to 50\% increases) in average employment rate  when compared to assignment decisions under current practice. As such, they 
represent a significant leap in the use of analytics to drive societal impact. However, little is known about the fairness implications that global optimization may have on different subgroups of refugees, \emph{e.g.}, ones defined by country of origin, age, or educational background.

Why might group fairness be of central importance in the setting of dynamic refugee assignment? 
First, the use of a refugee assignment algorithm can be jeopardized from a legal standpoint if it has disparate impacts on refugees from different origin countries (see Recital 71 of the EU GDPR~\cite{vollmer_2022}).  Second, fairness of employment outcomes across key subgroups has been identified by decision-makers as an important goal when designing  algorithms. This sentiment is captured by Sjef van Grinsven, former Director of Projects at COA:
\begin{center}
\begin{minipage}{0.85\linewidth}

``\emph{Because of the ever growing impact of AI on society, the topic of fairness is becoming increasingly important. For an organization such as COA, a semi-governmental organization working with a vulnerable target group, implementing fairness is not what you would call desirable: it’s an absolute necessity. The challenge is to translate an abstract concept like fairness into concrete programmable choices.}''
\end{minipage}
\end{center}

\subsection{Our contributions}
\label{ssec:contributions}
In this paper, we address the needs of policy makers and refugee resettlement organizations by presenting the first rigorous study of group fairness in the dynamic refugee assignment problem. The contributions of this paper are three-fold. First, we present a flexible framework to incorporate group fairness. Second, we develop algorithms with theoretically and empirically good performance.  Third, we perform a comprehensive empirical evaluation of the trade-offs between global optimization and group fairness in refugee assignment. These contributions demonstrate that principled assignment algorithms 
can equip  
policy makers with practical tools for achieving group fairness with low cost to global optimization. 

\paragraph{Our framework. } Achieving the goal of incorporating group fairness into an algorithm for the dynamic refugee assignment problem is not a straightforward task. First, there is a large and ever-growing number of ways  in which a policy maker may describe the fairness (or lack thereof) of an assignment of refugees to locations (see Section~\ref{ssec:rel_work}). Second, it is not clear how to extend fairness rules that are specified for static assignments to settings where refugees arrive and are assigned sequentially. To be practically useful to policy makers, any framework to measure group fairness in dynamic refugee resettlement must be \emph{flexible} (seamlessly adapting to different fairness rules) and \emph{simple} (not requiring the policy maker to reason about the intricacies of the dynamic arrivals of refugees); see discussion in Section \ref{sec:practical-implications}. 

In this paper, we introduce such a  framework for incorporating  group fairness  into the dynamic refugee assignment problem (Section~\ref{ssec:formal_model}). Our framework does not require policy makers to define a notion of fairness that explicitly accounts for the stochastic arrival process of refugees. Instead, our framework only requires policy makers to specify an ex-post feasible fairness rule that, in turn, generates an ex-post \emph{minimum requirement} on the average employment probability for each group. By only requiring the fairness rule to be defined when all refugee arrivals are known upfront, our framework is easy to understand by practitioners and thus desirable from an adoptability standpoint. At the same time, the minimum requirement can be generated using a variety of simple and  interpretable definitions of group fairness from the literature, e.g., Random, Proportionally Optimized within groups and MaxMin. All together, our framework provides policy makers with a concrete way to evaluate whether a particular deployed algorithm achieves the specified fairness criteria (Section~\ref{ssec:model_objectives}); see further discussion regarding practical implications in Section~\ref{sec:practical-implications}.

\paragraph{Our algorithms.}
Equipped with our framework for reasoning about group fairness in the dynamic refugee resettlement setting, we perform a numerical analysis of how well existing approaches fare with respect to different fairness criteria (Section \ref{sec:emp-fairness-algos}). We analyze data from two partner organizations in the US and the Netherlands and show that a natural Bid Price Control (\textsc{BP}), similar to \cite{ahani2021dynamic}, 
can create frequent \emph{unfair} group outcomes, when considering various fairness criteria considered in the literature. We also show that this is not simply a result of \textsc{BP}, but also occurs with a clairvoyant control that has advance knowledge of arrivals and assigns them optimally. Additionally, we show that a random assignment algorithm (a proxy for status quo procedures in the absence of optimization \cite{bansak2018improving}) yields outcomes that are inefficient in terms of total employment and routinely fails to achieve group fairness according to the rules considered.\footnote{One might intuitively think of random assignment as being fair \emph{a priori}. However, recall that we are concerned with ex-post fairness rather than ex-ante fairness or fairness in expectation.}

Motivated by these empirical findings, we develop two algorithms with provable  ex-post guarantees that hold for a class of fairness rules that fulfill a necessary sensitivity condition (see Section \ref{ssec:criteria}). First, we propose a modification of \textsc{BP} 
(Section~\ref{sec:bp}) that includes dual variables for each group. These dual variables amplify employment probabilities at the group level to help direct the most valuable locations towards the groups that need them the most in order to meet their minimum requirements. We refer to this algorithm as \textsc{Amplified Bid Price Control} (\textsc{ABP}). \textsc{ABP} has strong performance guarantees at both the population level and, for sufficiently large groups, the group level  (Theorem~\ref{thm:abp-dis-dep}). At the population level, this theoretical guarantee is based on a comparison to an offline solution that meets the minimum requirement of every group. However, for groups with small sizes, such as many based on country of origin in the US, the algorithm has no meaningful performance guarantees and can exhibit unfair outcomes in our numerical studies (Section~\ref{sec:emp-algos}). Our second algorithm, dubbed \textsc{Conservative Bid Price Control} (\textsc{CBP}), combines elements from \textsc{ABP} with occasional \emph{greedy} decisions to overcome the poor performance of \textsc{ABP} with respect to small groups. Greedy steps sacrifice global efficiency to help boost the employment of groups that need it, and we sparingly apply these greedy steps to control the loss in global optimality. \textsc{CBP} obtains guarantees that not only apply at the population level but also ensure that \emph{all groups}, regardless of their size, will approximately meet their minimum requirements (Theorem~\ref{thm:cbp-dis-dep}).

\paragraph{Empirical analysis of trade-offs between fairness and global optimality.}
Finally, our work sheds light on the trade-offs that fairness considerations introduce to global optimization. To illustrate the spectrum of possible implications,  consider a simple example (Figure~\ref{fig:simple-example}) with three different \emph{worlds} that are distinguished by the impact of global optimization on two different groups. In all three worlds, the unique  globally optimal matching assigns all Group a refugees to Location 2 and all Group b refugees to Location 1. In the Antagonistic world (A) and the Benign world (B), this would violate any reasonable group fairness definition because all members of Group b are assigned to the location that is worst for them (Location 1) while all members of Group a are assigned to the location that is best for them (Location 2). Moreover, correcting group unfairness can affect global optimality to different degrees. In World (A), assigning more Group b refugees to the preferable Location 2 significantly reduces the global employment; in World (B), doing so only slightly impacts global employment, suggesting that the trade-off between global optimality and group fairness might be minimal, even though the globally optimal matching is unfair. Furthermore, in the Collaborative world (C) the globally optimal matching is itself group fair, which illustrates that the trade-off between group fairness and global optimality may be significant (World (A)), small (World (B)), or nonexistent (World (C)). In practice, due to the number of groups, locations, and heterogeneous arrivals, the trade-offs are naturally more complicated and of higher dimension. Additionally, such real-world instances allow for different worlds to co-exist in parallel on different subsets of groups and locations. 

\captionsetup[table] {name=Figure}
\begin{table}
    \begin{tabularx}{\linewidth}{@{}c X @{}}
    \includegraphics[width=0.45 \linewidth,valign=c]{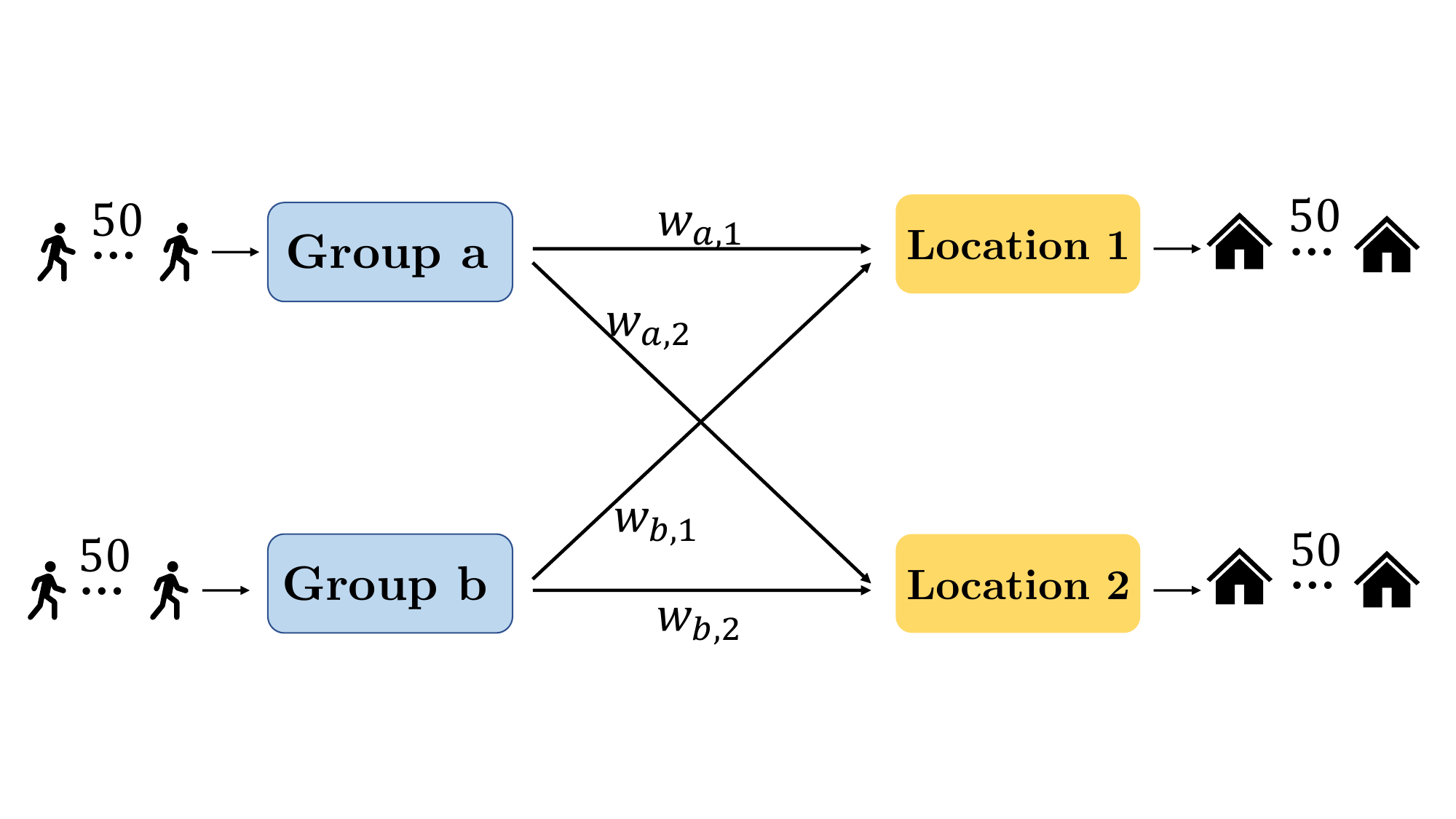}
    &   
\begin{tabular}{l l l l l}
\hline
                  & $w_{a,1}$ & $w_{a,2}$ & $w_{b,1}$ & $w_{b,2}$ \\ \hline
Antagonistic (A)  & 0.5       & 0.9       & 0.2       & 0.4       \\ \hline
Benign (B)        & 0.5       & 0.9       & 0.51      & 0.9       \\ \hline
Collaborative (C) & 0.5       & 0.9       & 0.9       & 0.5       \\ \hline
\end{tabular}
    \end{tabularx}
\caption{The three examples illustrate different trade-offs between global optimization and group fairness. In each case, there are two groups $\{a,b\}$ with $50$ refugees each and two locations $\{1,2\}$ with capacity to have $50$ refugees assigned to them. The variables $w_{a,1},w_{a,2},w_{b,1},w_{b,2}$ denote the employment probability for a group at one location.} \label{fig:simple-example}
\end{table}
\captionsetup[table] {name=Table}
\setcounter{figure}{1} 

Our empirical results (Section~\ref{sec:numerical}) suggest that real-world instances look more like a mix of the Benign and Collaborative worlds rather than an Antagonistic world. In particular, optimization can have a significant impact on the global efficiency (similar to what one would expect from World (C)); this is in line with the good empirical performance suggested by prior work on refugee assignment without fairness considerations \cite{ahani2021dynamic,elisabethpaper} . At the same time, while fairness-agnostic approaches yield severe fairness discrepancies, the cost in global optimization to account for fairness is relatively small (similar to what one would expect from World (B)); see Table \ref{tab:efficiency}. Indeed, in our data-driven numerical study, \textsc{ABP} incurs only a small loss in total employment when compared to \textsc{BP} (at most $2\%$) and the offline problem without fairness constraints (at most $4\%$; see Table~\ref{tab:efficiency} in Section~\ref{sec:emp-algos}), whereas \textsc{CBP} incurs only slightly more loss in total employment when compared to \textsc{ABP} ($\approx 0-1\%$). This suggests that, in practice, incorporating group fairness considerations need not be costly in terms of total employment.

\subsection{Related work}\label{ssec:rel_work}
% !TEX root = intro.tex
\paragraph{Refugee resettlement.} In the context of refugee resettlement, recent work introduces data analytics and market design to improve operational efficiency. Our work is closest to the aforementioned works \cite{bansak2018improving,ahani2021placement,elisabethpaper,ahani2021dynamic} that optimize the assignment of refugees to locations, either in static or in dynamic settings, where the objective function is the average employment rate. 
\cite{elisabethpaper,ahani2021dynamic} use resolving techniques, backtest on historical data, and find that their online algorithms incur little loss relative to the optimal objective in hindsight. \cite{elisabethpaper} incorporates an additional queuing/load balancing component in their objective to ensure that the resources at each location are evenly utilized throughout the year. The methods proposed by \cite{elisabethpaper,ahani2021dynamic} are currently being used in Switzerland and the US, respectively. \cite{golz19migration} studies a variation wherein the employment score at a location is a submodular function of the assigned refugees. In terms of fairness, \cite{andersson2018dynamic} considers envy-freeness between locations. To the best of our knowledge, our work is the first to design algorithms that promote fairness for refugees. Closely connected to our motivation is \cite{bansak2021algorithmic}, which identifies similar concerns to ours (Section~\ref{sec:emp-fairness-algos}) when studying a static refugee resettlement problem, but does so with a more restricted lens of group fairness and does not provide algorithmic solutions to overcome these concerns. 
Other works in refugee resettlement design mechanisms to satisfy preferences of refugees and locations \cite{andersson2016assigning,moraga2014tradable,nguyen2021stability,jones2018local,aziz2018stability}, or allow decision-makers to trade-off  between refugees’ preferences and employment maximization~\cite{acharya2022combining,delacretaz2019matching,olberg2022enabling}.

\paragraph{Methodological connections to online packing.} Our formulation and algorithms are similar to those studied in online packing problems, including canonical quantity-based network revenue management \cite{mcgill1999revenue} and the AdWord problem \cite{mehta2007adwords}. These are sequential decision-making problems, in which a stream of $T$ requests occurs, and a decision is made for each request: for  network revenue management the decision is to \emph{accept}/\emph{reject}, for AdWords it is \emph{where to accept} a request. In both problems accepting a request yields a reward (in AdWords, the reward depends on \emph{where}), but also consumes a set of resources. The objective is to maximize the accumulated rewards and is often measured by regret, defined by the difference in accumulated rewards compared to the offline optimal solution, while satisfying capacity constraints \cite{talluri2004theory}. Classical bid price control \cite{williamson1992airline} and randomized LP approaches \cite{talluri1999randomized} have long been known to achieve the fluid optimal $O(\sqrt{T})$ regret \cite{talluri1998analysis} for these problems, though these approaches have been improved through resolving techniques, that eventually led to $O(1)$ regret guarantees \cite{reiman2008asymptotically, jasin2012re,bumpensanti2020re,vera2021bayesian,vera2021online,freund2019good}. All of these approaches require resolving some LP from time to time. In this paper's context, which includes fairness constraints, using resolving can lead to infeasibility issues along the horizon (see Appendix~\ref{app:intro}). Moreover, our methods, which do not resolve, yield results close to the hindsight-optimal solution on real-world data sets (Section~\ref{sec:numerical}). Thus, we do not focus on expanding resolving methods to this setting, and leave this as a future area of work.

An orthogonal set of approaches to online packing problems arises in online convex optimization \cite{hazan2016introduction,agrawal2014fast}. These approaches do not rely on resolving, but cannot easily adapt to the presence of sub-linear size groups, as we find in our work. Applying previous results, from either stream, leads to vacuous guarantees for small groups (see discussion after Lemma~\ref{lem:bid-price-group}), forcing us to develop new algorithmic and analytical ideas to warrant good performance for all groups (Theorem~\ref{thm:cbp-dis-dep}).

\paragraph{Fairness considerations.} Beyond refugee resettlement, fairness has been studied in many other operational settings such as kidney exchange \cite{bertsimas2013fairness}, food banks \cite{sinclair2022sequential}, online advertising \cite{FeldmanHKMS10,balseiro2021regularized,BateniCCM22}, and pricing \cite{cohen2022price}; see also \cite{de2022algorithmic} for a wider review of applications of algorithmic fairness in operations. We now expand on the connection to works that are closer to ours. The trade-off between optimizing a global objective and incorporating fairness considerations is often captured in the literature via the \emph{price of fairness} \cite{bertsimas2011price,bertsimas2012efficiency}. Quantifying fairness considerations via a \emph{requirement} for each group resembles a setting from the literature where one aims to minimize wait time in a queuing system subject to a requirement on the idle time of each server \cite{armony2010fair}. A major difference to these works is that our setting involves dynamic decisions occurring over a finite number of rounds; the uncertainty in arrivals introduces complexities, especially with respect to groups with small arrival probability. In contrast, the former line of works assumes a static optimization while the latter focuses on a steady-state behavior of the system.

The fairness rules that we consider in our work are closely related to two dynamic decision-making lines of work from the literature. The Proportionally Optimized fairness rule (Example~\ref{ex:prop-opt-fairness-rule}) aims to mitigate the negative effect that the existence of a group may cause to members of another group by positing that the group should have performance at least as good as what it would have \emph{in isolation}. This was initially studied in a stochastic online learning setting with disjoint groups \cite{RaghavanSliWorWu18} and has been subsequently extended in a non-stationary setting with overlapping groups \cite{BlumLyk20}.\footnote{Since the groups can be overlapping, the latter work jointly optimizes the global objective as well as group objectives by positing that the whole population constitutes a group.} A key difference of our application is that individual cases are organically coupled via resource constraints; as a result, it is not clear how to define \emph{in isolation}. To incorporate this aspect, we extend the rule by asserting that the average employment probability of each group is no less than what it would have been if the group optimized its assignments over hypothetical capacities proportional to its size (without considering other groups). The Maxmin fairness rule (Example~\ref{ex:maxmin-fairness-rule}) aims to optimize the performance of the group with the lowest average employment outcome; this is a classical fairness criterion \cite{kalai1975other,rawls2004theory,kleinberg1999fairness,asadpour2007approximation,bertsimas2011price}. In a dynamic setting, there has been a sequence of recent works studying how to optimize this quantity \cite{lien2014sequential,manshadi2022fair, ma2020group,ma2021fairness}. A key difference in our setting is that we aim to optimize the global objective subject to a requirement on the lowest group's outcome,
rather than maximizing the latter quantity only with no regard to global efficiency.

Finally, fairness has been widely studied in the context of supervised machine learning with multiple different definitions mostly aiming to equalize a fairness metric across different groups starting from \cite{hardt2016equality, chouldechova2017prediction, kleinberg_inherent_2017, 
CorettGoel18}. The challenge in that line of literature is statistical: how can one try to use an inaccurate machine learning model? In contrast, we assume that the machine learning model is completely accurate (with respect to employment probabilities) and the challenge arises from the uncertainty in the arrival process.

\section{Model}\label{sec:model}
% !TEX root = main.tex
In this section, we introduce our framework to incorporate group fairness in the dynamic refugee assignment problem. We introduce the dynamic refugee assignment problem in Section~\ref{ssec:formal_model} and formalize the framework to incorporate group fairness into this problem in Section~\ref{ssec:model_objectives}. In Section~\ref{ssec:vanish-regret}, we discuss the criteria by which we evaluate algorithms in our framework. In Section~\ref{sec:fairness-rule}, we provide examples to illustrate the types of fairness constraints our framework encapsulates.

\subsection{The dynamic refugee assignment problem}\label{ssec:formal_model} 
We begin with a summary of the system parameters that fully describe the dynamic refugee assignment problem with group fairness considerations: the number of refugee arrivals $T$; a set of locations $\set{M}$ with $\hat{\bolds{s}}$ being the vector of fractional capacity of each location; a space of features of refugees $\Theta$ endowed with a feature distribution $\mathcal{P}$; a mapping $\bolds{w}(\cdot)$ that maps refugees' features to their employment probabilities which we also refer to as their \emph{scores}; a group mapping $g(\cdot)$ that maps refugees' features to their groups in $\set{G}$; and a selected fairness rule $\mathcal{F}$. Throughout, we assume perfect knowledge of $T$ and the vector $(\set{M},\hat{\bolds{s}}, \Theta, \mathcal{P},\bolds{w}(\cdot), g(\cdot),\set{G},  \mathcal{F})$. We next expand on each of these parameters.

In the dynamic refugee assignment problem, a refugee resettlement agency or decision maker (DM) is faced with the task of assigning $T$ refugee cases with labels $\{1,\ldots,T\}$ to a set $\Sloc$ of resettlement locations in a host country; we denote by~$M$ the number of locations. For simplicity, we assume that each refugee case consists of one individual. Each location $j\in\Sloc$ has a capacity~$s_j$ proportional to the number of arrivals, i.e., $s_j = \hat{s}_j T$ for a constant \emph{fractional capacity} $\hat{s}_j$.\footnote{\emph{E.g.}, in Switzerland, the government policy assigns cases
proportionally across regions \cite{bansak2018improving}.} In other words, $\hat{s}_j$ denotes the maximum fraction of the $T$ refugee cases that can be assigned to location~$j$, and  location~$j$ can have at most $s_j$ cases assigned to it. We assume there is sufficient capacity to assign each refugee case to a location,   $\sum_{j \in \Sloc} \hat{s}_j \geq 1$,  and we denote the minimum fractional capacity across the locations by $\hat{s}_{\min}=\min_{j\in\Sloc}\hat{s}_j$.

The sequence of refugee cases is denoted by the stochastic process $\boldsymbol{\omega} \equiv (\boldsymbol{\theta}_1,\ldots,\boldsymbol{\theta}_T )$, where $\boldsymbol{\theta}_t \in \Theta$ denotes the random feature vector associated with refugee case $t$. The feature vectors $\boldsymbol{\theta}_t$ are drawn identically and  independently from a known \emph{feature distribution} $\mathcal{P}$ over $\Theta$ and they include information that is known about the refugee such as their country of origin, gender,  education level, etc. When it is clear from context, we denote the set of all possible realizations of this stochastic process by $\Omega \equiv \Theta^T$. Our presentation in the main body of the paper focuses on a stationary feature distribution, which we extend to a non-stationary setting in Appendix~\ref{app:nonstationary}. 

The feature vector $\boldsymbol{\theta}_t$ lets us extract the following relevant information about refugee case $t$. First, we let $w_{t,j} \equiv w_j(\boldsymbol{\theta}_t) \in [0,1]$  denote the estimated probability that a refugee with features~$\boldsymbol{\theta}_t$ finds employment at each location $j \in \mathcal{M}$ within a specified time frame of interest.\footnote{In the US, employment after 90 days of arrivals is the only integration outcome that is systematically tracked and reported. Many other host countries, for example the Netherlands and Switzerland, also track employment after arrival as one of the key integration metrics, although the time-frame of interest varies by country.} The function $\bolds{w}: \Theta \to [0,1]^M$ is estimated by a supervised machine learning model and is available to the DM. We assume that the feature distribution $\set{P}$ is such that $\bolds{w}$ is drawn from a continuous cumulative distribution function over~$[0,1]^M$, and often refer to $w_{t,j}$ as the score of case $t$ at location $j$. Second, we let $g(t) \equiv g(\boldsymbol{\theta}_t) \in \mathcal{G}$ denote the (unique) group of case $t$, where the set of all groups~$\mathcal{G}$ has cardinality~$G$. Group definitions are outside the purview of the DM and are decided by an external policy maker;  for example, groups may be defined by level of education or by country of origin.\footnote{We emphasize the distinction between DM and policy maker to clarify what is in the purview of the algorithm and its designers. In particular, the definitions of groups and group fairness are exogenous to the algorithm.}  

Refugee cases arrive in ascending order of their labels and, at each period $t=1,\ldots, T$, the DM observes the feature vector $\boldsymbol{\theta}_t$ associated with refugee case $t$ and needs to decide a location for the case. This decision needs to be non-anticipative and irrevocable: the DM needs to commit on the location for refugee case $t$ prior to seeing the information of cases $t'>t$. 
Formally, we denote the DM's decision for refugee case $t$ by the assignment vector $\mathbf{z}_t  \equiv  \mathbf{z}_t(\boldsymbol{\theta}_1,\ldots,\boldsymbol{\theta}_t)\in \{0,1\}^M$, where  $\boldsymbol{\theta}_1,\ldots,\boldsymbol{\theta}_{t}$ is the information that has been revealed to the DM after the arrival of case $t$, and where the equality $z_{t,j} = 1$ holds if and only if refugee case $t$ is assigned to location $j \in \Sloc$. We say that a sequence of assignments is  \emph{feasible}  if and only if each case is assigned to a single location in a way that does not violate any location's capacity; that is, a sequence of assignments  $(\mathbf{z}_1,\ldots,{\mathbf{z}}_T) \in \{0,1\}^{T \times M}$ is feasible if and only if $(\mathbf{z}_1,\ldots,\mathbf{z}_T)$ is an element of $\mathcal{Z} = \tilde{\mathcal{Z}} \cap \{0,1\}^{T \times M}$, where  $\tilde{\mathcal{Z}}$ is the  \emph{fractional assignment polytope} defined as
\begin{align*}
    \tilde{\mathcal{Z}} = \left \{(\tilde{\mathbf{z}}_1,\ldots,\tilde{\mathbf{z}}_T) \in \mathbb{R}_{\geq 0}^{T \times M}: \quad 
 \begin{aligned}
    \sum_{j \in \mathcal{M}} \tilde{\mathbf{z}}_{t,j} = 1 && \forall t \in \{1,\ldots,T\} &&  \textrm{and} && 
    \sum_{t=1}^T \tilde{\mathbf{z}}_{t,j} \le s_j && \forall j \in \Sloc
    \end{aligned} \right \} .
\end{align*}

In the absence of group fairness considerations, the DM aims to select a feasible sequence of integer assignments that achieves high average employment probability across refugee cases. Formally, a sequence of assignments $({\mathbf{z}}_1,\ldots,{\mathbf{z}}_T) \in \mathcal{Z}$ induces a \emph{global objective value} of $\frac{1}{T}  \sum_{t=1}^T \sum_{j \in \Sloc} w_{t,j} z_{t,j}$, where we recall that $w_{t,j}$ and $z_{t,j}$ are shorthand for $w_j(\boldsymbol{\theta}_t)$ and $z_{t,j}(\boldsymbol{\theta}_1,\ldots,\boldsymbol{\theta}_t)$.

\begin{remark}[Discussion on assumptions]
    We now briefly comment on two modeling assumptions. First, our analysis assumes that the feature distribution $\mathcal{P}$ is known to the DM. In practice, the algorithm estimates  $\mathcal{P}$ based on the empirical feature distributions in past fiscal years. In Section \ref{sec:numerical} we discuss the empirical performance of our algorithm when it has access to a prior year's arrivals but not the current year's distribution $\mathcal{P}$. Second, we assume that each case consists of a single individual; this assumption is violated when cases correspond to families. Our results seamlessly extend to cases that consist of multiple refugees if all members of each case belong to the same group. As we discuss in Section~\ref{sec:conclusions}, handling intersectional groups is beyond the reach of our results and is an interesting open direction.
\end{remark}

\subsection{Framework for incorporating group fairness} \label{ssec:model_objectives}

Our goal is \emph{not} to impose a definition of group fairness; rather, it is to provide flexible and expansive frameworks that allow policy makers to specify the group fairness desiderata they wish to attain. With this in mind, our framework allows external policy makers to specify a \emph{minimum requirement} on the average employment probability for every group and every realization of the dynamic refugee assignment problem. Formally, the policy maker selects a fairness rule $\frule$ which provides a mapping $O_{g,\frule}: \Theta^T \to [0,1]$ for each group $g \in \mathcal{G}$, where the quantity  $O_{g,\frule}(\boldsymbol{\omega})$ represents the minimum requirement on the average employment probability of members of group $g$. 
Though we define fairness rules as arbitrary mappings from the set of possible realizations to minimum requirements,  Section~\ref{ssec:fairness_rule_examples} provides three concrete examples. These illustrate how common fairness notions can be succinctly specified as fairness rules. Furthermore, Section~\ref{ssec:criteria} describes the properties of fairness rules that allow for algorithms with strong algorithmic guarantees.

To state our framework more formally, we begin by defining the group-level performance of an assignment policy. We denote the set of cases that belong to group~$g$, among the first $t$ arrivals of realization $\boldsymbol{\omega}$, by $\mathcal{A}(g,t) \equiv\mathcal{A}(g,t,\bomega)  \triangleq \left \{ \tau \leq t: g(\boldsymbol{\theta}_{\tau}) = g \right \}$ and the cardinality of this set by $N(g,t) \equiv N(g,t,\bomega)\triangleq|\mathcal{A}(g,t,\bomega)|$. A group is called \emph{non-empty} for realization~$\bomega$ if $N(g,T,\bomega)>0$. This lets us define, for each {group} $g \in \mathcal{G}$, the average  employment probability of that group as\footnote{$\lor$ is the maximum of two numbers, i.e., $a\lor b=\max\{a,b\}$. For realizations $\omega$ where  group $g$ is empty, $\alpha_g(\omega)=0$.} 
$$\alpha_g(\bomega) \triangleq \frac{1}{N(g,T,\bomega) \lor 1} \sum_{t \in \mathcal{A}(g,T)}\sum_{j\in\Sloc}w_{t,j} z_{t,j}.$$
The minimum requirements serve as an algorithmic benchmark: for a realization $\bomega\in\Omega$, we want an algorithm to achieve an average employment probability for every group that is lower bounded by $O_{g,\mathcal{F}}(\bomega)$. This gives the definition of \emph{ex-post $g-$regret}: for a realization $\omega\in\Omega$, if an algorithm induces average employment probability $\alpha_{g}(\bomega)$ for a group $g\in\Sgro$, then its ex-post $g$-regret is
\[
\exR_{g,\frule}(\bomega) = O_{g,\frule}(\bomega) - \alpha_g(\bomega).
\]
The goal of the DM in our framework is to  maximize the average employment probability for all refugee cases independent of group, subject to the group-level minimum requirements. To evaluate how well an algorithm performs, we compare against the best  fractional assignments for the realization~$\bomega$. Formally,
\begin{align}\label{eq:outcome-benchmark}\tag{$\text{OFFLINE}_{\frule}$}
O^\star_{\frule}(\bomega)\triangleq \quad &\underset{\mathbf{\tilde{z}}\in \tilde{\mathcal{Z}}}{\textnormal{max}} \frac{1}{T} \sum_{t=1}^T\sum_{j \in \Sloc} w_{t,j}\tilde{z}_{t,j} \textnormal{ s.t. }\forall g:\sum_{t \in \mathcal{A}(g,T)}\sum_{j\in\Sloc}w_{t,j} \tilde{z}_{t,j} \ge N(g,T)\;O_{g,\mathcal{F}}(\boldsymbol{\omega}).
\end{align}
Accordingly, we define the \emph{global regret} of an algorithm by
\begin{align*}
\set{R}_{\frule}(\bomega) \triangleq  \mathbb{E}_{\bomega' \sim \mathcal{P}^T}[O_{\frule}^\star(\bomega')] - \frac{1}{T}  \sum_{t=1}^T \sum_{j \in \Sloc} w_{t,j} z_{t,j}.
\end{align*}
 To highlight the dependence on the number of arrivals and the feature distribution, we also use $\exR_{g,\frule}(T,\mathcal{P}, \bomega)$ and $\set{R}_{\frule}(T, \mathcal{P}, \bomega)$. For our regret to be well-defined, we require measurability of $O_{\frule}^\star$ and $O_{g,\frule}$ for  any group $g \in \Sgro$ (see Section~\ref{ssec:fairness_rule_examples}  for a further discussion).

Ideally, an algorithm should make assignment decisions that incur neither global nor ex-post $g$-regret; to combine these goals, our algorithms aim to minimize the \emph{maximum regret}, defined as:
\begin{equation}\label{def:max-regret}
\set{R}_{\frule}^{\max}(T,\mathcal{P},\bomega) \triangleq \max\left(\set{R}_{\frule}(T,\mathcal{P}, \bomega), \max_{g \in \set{G}} \exR_{g,\frule}(T, \mathcal{P},\bomega) \right),
\end{equation}
We use the maximum regret as a short-hand to capture the DM's goals: small global and small ex-post $g$-regret. An algorithm achieving these goals meets the policy maker's goal as it ensures that the fairness requirement is met for every group simultaneously and the global outcome is also efficient. When an algorithm fails to achieve these goals, the maximum regret alone does not reveal which part of the objective contributes to it, and thus our later numerical results focus on the more granular global and ex-post $g$-regret metrics.

\subsection{Distribution-dependent and distribution-independent vanishing regret}\label{ssec:vanish-regret}
In this section we formalize what we previously alluded to as \emph{small} or \emph{low} regret. Specifically, we define two vanishing regret properties to characterize that an algorithm incurs almost no maximum regret asymptotically with high probability: distribution-dependent and distribution-independent vanishing regret. We highlight that our context requires such a high-probability guarantee because arrivals occur just once per year; given the societal significance of refugee resettlement, an algorithm should perform well \emph{every} year rather than on average.

We first define the distribution-dependent vanishing regret property. Suppose that $\mathcal{C}$ is a given subset of all possible feature distributions over $\theta \in \Theta$ under which the scores $w_j(\btheta)$ have continuous distribution and each group $g$ has positive arrival probability. We say that an algorithm has \emph{distribution-dependent vanishing regret} for $\mathcal{C}$ and a sequence of fairness rules $\{\set{F}(T)\}_{T \geq 1}$ if 
\begin{equation}\label{eq:dis-dep-reg}
\forall \xi > 0,~\mathcal{P} \in \set{C},~\lim_{T \to \infty} \Pr_{\bomega \sim \mathcal{P}^T} \left\{\set{R}_{\frule(T)}^{\max}(T,\mathcal{P},\bomega) > \xi\right\} = 0.
\end{equation}
Here the sequence of fairness rules is such that for any $T$, $\set{F}(T)$ is a fairness rule defined over the sample space $\Theta^T$. The distribution-dependent vanishing regret property means that for any feature distribution $\mathcal{P} \in \set{C}$, the maximum regret is asymptotically non-positive with high probability as $T$ goes to infinity. Distribution-dependent vanishing regret allows for the regret of an algorithm to vanish at a rate that depends explicitly on group sizes (a property of the feature distribution).

When groups can be very small (e.g., when defined by country of origin; see Table~\ref{tab:group-stats}), one might desire stronger performance guarantees that are distribution-independent, i.e., regret that vanishes irrespective of the specific feature distribution  $\mathcal{P} \in \set{C}$. We say an algorithm has \emph{distribution-independent vanishing regret} for a fairness rule sequence $\{\set{F}(T)\}_{T \geq 1}$ if 
\begin{equation}\label{eq:dis-ind-reg}
\forall \xi > 0,~\lim_{T \to \infty} \sup_{\mathcal{P} \in \mathcal{C}}\Pr_{\bomega \sim \mathcal{P}^T} \left\{\set{R}_{\frule(T)}^{\max}(T,\mathcal{P},\bomega) > \xi\right\} = 0.
\end{equation}

\section{Examples and Properties for Fairness Rules}\label{sec:fairness-rule}

In this section we first provide three examples of interpretable fairness rules (Section \ref{ssec:fairness_rule_examples}) and then characterize natural properties of fairness rules (Section \ref{ssec:criteria}) that make them amenable to algorithms with vanishing regret guarantees. In particular, we include impossibility results to highlight that for some fairness rules/feature distributions, no algorithms can have distribution- dependent/independent vanishing regret. In providing both types of results, we characterize when vanishing regret guarantees of each type are achievable.

\subsection{Examples of Fairness Rules}\label{ssec:fairness_rule_examples}
In this section we introduce three examples of fairness rules: the \emph{Random}, the \emph{Proportionally Optimized}, and the \emph{MaxMin} fairness rules.

Our first example, the Random fairness rule, protects groups from being worse off than they would be in expectation under a random assignment. This reflects the status quo observed by \cite{bansak2018improving}, who suggest that in many host countries refugee resettlement assignments were effectively random \cite{bansak2018improving} before tools like GeoMatch were deployed (see Section \ref{sec:practical-implications}). Setting this benchmark as a fairness requirement ensures the ancient maxim of ``primum non nocere'' (``first, do no harm''), i.e., it ensures that the introduction of optimization is not to the detriment of any group.
\begin{example}[Random Fairness Rule]\label{ex:random-fairness-rule}
To mimic a random assignment, for each realization $\bomega\in\Omega$ and non-empty group $g \in \mathcal{G}$, we consider the fractional assignment $\tilde{z}_{t,j}^{\mathrm{random}} = \frac{s_j}{\sum_{j'\in \Sloc} s_{j'}}$ and set the minimum requirement as
$$O_{g,\mathrm{random}}(\bomega) = \frac{1}{N(g,T)}\sum_{t \in \Sarr(g,T)} \sum_{j\in\Sloc}w_{t,j} \tilde{z}_{t,j}^{\mathrm{random}}.$$
\end{example}

Our second example goes beyond ``primum non nocere'' to formalize the principle that a group should not be worse off due to the presence of other groups, which has been studied in sequential group fairness without resource constraints (see Section~\ref{ssec:rel_work}).  Specifically, in our resource-constrained context this requires that synergies between different arrivals within each group should be realized. Formally, the Proportionally Optimized fairness rule posits that, for every realization  $\bomega\in \Omega$, a non-empty group $g\in\mathcal{G}$ receives a capacity proportional to the group size for all locations $j\in\Sloc$, and each group then finds the optimal assignment using their allocated capacity.
\begin{example}[Proportionally Optimized Fairness Rule]\label{ex:prop-opt-fairness-rule}
Define the fractional assignment polytope restricted to non-empty group $g$ with allocated capacity $s_{j,g} = N(g,T)\hat{s}_j$ by $$\set{\tilde{Z}}_g = \{\bolds{\tilde{z}} \in [0,1]^{\Sarr(g,T) \times M} \colon \sum_{j\in\Sloc}\tilde{z}_{t,j} \leq 1,\forall t\in \Sarr(g,T);~\sum_{t \in \Sarr(g,T)} \tilde{z}_{t,j} \leq s_{j,g},\forall j \in \Sloc\}.$$ The resulting minimum requirement captures the maximum value a group would receive in isolation:
$$O_{g,\mathrm{pro}}(\bomega) = \max_{\bolds{\tilde{z}} \in \set{\tilde{Z}}_g} \frac{1}{N(g,T)} \sum_{t\in \Sarr(g,T)}\sum_{j \in \Sloc} w_{t,j}\tilde{z}_{t,j}.$$
\end{example}

Neither of the above examples guarantee that the most vulnerable (least employable) group be especially protected; this is captured by the well-studied MaxMin fairness principle, which we formalize as our third example of a fairness rule. \begin{example}[MaxMin Fairness Rule] \label{ex:maxmin-fairness-rule}
In the MaxMin fairness rule, the corresponding minimum requirement for realization $\omega\in\Omega$ is the same across all non-empty groups $g\in\mathcal{G}$ and maximizes the minimum average value across all groups:
$$O_{g,\mathrm{maxmin}}(\bomega)=\max_{\bolds{z} \in \tilde{\set{Z}}}\min_{g' \in \Sgro\colon N(g',T)>0}\frac{1}{N(g',T)}\sum_{t \in \Sarr(g',T)} \sum_{j \in \mathcal{M}} w_{t,j}z_{t,j}.$$ 
\end{example}
All three examples naturally induce a sequence of fairness rules when applied to differing numbers of arrivals $T \geq 1$. In Appendix~\ref{app:measurable}, we show that they all satisfy the measurability assumption in our model. Abusing notation, we refer to the fairness rules as the induced fairness rule sequences for different $T$ when we say an algorithm has vanishing regret for a fairness rule.

\subsection{Properties of Fairness Rules}\label{ssec:criteria}
In this section we characterize two natural properties of fairness rules that jointly enable vanishing regret guarantees: ex-post feasibility and low sensitivity.

\noindent\textbf{Ex-post feasibility.} The first property, ex-post feasibility, holds when the fairness rule's minimum requirement is feasible in hindsight for the (relaxed) assignment problem. 
This is necessary for a fairness rule to be meaningful: without it, even with full information the DM would not be able to achieve the minimum requirements without violating the capacity constraints. All examples in Section~\ref{ssec:fairness_rule_examples} satisfy ex-post feasibility (for the first two, we need the definition to allow for fractional assignments). 
\begin{definition}\label{def:fairness_feasible}
    A fairness rule $\mathcal{F}$ is \emph{ex-post feasible} if, for every sample path $\boldsymbol{\omega} \in \Omega$, there exist fractional ex-post decisions $\tilde{z}_{t,j}(\boldsymbol{\omega}) \in [0,1]$ for each case $t\in[T]$ and location $j \in \Sloc$ such that 
\begin{align*}
    \begin{aligned}    (\tilde{\textbf{z}}_1,\ldots, \tilde{\textbf{z}}_T)\in \tilde{Z} && \text{and} && \frac{1}{N(g,T)\lor 1} \sum_{t \in \mathcal{A}(g,T)}\sum_{j\in\Sloc}w_{t,j} \tilde{z}_{t,j}&\ge O_{g,\mathcal{F}}(\boldsymbol{\omega}) &&  \forall g \in \mathcal{G} .
   \end{aligned} 
\end{align*}
Note that the constraints imply $O_{g,\mathcal{F}}(\bomega)=0$ for an empty group $g$, which we assume throughout. 
\end{definition}
 
\noindent\textbf{Low sensitivity.} The second property, sensitivity, captures how stable the minimum requirement of a group is among similar sample paths. Formally, letting $Q_{g,\frule}(\bomega) = N(g,T,\bomega)O_{g,\frule}(\bomega)$ 
be the minimum total requirement for group $g$ in a sample path $\bomega$ and $p_g(\set{P})$ be the arrival probability of that group under feature distribution $\set{P}$, we define the sensitivity of a fairness rule as follows. 
\begin{definition}\label{def:irr-rule}
A fairness rule $\frule$ is $(\chi,\delta)$-sensitive for a feature distribution $\set{P}$, if there exists an event $\set{B} \subseteq \Omega$ with $\Pr\{\set{B}\} \geq 1 - \delta / T$, such that for any pair of $\bomega, \tilde{\bomega} \in \set{B}$ that differ in at most one arrival $t$ with feature $\btheta_t$ and 
$\btheta'_{t}$ respectively, we have 
\[
|Q_{g,\frule}(\bomega) - Q_{g,\frule}(\tilde{\bomega})| \leq \left\{
\begin{aligned}
    \sqrt{p_g(\set{P})}\chi, & \text{ if }g(\btheta_t) \neq g\text{ and } g(\btheta'_t) \neq g \\
    \chi, & \text{ otherwise.}
\end{aligned}\right.
\]
\end{definition}
For the sensitivity  of a fairness rule to be low, we require that the total minimum requirement for a group $g$ should \emph{seldom} differ too much if only one refugee has a changed feature vector. Moreover, if this refugee is not in this group $g$ (either before or after the change), then the impact on group $g$ scales at most  with its group size (reflected by the $\sqrt{p_g(\set{P})}$ term). From a technical perspective, this sensitivity condition will allow us to establish a law of large number type of result for the minimum requirement (see Lemma~\ref{lem:bpc-og-conc}). We refer to the maximum allowed difference, $\chi$, as the \emph{sensitivity} of a fairness rule.

The following result shows that both the Random fairness rule (Example~\ref{ex:random-fairness-rule}) and the Proportionally Optimized fairness rule (Example~\ref{ex:prop-opt-fairness-rule}) have constant sensitivity $\chi$ (proof in Appendix~\ref{app:irr-rules}).
\begin{proposition}\label{prop:irr-rules}
Random and Proportionally Optimized are $(1,\delta)$ and $(2,\delta)$-sensitive for $\delta \geq 0$.
\end{proposition}
The above result implies that the sensitivity of Random and Proportionally Optimized is independent of the feature distribution almost surely. However, this is not the case for the MaxMin fairness rule. Denoting the minimum group arrival probability $\min_{g \in \Sgro} p_g(\bolds{P})$ by $p_{\min}(\bolds{P})$, we show the following bound on the sensitivity of MaxMin (proof in Appendix~\ref{app:irr-maxmin}). 
\begin{proposition}\label{prop:irr-maxmin}
MaxMin is $(25p_{\min}^{-1}(\bolds{P}),\delta)$-sensitive for any $\delta \geq GTe^{4-p_{\min}(\bolds{P})T/8}.$
\end{proposition}
It is intuitive that the sensitivity of MaxMin would depend on $p^{-1}_{\min}(\bolds{P})$: suppose, for large $T$,~$\bolds{P}$ is such that with constant probability a particular group observes no arrivals at all. If that group is the most vulnerable, then adding just one arrival of that group would have a large effect on the minimum requirements of the remaining groups (this reasoning is formalized in Appendix \ref{app:max-min-imp}).

In Section~\ref{sec:bp}, we show that ex-post feasibility and low sensitivity enable algorithms with distribution-dependent vanishing regret. In Section~\ref{sec:cons} we extend these results to design an algorithm with the stronger distribution-independent vanishing regret guarantee.

Finally, we show the necessity of the sensitivity property to achieve distribution-dependent (and thus independent) vanishing regret: in  Appendix~\ref{app:hard-fairness-approx-target}, we show that no algorithm can have distribution-dependent vanishing regret when we allow for an arbitrary ex-post feasible fairness rule. In Appendix \ref{app:max-min-imp}, we show that no algorithm can have distribution-independent vanishing regret for the MaxMin fairness rule, thereby showing the necessity of constant sensitivity for  distribution-independent guarantees (recall from Proposition~\ref{prop:irr-maxmin} that the sensitivity of MaxMin depends on the inverse of the minimum group size). Nevertheless, our algorithm in Section~\ref{sec:bp} has instance-dependent vanishing regret for MaxMin because it has the sensitivity property asymptotically for a fixed distribution; see the discussion after Theorem~\ref{thm:abp-dis-dep}.

\section{Algorithm with Distribution-Dependent Vanishing Regret}\label{sec:bp}
% !TEX root = main.tex
This section develops an algorithm with distribution-dependent vanishing regret when the fairness rule is ex-post feasible and has low sensitivity. Our approach builds on the classical \textsc{Bid Price Control} that minimizes global regret (without fairness considerations) and was initially developed in the context of network revenue management \cite{talluri1998analysis}, which is an admission control problem with only accept/reject decisions. We show that a simple modification of this algorithm enjoys distribution-dependent vanishing regret with bounds on global regret of order $1/\sqrt{T}$ and on ex-post $g$-regret of order $1/\sqrt{p_g(\mathcal{P})\cdot T}$ where $p_g(\mathcal{P})$ denotes the probability of a group $g$ refugee under feature distribution $\mathcal{P}$. In the next section, we extend our algorithm to attain {ex-post} $g$-regret that does not depend on $p_g(\mathcal{P})$, and thereby guarantee distribution-independent vanishing regret.

\subsection{Amplified Bid Price Control}

\textsc{Bid Price Control} (\textsc{BP}) aims to maximize the global objective subject to capacity constraints, i.e., find an assignment $\bolds{z} \in \set{Z}$ that maximizes $\sum_{t=1}^T \sum_{j \in \Sloc} w_{t,j}z_{t,j}$.
Myopically assigning each case to the location that maximizes its score may seem like a natural candidate algorithm to achieve this objective; yet, this quickly depletes the capacities at locations with universally good refugee employment prospects. As a result, optimally assigning a case needs to take into account not only  the \emph{present} case-location score but also the \emph{opportunity cost} of a slot at location $j$, i.e., the potential decrease in future assignment scores due to depleted capacity at $j$. 
\textsc{Bid Price Control} quantifies the opportunity cost of each assignment and additively adjusts the employment probabilities by this amount. Effectively, this penalizes the selection of locations that are valuable in the future. That said, this adjustment of scores does not take into account group-fairness concerns and may severely violate minimum requirements (see Figure~\ref{fig:fairness-results} of Section~\ref{sec:numerical}).

Our modification (Algorithm \ref{algo:bpc}), which we term \textsc{Amplified Bid Price Control} (\textsc{ABP}), explicitly captures these fairness concerns by introducing group-dependent score amplifiers. In particular, the scores of each case are first scaled by their group's amplifier  and are then additively adjusted by the opportunity cost of each potential assignment. These amplifiers help direct resources towards groups that need them the most in order to satisfy their minimum requirements. For example, a group that is likely to have a loose fairness constraint in \ref{eq:outcome-benchmark} would have a small amplifier and be assigned based on the tradeoff between $w_{t,j}$ and the locations' opportunity cost. To formalize the above intuition, we denote by $\mu^\star_j$ and $\lambda^\star_g$ the opportunity cost of location $j\in \set{M}$ and the amplifier of group $g\in \Sgro$,  respectively. These quantities are initially instantiated 
(see line~\ref{algoline:set-lagra-multi} in Algorithm~\ref{algo:bpc}); we elaborate on this instantiation below. 

\textsc{ABP} wants to assign
case $t$ to the location $J^{\BP}(t)\triangleq \arg\max_{j \in \Sloc} (1+\lambda_{g(t)}^{\star})w_{t,j} - \mu_j^{\star}$ with the highest adjusted score, breaking ties arbitrarily.\footnote{Note, however, that under our assumption of continuous scores, tie-breaking almost surely does not occur.} If  $J^{\BP}(t)$ has no capacity, the final assignment $J^{\ABP}(t)$ is an arbitrary location with remaining capacity. For our numerical results, we select the location with the highest adjusted score among locations with remaining capacity (see Appendix~\ref{app:emp-setup}).

\begin{algorithm}[H]
\LinesNumbered
\DontPrintSemicolon
\caption{\textsc{Amplified Bid Price Control}
\label{algo:bpc}%Online algorithm of queue $i$
}
\SetKwInOut{Input}{input}\SetKwInOut{Output}{output}
\Input{Capacity $s_j$ for $j \in \Sloc$; Time horizon $T$; Fairness rule $\frule$ inducing scores $O_g(\omega)$} 
 Set $(\bolds{\mu}^\star,\bolds{\lambda}^\star)\gets \arg \min_{\bolds{\mu} \in \mathbb{R}_+^M,\bolds{\lambda} \in \mathbb{R}_+^G} \expect{L(\bolds{\mu},\bolds{\lambda})}$ where $L(\bolds{\mu},\bolds{\lambda})$ is defined in \eqref{eq:lagrangian} \label{algoline:set-lagra-multi}\;
\For{case $t = 1, \ldots,T$ of group $g(t)\in\Sgro$ with scores $w_{t,j}$ for $j\in\Sloc$}{
$J^{\ABP}(t)\gets J^{\BP}(t)$ where
$J^{\BP}(t) \gets \arg\max_{j \in \Sloc} (1+\lambda_{g(t)}^{\star})w_{t,j} - \mu_j^{\star}$ \label{algoline:bpc-select}\;
\lIf{\text{$J^{\BP}(t)$ has no capacity}} {\label{algoline:bpc-no-cap}
     $J^{\ABP}(t)\gets$ a location with remaining capacity}
}
\end{algorithm}

We now expand on the instantiation of $\boldsymbol{\mu}^\star,\boldsymbol{\lambda}^\star$. The problem without fairness constraints aims to find an assignment ${\boldsymbol{z}}\in{\mathcal{Z}}$ that maximizes $\sum_{t \in [T]}\sum_{j \in \Sloc} w_{t,j}z_{t,j}$. The classical \textsc{Bid Price Control} computes opportunity costs $\boldsymbol{\mu}^\star$  by solving its Lagrangian relaxation over the fractional assignments without capacity constraints, i.e., the polytope $\widehat{\mathcal{Z}} \triangleq \{\widehat{\boldsymbol{z}}\in[0,1]^{T\times M} \colon \sum_{j \in \Sloc} \widehat{z}_{t,j} = 1 \forall t\}$:
$$L(\boldsymbol{\mu}) = \frac{1}{T}\max_{\widehat{\boldsymbol{z}}  \in \widehat{\mathcal{Z}}}
\left(\sum_{t=1}^T\sum_{j \in \Sloc} w_{t,j}\widehat{z}_{t,j} + \sum_{j \in \Sloc} \mu_j\left(s_j - \sum_{t \in [T]} \widehat{z}_{t,j}\right)\right),$$ where $\mu_j \geq 0$ is the Lagrange multiplier associated with the capacity constraint for location $j$. For any vector $\boldsymbol{\mu}\geq 0$, $L(\boldsymbol{\mu})$ upper bounds the original objective.\footnote{This is because every $\tilde{\boldsymbol{z}} \in \tilde{\set{Z}}$  is also an element of $\widehat{\mathcal{Z}}$ as $\widehat{\mathcal{Z}}$ relaxes the capacity constraints; hence, when evaluating $\tilde{\boldsymbol{z}} \in \tilde{\set{Z}}$ for the optimization problem in $L(\boldsymbol{\mu})$, this only adds a non-negative component to the original objective.} \textsc{BP} selects Lagrange multipliers $\boldsymbol{\mu}^\star \in \arg\min_{\boldsymbol{\mu} \in \mathbb{R}_+^{M}} \expect{L(\boldsymbol{\mu})}$ that give the tightest upper bound in expectation. For each case $t$, the optimal $\widehat{\boldsymbol{z}}$ in the inner maximization of $L(\boldsymbol{\mu}^\star)$ chooses the location $j$ that maximizes $w_{t,j} - \mu^\star_j$. This matches the intuition that the multiplier $\mu^\star_j$ is the opportunity cost associated with location $j$. 

We follow a similar approach while also incorporating group fairness constraints. In particular, denoting the Lagrange multiplier associated with the constraint for group $g$ by $\lambda_g$, the Lagrangian relaxation $L(\bolds{\mu},\bolds{\lambda})$ of our offline maximization program \eqref{eq:outcome-benchmark} is written as:
\[
\max_{\widehat{\bolds{z}} \in \widehat{\set{Z}}} \frac{1}{T}\left(\sum_{t = 1}^T\sum_{j \in \Sloc} w_{t,j}\widehat{z}_{t,j} + \sum_{j\in \Sloc} \mu_j\left(s_j - \sum_{t=1}^T \widehat{z}_{t,j}\right)+ \sum_{g \in \Sgro} \lambda_g\left(\sum_{t \in \Sarr(g,T)} \sum_{j \in \Sloc} w_{t,j}\widehat{z}_{t,j} - O_g N(g,T)\right)\right).
\]
Rearranging terms, we find that $\widehat{z}_{t,j}$ appears as a coefficient for $(w_{t,j}-\mu_j+\lambda_{g(t)}w_{t,j})=(1+\lambda_{g(t)})w_{t,j}-\mu_j$. Similar to \textsc{BP}, the maximizing assignment $\widehat{\boldsymbol{z}}\in\widehat{\mathcal{Z}}$ now sets, for each case $t$, $\widehat{z}_{t,j}=1$ for the location~$j$ that maximizes $(1+\lambda_{g(t)})w_{t,j}-\mu_j$. Hence, the Lagrangian relaxation is
\begin{equation}\label{eq:lagrangian}\tag{$\text{LAGR}$}
\begin{aligned}
L(\boldsymbol{\mu},\boldsymbol{\lambda})=\frac{1}{T}\left(\sum_{t=1}^T \max_{j \in \Sloc} ((1+\lambda_{g(t)})w_{t,j} - \mu_j) + \sum_{j \in \Sloc} \mu_j s_j -\sum_{g \in \Sgro} \lambda_g O_gN(g,T)\right).
\end{aligned}
\end{equation}
\textsc{ABP} picks opportunity costs $\bolds{\mu}^\star$ and amplifiers $\bolds{\lambda}^\star$ that minimize $\mathbb{E}_{\bomega}[L(\bolds{\mu},\bolds{\lambda})]$ subject to non-negativity. As alluded to above, the scores of an arrival are amplified by $\lambda_g$ before they are additively adjusted by opportunity cost $\mu_j$. Further, by the same reasoning as for $\textsc{BP}$, $L(\boldsymbol{\mu},\boldsymbol{\lambda})\geq O^\star$ for any $\boldsymbol{\mu},\boldsymbol{\lambda}\geq0$. In particular, $L(\boldsymbol{\mu}^\star,\boldsymbol{\lambda}^\star)\geq O^\star$ and thus $\mathbb{E}_{\bomega}[L(\bolds{\mu}^\star,\bolds{\lambda}^\star)]$ is an upper bound for $\mathbb{E}_{\bomega}[O^\star]$. With $\mathbb{E}_{\bomega}[L(\boldsymbol{\mu},\boldsymbol{\lambda})]$
being an expectation over convex functions (the maximum across linear functions), it is a convex function of $\bolds{\mu},\bolds{\lambda}$ enabling the efficient instantiation of $\bolds{\mu}^\star,\bolds{\lambda}^\star$ (line 1 of Algorithm~\ref{algo:bpc}). 

\subsection{Performance guarantee}\label{ssec:bp-thm}
In this section we establish that $\ABP$ enjoys distribution-dependent vanishing regret:

\begin{theorem}\label{thm:abp-dis-dep}
If a fairness rule sequence $\{\set{F}(T)\}$ satisfies that for any distribution $\set{P} \in \set{C}$, each rule $\set{F}(T)$ is {(1)} ex-post feasible; {and (2)} $(\chi(\set{P}),\delta(T))$-sensitive with a constant $\chi(\set{P}) \geq 1$ such that $\lim_{T \to \infty} \delta(T) = 0$ and $\lim_{T \to \infty} \ln(\delta(T))/T = 0$, then $\textsc{ABP}$ has distribution-dependent vanishing regret for {$\{\set{F}(T)\}$}.
\end{theorem}

By setting $\delta(T) = 1/T$, when combined with Proposition~\ref{prop:irr-rules}, Theorem~\ref{thm:abp-dis-dep} implies that \textsc{ABP} has distribution-dependent vanishing regret for the Random and the Proportionally Optimized fairness rules. Moreover, by taking $\delta(T) = \max(1/T, GTe^{4-p_{\min}(\set{P})T/8}), $ Proposition~\ref{prop:irr-maxmin} shows that the MaxMin fairness rule has a constant sensitivity $25p^{-1}_{\min}(\set{P})$ for any feature distribution $\set{P}.$ Since $\lim_{T \to \infty} \delta(T) = 0$ and $\lim_{T \to \infty} ln(\delta(T))/T = 0$ for a fixed $\set{P},$ Theorem~\ref{thm:abp-dis-dep} shows that \textsc{ABP} also has distribution-dependent vanishing regret for the MaxMin fairness rule.

The proof of Theorem~\ref{thm:abp-dis-dep} (Appendix~\ref{app:abp-dis-dep})  relies on the following two key lemmas, which bound the global regret and ex-post $g-$regret for $\ABP$ respectively. In it, we show formally that the two bounds imply that $\textsc{ABP}$ fulfills \eqref{eq:dis-dep-reg}, i.e., that it has distribution-dependent vanishing regret.

\begin{lemma}\label{lem:bid-price-global}
Fix an ex-post feasible fairness rule $\frule$. For any $\delta>0$,
\textsc{Amplified Bid Price Control} has global regret $\set{R}_{\set{F}}^{\ABP} \leq 
\sqrt{\frac{\ln(M/\delta)}{T}}\cdot\left(\frac{1}{\hat{s}_{\min}}+1\right)$ with probability at least~$1 - \delta$.
\end{lemma}

\begin{lemma}\label{lem:bid-price-group} 
For any $\delta > 0$, fix an ex-post feasible fairness rule $\frule$ that is $(\chi,\delta)$-sensitive with $\chi \geq 1$. Then for any $g \in \Sgro$, $\textsc{ABP}$ has ex-post $g$-regret $\set{R}_{g,\frule}^{\mathrm{ex},\textsc{ABP}} \leq \sqrt{\frac{\ln (e^2M/\delta)}{p_g T}}\left(\frac{1}{\hat{s}_{\min}} + {31\chi}\right) + 2\chi \delta$ with probability at least $1-3\delta$.
\end{lemma}
Several observations about the practicality of these bounds are in order. First, the capacity of each location typically scales with the number of cases $T$; which means that $\hat{s}_{\min}$ can be treated as a constant with respect to $T$. 
Therefore, the upper bound on global regret in Lemma~\ref{lem:bid-price-global} vanishes at a rate of $\tilde{O}(1/\sqrt{T})$ as~$T$ grows large with high probability, \emph{e.g.}, when setting $\delta=1/T$, and this holds regardless of the sensitivity of the fairness rule.
Second, when fixing the feature distribution $\mathcal{P}$, the upper bound on {ex-post} $g$-regret in Lemma~\ref{lem:bid-price-group} vanishes at a rate of $\tilde{O}(1/\sqrt{T})$ as $T$ increases. 

We show in Appendix~\ref{app:abp-lowerbound} that the convergence rate of $\tilde{O}(1/\sqrt{T})$ is unimprovable for both global regret and {ex-post} $g-$regret\footnote{Some recent literature focuses on when stronger convergence rates cannot be obtained \cite{besbes2022multi,jiang2022degeneracy}, with a focus on particular types of indiscrete distributions; in contrast, our lower bound reflects ones in \cite{freund2019good} with particular groups having $o(T)$ arrivals with high probability.} and thus the regret of $\textsc{Amplified Bid Price Control}$ vanishes at a near-optimal rate. However, this property is distribution-dependent and relies on a scaling regime where group sizes scale linearly with the number of arrivals. Such a regime might not be an appropriate model in practice where we frequently observe ``small'' groups.\footnote{\emph{E.g.}, in the U.S., out of 1,175 cases, 12 out of 26 groups defined by country of origin have a handful of cases} To address this limitation we derive distribution-independent guarantees in Section~\ref{sec:cons}.

\noindent\textbf{Proof Sketch of Lemmas~\ref{lem:bid-price-global} and \ref{lem:bid-price-group}.} Our starting point towards proving the lemmas is to analyze the global and group objectives when case $t$ is always assigned to location  $J^{\BP}(t)$ irrespective of capacities. The objectives resulting from such an assignment have strong guarantees as we show in Lemma~\ref{lem:bp-property-obj} and~\ref{lem:bp-property-fair} using KKT conditions on $\bolds{\mu}^\star,\bolds{\lambda}^\star$ and concentration arguments (see Appendix~\ref{app:bp-property} for full proofs). We drop the dependence on the feature distribution $\set{P}$ in $p_g(\set{P})$ when clear from context.
\begin{lemma}\label{lem:bp-property-obj}
We have $\sum_{t \in [T]}\expect{w_{t,J^{\BP}(t)}} \geq T\expect{O^\star}$.
\end{lemma}

\begin{lemma}\label{lem:bp-property-fair}
For any $t$ and $g$: $\expect{w_{t,J^{\BP}(t)}\mid g(t)=g} \geq \expect{O_g} - \sqrt{\frac{1}{p_g T}}$. 
\end{lemma}
Of course, ignoring capacities is unlikely to yield a feasible solution since we have hard capacity constraints. However, until some location runs out of capacity, \textsc{ABP} follows exactly these assignments. As a result, we want to bound the first time when a location's capacity is depleted; we denote this time by $\Temp$. The following lemma shows that, with high probability, this event occurs in the last $\Delta^{\BP}:=
\frac{\sqrt{\frac{1}{2}T\ln\left(\frac{M+2}{\delta}\right)}}{\hat{s}_{\min}}$  periods. The proof uses that, in expectation, $J^{\BP}$ assigns at most~$s_j$ cases to each location $j$. Hence, by concentration arguments, \textsc{ABP} uses at most $\hat{s}_j t + O(\sqrt{t\ln(\nicefrac{1}{\delta})})$ capacity in a location $j$ up to case $t$, and the algorithm does not deplete capacity in any location until assigning about $T-O(\sqrt{T\ln(\nicefrac{1}{\delta})})$ cases (see Appendix~\ref{app:bpc-temp} for a full proof).

\begin{lemma}\label{lem:bpc-temp}
For any $\delta>0$ with probability at least $1-\frac{M\delta}{M+2}$: $\Temp \geq T -\Delta^{\BP}$. 
\end{lemma}

The above results bound the global regret and bound the average score of  group $g$, under \textsc{ABP}, in terms of the expected minimum requirement $\expect{O_g}$. To obtain a guarantee on ex-post $g-$regret, we further need the following lemma to show that for fairness rules with low sensitivity the sample-path minimum requirement $O_g(\bomega)$ concentrates near its expectation $\expect{O_g}$. The proof combines the sensitivity definition with a novel extended version of McDiarmid's Inequality, and is given in Appendix~\ref{app:bpc-og-conc}. In contrast to more traditional bounds that do not rely on sensitivity, the dependence of this bound on $p_g$ is of order $1/\sqrt{p_g}$ rather than $1/p_g$. This is particularly important for our analysis in the next section, wherein $p_g$ may be of order $o(\sqrt{T})$, and as a result $1/p_g\sqrt{T}$ would give vacuous bounds.
\begin{lemma}\label{lem:bpc-og-conc}
For any $\delta > 0$, if $\frule$ is $(\chi,\delta)-$sensitive with $\chi \geq 1$, then for group $g \in \Sgro$, with probability at least $1 - 2\delta$, we have $O_{g,\frule}(\bomega) \leq \expectsub{\bomega'}{O_{g,\frule}(\bomega')}+16\chi\sqrt{\frac{\ln(e^2/\delta)}{p_g T}}+2\chi \delta$. 
\end{lemma}
We prove Lemma \ref{lem:bid-price-global} in Appendix~\ref{app:thm_bp_proof_global} by combining Lemma~\ref{lem:bp-property-obj} and Lemma~\ref{lem:bpc-temp} and Lemma~\ref{lem:bid-price-group} in Appendix~\ref{app:thm_bp_proof_group} by combining Lemmas~\ref{lem:bp-property-fair}, \ref{lem:bpc-temp} and \ref{lem:bpc-og-conc}.

\section{Algorithm with Distribution-Independent Vanishing Regret}\label{sec:cons}
% !TEX root = main.tex
In this section, we design an algorithm, which we term \textsc{Conservative Bid Price Control} (CBP), with a distribution-independent vanishing regret guarantee. This contrasts with our guarantee for \textsc{ABP} (Theorem~\ref{thm:abp-dis-dep}), which is distribution-dependent because our ex-post $g-$regret bound of $\ABP$ depends on a group's size (Lemma~\ref{lem:bid-price-group}).

The dependence on group size occurs due to two issues: First, \textsc{ABP} is designed to guarantee that the expected \emph{total} scores of a group is no less than the total minimum requirement of that group. However, ex-post $g$-regret measures the difference between the \emph{average} score of a group and its minimum requirement. This mismatch creates a dependence on a group's expected size (see Lemma~\ref{lem:bp-property-fair}). Second, ex-post $g$-regret aims for a sample-path group fairness guarantee while \textsc{ABP} only considers group fairness in expectation. Although Lemma~\ref{lem:bpc-og-conc} shows that the sample-path minimum requirement of a group concentrates near its expectation, the rate of this concentration depends on the group size. In this section we design a second algorithm that explicitly aims to avoid these dependencies. 

\subsection{Conservative Bid Price Control}
\textsc{CBP} (Algorithm~\ref{algo:conservative}) modifies \textsc{ABP} in two ways. First, \textsc{CBP} proactively 
offers greedy assignments to groups which are predicted to not meet their requirement otherwise. 
Second, \textsc{CBP} maintains a conservative upper-confidence-bound estimate of the sample-path minimum requirement using Lemma~\ref{lem:bpc-og-conc}. With these modifications, \textsc{CBP} approximately maintains the global guarantee of \textsc{ABP} (Lemma~\ref{lem:cons-global}) while also obtaining group-size independent guarantees for ex-post $g-$regret (Lemma~\ref{lem:cons-group}). These results then imply distribution-independent vanishing regret in Theorem~\ref{thm:cbp-dis-dep}.

We next explain how \textsc{CBP} predicts whether an arrival $t$ should receive a greedy assignment $J^{\GR}(t)$ instead of an \textsc{ABP} assignment $J^{\BP}(t)$ in Line~\ref{algoline:bpc-select} of Algorithm~\ref{algo:bpc}. Suppose for a group $g$, a clairvoyant knew its future arrivals and could assign each arrival to its score-maximizing (greedy) location $J^{\GR}(t)$.  If assigning~arrival $t$ to~$J^{\BP}(t)$ and all subsequent arrivals of group $g(t)$ greedily is insufficient for this group to meet its requirement, then our clairvoyant assigns arrival $t$ to $J^{\GR}(t)$. In contrast, if it is sufficient, our clairvoyant assigns arrival $t$ to $J^{\BP}(t)$. Formally, let $J^{\alg}(\tau)$ denote the assignment that \textsc{CBP} made at time $\tau$. Then, our clairvoyant condition is
\[
\sum_{\tau\in \Sarr(g,t-1)}w_{\tau,J^{\alg}(\tau)} + w_{t,J^{\BP}(t)}+ \sum_{\tau \in \Sarr(g,T) \setminus \Sarr(g,t)}w_{\tau,J^{\GR}(\tau)} \geq N(g,T)O_g(\bomega).
\]
\noindent The terms on the left-hand side denote (i) the scores of already-assigned arrivals, (ii) the score of the current arrival $t$ under $J^{\BP}$, and (iii) the cumulative scores of assigning the remaining arrivals of group $g$ greedily. Since there are exactly $N(g,T)$ summands, the condition is equivalent to  
\begin{equation}\label{eq:clairvyoant}
\sum_{\tau\in \Sarr(g,t-1)}\left(w_{\tau,J^{\alg}(\tau)}-O_g(\bomega) \right)+ w_{t,J^{\BP}(t)}-O_g(\bomega)+ \sum_{\tau \in \Sarr(g,T) \setminus \Sarr(g,t)}\left(w_{\tau,J^{\GR}(\tau)}-O_g(\bomega)\right) \geq 0.
\end{equation}
As we are not a clairvoyant, we cannot implement the above condition for two reasons:  (1) we do not know the path-dependent minimum requirement $O_g(\bomega)$; and (2) we do not know the number and the scores of future group-$g$ arrivals in the last summation.

We circumvent this limitation by creating two confidence bound proxies. Specifically, let $\beta>0$ be a hyperparameter to capture how conservative the DM wants to be for the proxies. First, we replace $O_g(\bomega)$ with $\bar{O}_g(\beta) = \expect{O_g(\bomega)} + \gamma_g(\beta)$, which is an upper confidence bound  of $O_g(\bomega)$ with probability determined by $\beta$. We obtain such a confidence bound via Lemma~\ref{lem:bpc-og-conc} or by Monte-Carlo methods. When clear from the context, we use $\gamma_g$ that drops the dependence on $\beta$. Second, for the last summation in \eqref{eq:clairvyoant}, we set $\Psi(g,t,\beta)$ so that, with  probability determined by $\beta$, $\Psi(g,t,\beta)\leq \sum_{\tau \in \Sarr(g,T) \setminus \Sarr(g,t)} \left(w_{\tau,J^{\GR}(\tau)} - \expect{O_g}-\gamma_g\right)$ for a group $g$ and a case $t$. Though $\Psi(g,t,\beta)$ gives a valid lower confidence bound on the sum, the sum itself may be an overestimate on the effect of future greedy assignments since those assignments can be infeasible when the capacity at some locations is depleted towards the end of the horizon. To compensate for this, we  remove an additional buffer term $\Cex$ from $\Psi(g,t,\beta)$. Dropping the score 
$w_{t,J^{\BP}(t)}$ from \eqref{eq:clairvyoant} and letting $V_{g}[t-1]=\sum_{\tau\in \Sarr(g,t-1)}\left(w_{\tau,J^{\alg}(\tau)}-\mathbb{E}_{\bomega}[ O_{g}(\bomega)]-\gamma_g \right)$, we replace the clairvoyant condition \eqref{eq:clairvyoant} by the following \emph{online} predict-to-meet condition:
\begin{equation}\label{eq:cond-predict}\tag{predict-to-meet}
V_{g(t)}[t-1]-\mathbb{E}_{\bomega}[O_{g(t)}(\bomega)]-\gamma_g+\Psi(g(t),t,\beta) \geq \Cex.
\end{equation}
As we alluded to above, \textsc{CBP} wants to assign case $t$ to $J^{\BP}(t)$ if this condition holds and to $J^{\GR}(t)$ otherwise; we denote this wanted assignment by $J^{\CONS}(t)$ (see line~\ref{algoline:cons-select} of Algorithm \ref{algo:conservative}). When the capacity at the wanted location is depleted, \textsc{CBP} makes exceptions like \textsc{ABP} to pick an arbitrary location with remaining capacity. For our numerical results, we implement \textsc{CBP} with the same rule as line~\ref{algoline:cons-select} of Algorithm~\ref{algo:conservative} but restricting to the locations with remaining capacity; see Appendix~\ref{app:emp-setup}.

We now introduce an assumption that often holds in practice (see Figure~\ref{fig:slack} in Appendix~\ref{app:emp-estimation}) and simplifies our analysis; our distribution-independent vanishing result holds even without it (see Remark~\ref{rem:slackness}). Specifically, to ensure the effectiveness of the greedy assignment, we assume that a property of the fairness rule, \emph{slackness}, is positive. Formally, given a distribution $\set{P}$ and a fairness rule $\mathcal{F}$, we define the slackness for a group $g$, $\varepsilon_{g,\mathcal{F}}(\set{P})$, as the expected difference between the maximum score conditional on a group $g$ and its minimum requirement, i.e., $\varepsilon_{g,\mathcal{F}}(\set{P}) = \mathbb{E}_{\boldsymbol{\theta}\sim \set{P}} \left[ \max_{j \in \Sloc} w_{j}(\boldsymbol{\theta}) \mid g(\boldsymbol{\theta}) = g \right] - \mathbb{E}_{\bomega \sim \set{P}^T}[O_{g,\mathcal{F}}(\bomega)]$. The slackness property for $\mathcal{F}$ and $\set{P}$ is as follows: 
\begin{definition}\label{def:slackness}
A fairness rule $\set{F}$ has slackness $\varepsilon$ for a feature distribution $\set{P}$ if $\varepsilon_{g,\set{F}}(\set{P}) \geq \varepsilon, \forall g \in \Sgro$.
\end{definition}
Intuitively, the slackness of a fairness rule captures by how much a group would exceed its requirement if all members of the group were assigned to their best-possible location. Ex-post feasibility guarantees the slackness is non-negative.

\begin{algorithm}[H]
\LinesNumbered
\DontPrintSemicolon
\caption{\textsc{Conservative Bid Price Control}
\label{algo:conservative}%Online algorithm of queue $i$
}
\SetKwInOut{Input}{input}\SetKwInOut{Output}{output}
\Input{Capacities $s_j$; Time horizon $T$; A $(\chi,\delta)-$sensitive fairness rule $\frule$ with slackness $\varepsilon$}
%\tcc{initialize}
 Set $(\bolds{\mu}^\star,\bolds{\lambda}^\star)\gets \arg \min \expect{L(\bolds{\mu},\bolds{\lambda})}$ where $L(\bolds{\mu},\bolds{\lambda})$ is defined in \eqref{eq:lagrangian}\;
 Set $\beta = \left(\frac{\delta}{12(M+G)T}\right)^{1/4}$,$\{\Psi(g,t,\beta)\}_{g \in \Sgro,t\in [T]}, \{\gamma_g(\beta)\}_{g \in \Sgro}$ by \eqref{eq:parameter-setting}\label{algoline:cons-init}\; 
\For{$t = 1 \ldots T$ of group $g(t)\in\Sgro$ with scores $w_{t,j}$ for $j\in\Sloc$}{
    
% Observe a new case of group $g(t)$ and its scores $w_{t,j}$\;
% $\isgreedy \gets 0$\;

% \tcc{Conservative bid price selection rule}
\textbf{if }{\eqref{eq:cond-predict} \textbf{ then } \label{algoline:cons-select} $J^{\CONS}(t) \gets J^{\BP}(t)$ \textbf{ else } $J^{\CONS}(t) \gets J^{\GR}(t)$
}\\
\textbf{if }$J^{\CONS}(t)$ has capacity \textbf{then} $J^{\alg}(t) \gets J^{\CONS}(t)$ \label{algoline:cbp}\\ \textbf{else } $J^{\alg}(t) \gets $a location with remaining capacity}\label{algoline:cons-no-capacity}%{ 
\end{algorithm}

Assuming the slackness property $\varepsilon>0$ and a $(\chi,\delta)-$sensitive fairness rule, we next set the parameters of $\alg$ (line \ref{algoline:cons-init} of Algorithm~\ref{algo:conservative}). In practice, these are set in a data-driven way via Monte Carlo methods or parameter tuning {without requiring the slackness property} (see Appendix~\ref{app:emp-setup} for further discussion). For our analysis, we instantiate parameters as:
\begin{small}
\begin{equation}
\label{eq:parameter-setting}\tag{$\textsc{PAR}$}
\begin{aligned}
\Cex = 6\ln(1/\beta),&~\gamma_g(\beta) = \min\left(1 - \expect{O_g}, 16\chi\sqrt{\frac{2\ln(1/\beta)}{p_g T}}\right); \\
\forall g \in \Sgro, t\in[T],&~\Psi(g,t,\beta) = \left\{
\begin{aligned}
-\ln(1/\beta)/\varepsilon_g,&~\text{if }\gamma_g \leq \varepsilon_g/2; \\
-T,&~\text{o.w.}
\end{aligned}
\right.
\end{aligned}
\end{equation} 
\end{small}
We let $\beta$ be a function of $\delta$ in Algorithm~\ref{algo:conservative} and set the confidence bound $\gamma_g$ based on Lemma~\ref{lem:bpc-og-conc}. A smaller value of $\beta$ yields larger values for $\Cex$, making it more difficult to satisfy \eqref{eq:cond-predict}. Roughly speaking, this leads the algorithm to more frequently assign cases greedily.

\subsection{Performance guarantee}
In the following theorem, we establish distribution-independent vanishing regret for $\alg$.
\begin{theorem}\label{thm:cbp-dis-dep}
Given a fairness rule sequence $\{\set{F}(T)\}$, suppose that each rule $\set{F}(T)$  (1) {is} ex-post feasible; (2) {is} $(\chi,\delta(T))$-sensitive with a constant $\chi \geq 1$, $\lim_{T \to \infty} \delta(T) = 0$ and $\lim_{T \to \infty} T\delta(T)^{2.6} = \infty$; and (3) has slackness $\varepsilon > 0$ uniformly for any feature distribution $\set{P} \in \set{C}$. Then $\textsc{CBP}$ has distribution-independent vanishing regret for {$\{\set{F}(T)\}$}.
\end{theorem} 
Theorem~\ref{thm:cbp-dis-dep} implies that \textsc{CBP} has distribution-independent vanishing regret for the Random and the Proportionally Optimized fairness rule; this follows from Proposition~\ref{prop:irr-rules} by setting $\delta(T)$ to be, e.g., $T^{-1/3}$, when the slackness property holds true. The MaxMin fairness rule does not satisfy our conditions as its sensitivity depends on the feature distribution (Proposition~\ref{prop:irr-maxmin}). This limitation aligns with our result in Appendix~\ref{app:max-min-imp} where we prove that it is impossible for any algorithm to obtain distribution-independent vanishing regret for MaxMin.

\begin{remark}[Removing slackness]\label{rem:slackness}
Although we state Theorem~\ref{thm:cbp-dis-dep} with stronger assumptions (positive slackness $\varepsilon>0$) than Theorem~\ref{thm:abp-dis-dep}, a simple adaptation of \textsc{CBP} allows us to obtain distribution-independent vanishing regret even when the slackness is not strictly positive. Specifically, given the fairness rule $\set{F}$ and a hyperparameter $\tilde{\varepsilon}$, we create a fictitious fairness rule $\tilde{\set{F}}$ such that its minimum requirement is uniformly $\tilde{\varepsilon}$ below that of $\set{F}$ for any group. By construction, $\tilde{\set{F}}$ comes with $\tilde{\varepsilon}$ slackness. When $\tilde{\varepsilon}$ is chosen carefully, we show that running $\alg$ for the fictitious fairness rule $\tilde{\set{F}}$ yields distribution-independent vanishing regret for the original fairness rule $\set{F}$ (see Appendix~\ref{app:zero-slackness}).  \end{remark}
Similar to Theorem~\ref{thm:abp-dis-dep}, we prove Theorem~\ref{thm:cbp-dis-dep} by obtaining global and ex-post $g-$regret guarantees for $\textsc{CBP}$ in the below two lemmas. The full proof of Theorem~\ref{thm:cbp-dis-dep} is provided in Appendix~\ref{app:cor-cbp-dis-dep}. 
\begin{lemma}\label{lem:cons-global}
Fix $\delta > 0$ and a $(\chi,\delta)-$sensitive ex-post feasible fairness rule $\frule$ with $\chi \geq 1$ and slackness $\varepsilon > 0$. For any $T \geq 3$, $\alg$ has $\set{R}_{\frule}^{\alg} \leq \frac{404\chi \ln(T/\delta)}{\hat{s}_{\min}\varepsilon}\sqrt{\frac{G}{T}}$ with probability at least $1-\delta$.
\end{lemma}
\begin{lemma}\label{lem:cons-group}
Fix $\delta > 0$ and a $(\chi,\delta)-$sensitive ex-post feasible fairness rule $\frule$ with $\chi \geq 1$ and slackness $\varepsilon > 0$. There exists $T_1$ s.t. for any group $g \in \Sgro$, if $T \geq T_1$ then $\alg$ has $\set{R}^{\mathrm{ex},\alg}_{g,\frule} \leq \frac{4848\chi \ln(T/\delta)}{\hat{s}_{\min}\varepsilon}\sqrt{\frac{G}{T}} + 2\chi \delta$ with probability at least $1-3\delta$.
\end{lemma}
\textsc{CBP} addresses the shortcoming of $\ABP$: in contrast to Lemma~\ref{lem:bid-price-group}, the ex-post $g-$regret bound in Lemma~\ref{lem:cons-group} has no dependence on a group's expected size $p_g T$ and thus $\alg$ enjoy distribution-independent vanishing regret. This comes at a cost: the bound of global regret in Lemma~\ref{lem:cons-global} is weaker than that in Lemma~\ref{lem:bid-price-global} in its dependence on $G$ and $\varepsilon$. This occurs due to the roughly $\sqrt{p_gT}/\varepsilon$ greedy steps that \textsc{CBP} may execute for each group $g$, each of which trades off a benefit in group outcomes for a loss in global objective. 
A second limitation of the guarantee in Lemma~\ref{lem:cons-group} is that it only holds for large $T$.\footnote{This is in part due to $T_1$ relying on $\delta$, which is unavoidable for any algorithm as there are instances for which, for any $\delta>0$, if $T<0.075/\delta$, with probability $\delta$ the ex-post $g$-regret of a group is bounded below by $0.5$ (Appendix~\ref{app:small_groups_failure_example}).} That said, our numerical results show \textsc{CBP} performing well in practice, for both global objective and group fairness, including in settings where $T$ is not large. 

We now provide a proof outline for Lemmas~\ref{lem:cons-global} and \ref{lem:cons-group}. Our results are stated for any {confidence bounds} $\bolds{\gamma}$ with $\gamma_g \in \left[0,1-\expect{O_g}\right]$. The proofs of the theorems follow by setting $\bolds{\gamma}$ {from \eqref{eq:parameter-setting}}.

\paragraph{Global regret.} A vital quantity in the analysis of \textsc{CBP} is the first time a location runs out of capacity. In a slight abuse of notation, we denote this quantity by $\Temp$ as it plays the same role as in the analysis of \textsc{ABP}. 
The key difference in lower bounding $\Temp$ for \textsc{CBP} arises from the greedy selections before $\Temp$. 
Because the number of greedy steps is small, we can bound $\Temp$ similarly to Lemma~\ref{lem:bpc-temp} for~\textsc{ABP}. For any confidence bounds $\bolds{\gamma}$, define $$C(\bolds{\gamma}) = \frac{12}{\varepsilon}\sum_{g\colon \gamma_g \leq \nicefrac{\varepsilon_g}{2}}\gamma_gp_g T+6\sum_{g \colon \gamma_g > \nicefrac{\varepsilon_g}{2}}p_g T,$$ to capture the extra greedy steps that are necessary because of $\bolds{\gamma}$. Denoting $$\Delta^{\CONS}= \frac{1}{\hat{s}_{\min}}\left(\frac{20\sqrt{GT}\ln(1/\beta)}{\varepsilon}+\frac{51G\ln(1/\beta)}{\varepsilon^2}+C(\bolds{\gamma})\right)\text{, we show:}$$
\begin{lemma}\label{lem:cons-temp}
With probability at least $1 - (2T+1)(M+G)\beta^4$, we have $\Temp \geq T - \dcbp.$
\end{lemma}
\begin{proof}[Proof sketch]
We introduce a 
fictitious system in which each location has unlimited capacity. In this fictitious system, the clause in line~\ref{algoline:cons-no-capacity} is never triggered, so \textsc{CBP} assigns every case to $J^{\CONS}(t)$. In the real system, for every $t\leq\Temp$, $\textsc{CBP}$, also assigns to $J^{\alg}(t)=J^{\CONS}(t)$; hence~\textsc{CBP} makes the same assignments in the fictitious and the real system when~$t\leq\Temp$. 

We then show that the total number of greedy steps in both the fictitious
system and the real system
are upper bounded by $\tilde{O}(\frac{\sqrt{GT}}{\varepsilon})$ with high probability (Lemma~\ref{lem:bound-step-greedy}). Intuitively, Lemma~\ref{lem:bp-property-fair} implies that the difference between the requirement and the realized total score of a group $g$ is around $\sqrt{p_g T}+\gamma_gp_g T$. Therefore, to achieve condition~\eqref{eq:cond-predict}, there is a deficit of $\sqrt{p_g T}+\gamma_gp_g T$ to be filled by greedy steps. By the slackness property (Definition~\ref{def:slackness}), a greedy step (in expectation) fills at least $\varepsilon$ of the gap toward the minimum requirement. Therefore, a group $g$ takes around $(\sqrt{p_g T}+\gamma_gp_g T)/\varepsilon$ greedy steps to cover the deficit. The Cauchy-Schwarz inequality then gives a bound on the total number of greedy steps of all groups. The formal proof relies on a concentration bound motivated by {work on conservative bandits} \cite{DBLP:conf/icml/WuSLS16} and is provided in Appendix~\ref{app:proof-cons-temp}. 
\end{proof}
We now bound the global regret of $\alg$ for any confidence bound $\bolds{\gamma}$ (proof  in Appendix~\ref{app:proof-cons-global-conf}). 
\begin{lemma}\label{lem:cons-global-conf}
Fix an ex-post feasible fairness rule $\frule$ with slackness $\varepsilon$. For any $\delta > 0$, let $T_0 = \frac{12(M+G)}{\varepsilon^2}$. Then $\forall T\geq T_0$, \textsc{CBP} with confidence bound $\bolds{\gamma}$ has $\set{R}_{\frule}^{\alg} \leq \frac{20\ln(T/\delta)}{\hat{s}_{\min}\varepsilon}\sqrt{\frac{G}{T}}+\frac{C(\bolds{\gamma})}{\hat{s}_{\min} T}$ with probability at least $1-\delta$.
\end{lemma}
The bound of global regret in Lemma~\ref{lem:cons-global} follows by substituting in \eqref{eq:parameter-setting} for $\bolds{\gamma}$; with that substitution, the following lemma simplifies the value of $C(\bolds{\gamma})$ (proof in Appendix~\ref{app:bound-cgamma}).
\begin{lemma}\label{lem:bound-cgamma}
For any $\delta > 0$, if $T \geq 1849(M+G)$, $\bolds{\gamma}$ as in \eqref{eq:parameter-setting} and $\beta = \left(\frac{\delta}{12(M+G)T}\right)^{1/4}$, then $C(\bolds{\gamma}) \leq \frac{384\chi \ln(T/\delta)\sqrt{GT}}{\varepsilon}$ and $\dcbp \leq \frac{404\chi \ln(T/\delta)\sqrt{GT}}{\hat{s}_{\min}\varepsilon}$.
\end{lemma}
\begin{proof}[Proof of Lemma~\ref{lem:cons-global}]
 Suppose that $T \geq \frac{1849(M+G)}{\varepsilon^2}$. By Lemma~\ref{lem:cons-global-conf}, with probability at least $1-\delta$, 
\[
\set{R}_{\frule}^{\alg} \leq \frac{20\ln(T/\delta)}{\hat{s}_{\min}\varepsilon}\sqrt{G/T}+\frac{C(\gamma)}{\hat{s}_{\min}T} \overset{\text{Lemma~\ref{lem:bound-cgamma}}}{\leq} \frac{404\chi \ln(T/\delta)}{\hat{s}_{\min}\varepsilon}\sqrt{G/T}.
\]
If $T \leq 1849(M+G)/\varepsilon^2$, we have
$\frac{404\chi \ln(T/\delta)}{\hat{s}_{\min}\varepsilon}\sqrt{G/T} \geq \frac{404\chi M}{\varepsilon}\cdot \frac{\varepsilon}{\sqrt{1849(M+G)/G}} \geq 1,
$
where the first inequality is because $\ln(T/\delta) \geq 1$ since $T \geq 3$, and the second inequality is because $(404/\sqrt{1849}) \geq 9$ and $9M\sqrt{G/(M+G)}\geq 1$. Since we always have the global regret upper bounded by $1$, the bound $\set{R}_{\frule}^{\alg}\leq\frac{404\chi \ln(T/\delta)}{\hat{s}_{\min}\varepsilon}\sqrt{G/T}$ trivially holds for $T \leq 1849(M+G)/\varepsilon^2$. We thus complete the proof.
\end{proof}

\paragraph{Ex-post $g-$regret.} 
The crux of proving Lemma~\ref{lem:cons-group} is to connect the performance of a group $\alpha_g(\bomega)$ with the condition \eqref{eq:cond-predict}. We want to show that, with high probability, if a group meets condition~\eqref{eq:cond-predict} in any period, then its outcome score $\alpha_g$ will end up near the upper confidence bound $\bar{O}_g = \expect{O_g}+\gamma_g$. To formalize this, 
we define an event $\Sbp$ such that conditioned on it, the following are true: (i) $\Temp\geq T-\Delta^{\CONS}$, and (ii) if a group $g$ fulfills \eqref{eq:cond-predict} in some period $t$, where $g=g(t)$, then $\alpha_g\geq \expect{O_g} +\gamma_g- \frac{12\Delta^{\CONS}}{T}$. In the following Lemma we show that $\Sbp$ occurs with high probability.  Conditioned on this event, bounding group regret then reduces to showing the following: (a) for any group that meets condition~\eqref{eq:cond-predict} in some period, the result guarantees a good outcome; (b) a group with many arrivals and a small confidence bound is unlikely to never meet condition~\eqref{eq:cond-predict} in any period, and (c) a group with few arrivals, or a large confidence bound, under \textsc{CBP} has its cases assigned almost exclusively to their greedy location.

\begin{lemma}\label{lem:fairness-bpstep}
Suppose that $T \geq 36\dcbp$. Then $\mathbb{P}[\Sbp]\geq 1-(3T+3)(M+G)\beta^4$.
\end{lemma}

\begin{proof}[Proof sketch.]
Lemma \ref{lem:cons-temp} bounds the probability of $\Temp \geq T-\Delta^{\CONS}$. To prove the lemma we need to further show, with high probability, that if a group $g(t)$ fulfills condition~\eqref{eq:cond-predict} in some period $t$, then it ends up with $\alpha_g \geq \expect{O_g}+\gamma_g - \frac{12\Delta^{\CONS}}{T}$ at the end of the horizon. Consider the last period $t$ in which group $g(t)$ fulfills condition~\eqref{eq:cond-predict}. In every future period $\tau$ in which there is an arrival of group $g(t)$, the condition does not hold true. Thus, in those periods, $J^{\CONS}(\tau)=J^{\GR}(\tau)$. Moreover, we set $\Psi(g,t,\beta)$ in \eqref{eq:parameter-setting} to be a valid lower bound such that with high probability, $\Psi(g,t,\beta)\leq \sum_{\tau \in \Sarr(g,T) \setminus \Sarr(g,t)}\left(w_{\tau,J^{\GR}(\tau)}-\expect{O_g}-\gamma_g\right)$ (Lemma \ref{lem:valid-lowerbound} in Appendix~\ref{app:proof-fairness-bpstep}). Since \eqref{eq:cond-predict} holds for period $t$ we must have
$$\sum_{\tau \in \Sarr(g,T)}w_{\tau,J^{\CONS}(\tau)} \geq N(g,T)(\expect{O_g}+\gamma_g) +\Cex.$$
For every case $\tau$ before~$\Temp$, we have $J^{\alg}(\tau) = J^{\CONS}(\tau)$ because all locations have capacity and thus Line~\ref{algoline:cbp} of Algorithm~\ref{algo:conservative} applies. Therefore, 
\[
\sum_{\tau \in \Sarr(g,T)}w_{\tau,J^{\alg}(\tau)} \geq N(g,T)(\expect{O_g}+\gamma_g)  + \Cex - \left(N(g,T) - N(g,\Temp)\right).
\]
Dividing both sides by $N(g,T)$ gives $\alpha_g \geq \expect{O_g} +\gamma_g+ \frac{\Cex - \left(N(g,T) - N(g,\Temp)\right)}{N(g,T)}$. For sufficiently large groups, we can show that $\frac{N(g,T)-N(g,\Temp)}{N(g,T)}$ is of the order of $\frac{T-\Temp}{T}$ with high probability. For small groups, the choice of $\Cex$ ensures that $\Cex \geq N(g,T)-N(g,\Temp)$. Therefore, $\alpha_g$ is roughly lower bounded by $\expect{O_g} +\gamma_g- \frac{T-\Temp}{T}$. The full proof is provided in Appendix~\ref{app:proof-fairness-bpstep}.
\end{proof}
We next bound {the ex-post $g-$regret} for groups of sufficiently large size (expectation greater than a constant).
Specifically, Lemma~\ref{lem:fair-case-1} below shows that for such a group, the total greedy scores of cases that arrive before $\Temp$ is large compared to the minimum requirement of group $g$. Therefore, condition~\eqref{eq:cond-predict} is met for at least one case with high probability, allowing us to apply Lemma \ref{lem:fairness-bpstep}. The formal proof is in Appendix~\ref{app:proof-fair-case-1}.
\begin{lemma}\label{lem:fair-case-1}
Suppose that $T \geq 36\dcbp$ and $\beta \leq e^{-1}$. For a group $g$, if $\gamma_g \leq \varepsilon_g / 2$ and $p_g \geq \frac{144\ln(1/\beta)
}{\varepsilon_g^2 T}$, then $\alpha_g \geq \expect{O_g}+\gamma_g-\frac{12\dcbp}{T}$ with probability at least $1-(5T+5)(M+G)\beta^4$.
\end{lemma}
Now, consider a group $g$ with small $p_g$. If condition~\eqref{eq:cond-predict} is satisfied for one of its cases, Lemma~\ref{lem:fairness-bpstep} holds and its average score is guaranteed to nearly exceed its expected minimum requirement with high probability. If, on the other hand, the \eqref{eq:cond-predict} condition is not met for its cases before $T-\Temp$, then they all receive greedy assignments. Moreover, we show that this group has no arrival after $T-\Temp$ with high probability since it has small $p_g$. The average score is thus the highest possible and is at least the minimum requirement by ex-post feasibility of the fairness rule. These discussions show that for a small group $g$, its average score is at least close to $\min(O_g,\expect{O_g}+\gamma_g)$ with high probability. We summarize the result below (with a formal proof given in Appendix~\ref{app:proof-fair-case-2}).

\begin{lemma}\label{lem:fair-case-2}
Suppose that $T \geq 36U\Delta^{\CONS}$. For a group $g$ with $\gamma_g > \varepsilon_g / 2$ or $p_g < \frac{144\ln(1/\beta)}{\varepsilon_g^2 T}$, with probability at least $1 - (5T+5)(M+G)\beta^4-p_g\dcbp$: $\alpha_g \geq \min\left(\expect{O_g}+\gamma_g,O_g\right)-\frac{12\dcbp}{T}$.
\end{lemma}
The bound on group regret in Lemma~\ref{lem:cons-group} follows from combining Lemmas~\ref{lem:fair-case-1} and~\ref{lem:fair-case-2} and plugging in $\bolds{\gamma}$ according to \eqref{eq:parameter-setting}; see Appendix \ref{app:proof-cons-group} for the full proof.

\section{Numerical Experiments and Practical Implications}\label{sec:numerical}
% !TEX root = main.tex
In this section we construct instances based on real-world data from the US and the Netherlands. We compare the minimum requirements under different fairness rules, show the limitations of status-quo approaches in achieving them, and study the empirical performance of our algorithms.

\noindent\textbf{Data and instances.}
We use (de-identified) data from 2016 that cover
free cases among (i)~adult refugees that were resettled by one of the largest US resettlement agencies ($T=1,175$), and (ii)~asylum seekers that were resettled in the Netherlands\footnote{In the Netherlands, we specifically focus on the population of status holders who are granted residence permits and are assigned to a municipality through the regular housing procedure. We exclude some subsets of status holders who fall outside the scope of the objectives or for whom data are unreliable. First, we exclude status holders who fell under the 2019 Children’s Pardon. Second, we exclude resettlers / relocants / asylum seekers covered by the EU-Turkey deal due to ambiguity in the data recorded on their registration.} ($T=1,543$). In both countries, the outcome of interest is whether or not the refugee/asylum seeker found employment within a given time period (90 days in the US, 2 years in the Netherlands).\footnote{The period lengths and group definitions we consider stem from private communication with respective agencies.}
{For each case $t$ we apply the ML model of Bansak et al. \cite{bansak2018improving} to infer the employment probability $w_{t,j}$ at each location~$j$.\footnote{When cases consist of multiple individuals, we aggregate them by calculating the probability of at least one member of each case finding employment.} More details on the prediction method and properties of resulting predictions are described in Appendix \ref{app:emp-estimation}. Moreover, we set each capacity $s_j$ as the number of cases that were assigned to location $j$ in 2016}.

We define three scenarios: \textsc{NL-Age} and \textsc{NL-Edu} where groups in the Netherlands are defined by either \emph{age} or \emph{education level}, and \textsc{US-CoO} where groups in the US are defined by \emph{country of origin}. In the \textsc{NL-Age} scenario, age is segmented into $10$ brackets, whereas cases in the \textsc{NL-Edu} scenario fall into one of seven groups based on educational attainment and degrees earned. When referring to the groups, we index them in increasing order of their size for \textsc{US-CoO}, and in increasing order of their age and education level for \textsc{NL-Age} and \textsc{NL-Edu}, respectively. As some cases consist of multiple individuals, we use the age/education level of the primary applicant to categorize each case into a group. These group definitions were chosen for both illustrative purposes and based on their importance to our partner organizations, though our approaches apply beyond these specific ones. Table \ref{tab:group-stats} provides summary statistics on the arrivals and locations. We note that the \textsc{US-CoO} has some very small groups (of size 1, see Table \ref{tab:group-stats}) and thus is of particular interest for comparing ABP and CBP.

\begin{table}[H]
    \centering
    \begin{tabular}{cccccc}
    & $T$ & $M$ & $G$ & Smallest group size & Largest group size \\
    \hline
    \textsc{US-CoO} & 1,175 & 27 & 26 & 1 & 331\\
    \textsc{NL-Age} & 1,543 & 35 & 10 & 38 & 230 \\
    \textsc{NL-Edu} & 1,543 & 35 & 7 & 30 & 428 
    \end{tabular}
    \caption{Summary statistics}
    \label{tab:group-stats}
\end{table}

Our numerical results are displayed over 50 bootstrapped samples $\bomega_1,\ldots, \bomega_{50}$ drawn i.i.d. from the empirical distribution of 2016 arrivals. However, to illustrate the performance of our methods on real-world, non-stationary data, we also backtest them on 2016 arrivals in the order in which they occurred, with our algorithms only using information that was available at the start of 2016 (Appendix \ref{app:real-world}).
Implementation details of all online algorithms can be found in Appendix \ref{app:emp-setup}.

\subsection{Minimum requirements} \label{sec:emp-fairness-targets}
We start by discussing the minimum requirements prescribed by the Random, Proportionally Optimized and MaxMin fairness rules (Examples~\ref{ex:random-fairness-rule}-\ref{ex:maxmin-fairness-rule} in Section~\ref{sec:model}) for our three scenarios (see Figure~\ref{fig:fairness_benchmarks}). The minimum requirement under the Random fairness rule (orange) is always below that of the Proportionally Optimized fairness rule (green). This holds by construction: each group receives the same (fractional) capacity at each location under both rules but the latter rule optimizes assignments for each group. Interestingly, due to higher employment rates among younger populations, in the \textsc{NL-Age} scenario (Figure \ref{fig:fairness_benchmarks}a) the minimum requirements under both the Random and Proportionally Optimized fairness rule decrease as age increases.
Furthermore, the gap between {these two} requirements decreases with the group's age. This suggests that there is more ``room to optimize''---or more synergies between people and places---among younger populations. Conversely,  for small groups in the \textsc{US-CoO} scenario {with little room to optimize}, the requirements under Proportionally Optimized and Random are often equal.  For groups with larger index in Figure \ref{fig:fairness_benchmarks}c, the gap between the two rules is more pronounced. For the MaxMin fairness rule (grey), the requirements are equal across groups and have no clear relationship to the requirements of the other fairness rules. For example, for groups 1-5 of the \textsc{NL-Age} scenario, the MaxMin requirement is the least constraining whereas for groups 9 and 10 it is the most constraining.

The monotonic relationship between the minimum requirements of Random and Prorportional Optimized has a clear implication for the group fairness and total employment we can expect: for any benchmark that does not take into account the fairness rule when making decisions, the outcomes will appear fairer under Random than under Proportionally Optimized; and for algorithms like \textsc{ABP} and \textsc{CBP} that try to meet group fairness constraints, the global objective will be higher under Random than under Proportionally Optimized. Because MaxMin does not have a monotonic relationship with the other fairness rules, we cannot predict the outcomes under this fairness rule.
\begin{figure}[H]
    \centering    \includegraphics[width = .8\textwidth]{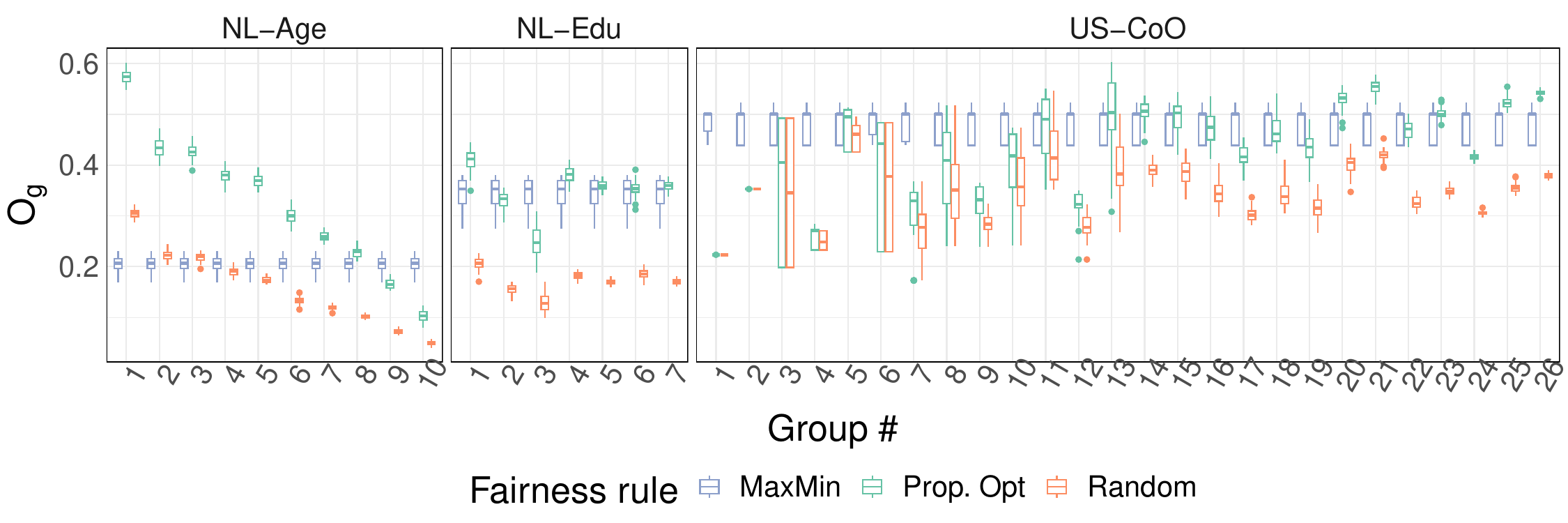}
    \caption{Minimum requirements for each group and scenario over bootstrapped arrival sequences.}
    \label{fig:fairness_benchmarks}
\end{figure}

\subsection{Group fairness limitations of  status quo approaches}\label{sec:emp-fairness-algos}
We now discuss the group fairness of two algorithms that closely mimic status quo procedures: \textsc{Rand} and \textsc{BP}. \textsc{Rand} assigns each case randomly to a location with available capacity; this closely resembles status quo operations in the absence of optimization.
\textsc{BP} (see Section~\ref{sec:bp}) is similar to deployed algorithms (see \cite{elisabethpaper,ahani2021dynamic}) and closely approximates the employment-maximizing (without fairness constraints) hindsight-optimal assignment, {referred to as} $\opt$. Although we discuss results in this section with respect to $\textsc{BP}$, similar insights arise for $\opt$ (see Appendix~\ref{app:emp-fairness-algos}).

To measure group fairness, for a group $g$ and a sample path $\bomega$, we define the \emph{fairness ratio}~as 
$$\textsf{FR}_{g}(\bomega)\triangleq\frac{\alpha_{g}(\boldsymbol{\omega})}{O_{g,\mathcal{F}}(\boldsymbol{\omega})}.$$
This is related to, but not the same as, ex-post $g-$regret, and allows for a more interpretable comparison across contexts. A fairness ratio of 0.5 means that a group's average employment score is half of its minimum requirement prescribed by~$\mathcal{F}$, and a  fairness ratio greater than one {implies} that employment levels are higher than the requirement. To aggregate across groups, we define the minimum fairness ratio $\textsf{FR}(\bomega)\triangleq\min_g\textsf{FR}_g(\bomega)$  and its sample average  $\textsf{FR}=\frac{1}{50}\sum_{n=1}^{50}\textsf{FR}(\bomega_n)$.

We first examine the performance of \textsc{Rand} under our three fairness rules  (see top of Figure \ref{fig:benchmark-fairness-performance}). Recall that, for case $t$, their minimum requirement under the Random fairness rule is equivalent to their \emph{expected} outcome under \textsc{Rand}. \textsc{Rand}, however, must choose a single location for each case. Thus, the fairness ratios are generally centered around one, but {$\textsf{FR}_{g}(\bomega)$} for a single sample path can be large, especially for the small groups in the \textsc{US-CoO} scenario. Results follow similar patterns under the Proportionally Optimized fairness rule but are more unfair. This arises because the minimum requirement obtained under the Proportionally Optimized fairness rule is larger than that obtained under the Random fairness rule (see Figure \ref{fig:fairness_benchmarks}), especially for large~groups.

We then consider the MaxMin fairness rule (see again top of Figure~\ref{fig:benchmark-fairness-performance}). In the \textsc{NL-Age} scenario, \textsc{Rand} is usually fair for groups 1-3, but not for groups 4-10. The observation that fairness decreases with age is again a consequence of the higher overall employment rates among younger populations. To build intuition, consider a stylized example where younger groups are highly employable at any location, but older groups have only one location with high employment scores. In order to achieve their MaxMin requirement, a disproportionate number (relative to group size) of older cases would need to be assigned to that location. Since, under \textsc{Rand}, their expected allotment in each location is exactly proportional to their size, they do not meet their minimum requirement. In the \textsc{NL-Edu} and \textsc{US-CoO} scenarios, results under the MaxMin fairness rule reflect those obtained under proportionally optimized fairness. 

Finally, we discuss the fairness ratio of $\textsc{BP}$ (bottom of Figure~\ref{fig:benchmark-fairness-performance}). For most groups {(with the exception of certain small groups in the  \textsc{US-CoO} scenario)} \textsc{BP} is fair under the Random fairness rule. However, BP may be unfair for certain groups under both the MaxMin and Proportionally Optimized fairness rules. The fairness ratio is particularly low under MaxMin fairness for groups 9 and 10 in the \textsc{NL-Age} scenario, group 3 in the \textsc{NL-Edu} scenario, and various groups in the \textsc{US-CoO} scenario. Therefore, although \textsc{BP} meets the minimum requirements for more groups than \textsc{Rand}, BP can be unfair, sometimes extremely so, for some groups.
\begin{figure}[h]
    \centering
    \includegraphics[width = .8\textwidth]{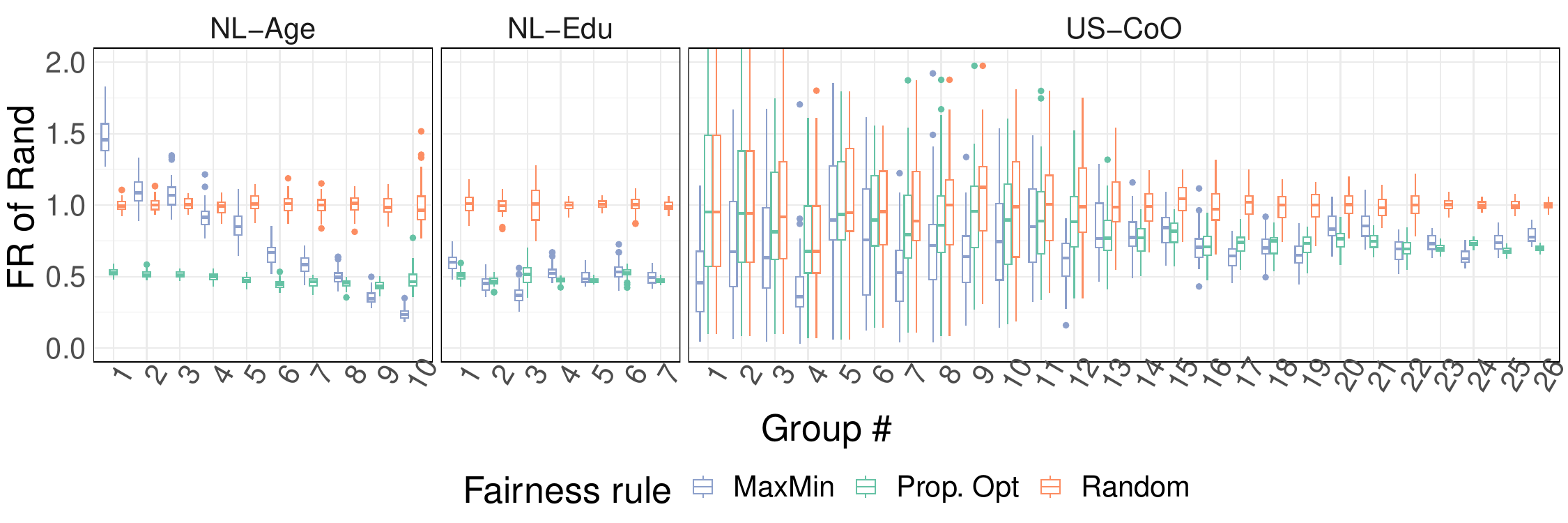}
    \includegraphics[width = .8\textwidth]{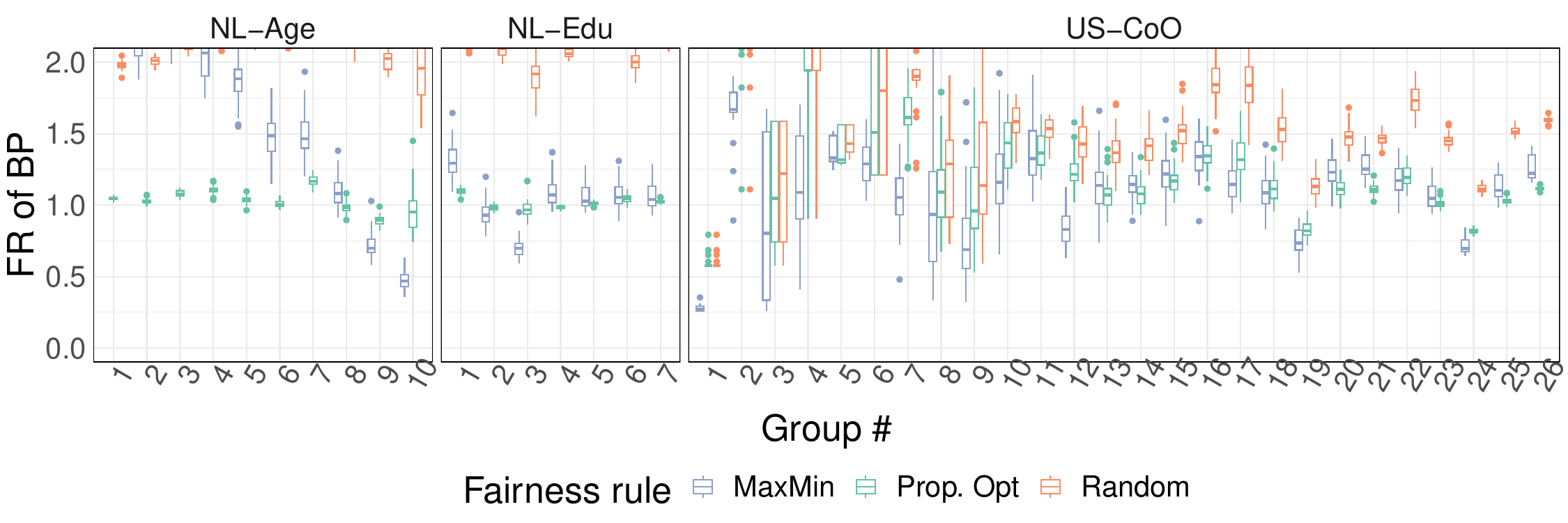}
    \caption{$\textsf{FR}_{g}(\bomega)$ of \textsc{Rand} and \textsc{BP} with respect to the three fairness rules in each scenario.}
    \label{fig:benchmark-fairness-performance}
\end{figure}

\subsection{Group and global objectives of our algorithms}\label{sec:emp-algos}
We now compare the fairness ratio and total employment score achieved by our algorithms (\textsc{ABP} and \textsc{CBP}) with the above status-quo algorithmic benchmarks (\textsc{Rand} and \textsc{BP}) and the optimal offline benchmarks, \ref{eq:outcome-benchmark} and $\opt$. We note that the latter rely on knowledge of the sample path $\bomega$ and thus are not implementable in practice but are useful as upper bounds for the total employment with and without fairness constraints, respectively.

\noindent\textbf{Group fairness. }
Figure \ref{fig:fairness-results} shows the distribution of \textsf{FR}$(\bomega)$ across groups over 50 random arrival sequences. The labels above or below each violin plot show \textsf{FR}---the average of \textsf{FR}$(\bomega)$ across the 50 sequences.

Our analysis  finds that each of our benchmarks\footnote{This is with the exception of \ref{eq:outcome-benchmark} that has a fairness ratio greater than one by definition.} exhibits an \textsf{FR} less than 0.5 in multiple instances (see Figure \ref{fig:fairness-results}). In comparison, \textsc{ABP} and/or \textsc{CBP} improve upon the \textsf{FR} in all instances, often dramatically, with the exception of the \textsc{NL-Edu} scenario under the Proportionally Optimized fairness rule.
 In this instance, \textsc{OPT} already has very large \textsf{FR} and therefore \textsc{BP}, \textsc{ABP}, and \textsc{CBP} perform similarly, and are all less fair than \textsc{OPT} due to the efficiency loss from making decisions online.
Focusing on the scenarios from the Netherlands (which have only large groups, see Table \ref{tab:group-stats}), \textsc{ABP} has an \textsf{FR} of 91\%; in the \textsc{US-CoO} scenario that does contain small groups, \textsc{CBP} has an \textsf{FR} of 92\%. More generally, in line with our expectations from Theorems~\ref{thm:abp-dis-dep} and~\ref{thm:cbp-dis-dep}, \textsc{ABP} and \textsc{CBP} perform similarly well on \textsf{FR} when all groups are large, and \textsc{CBP} does significantly~better than \textsc{ABP} in the US-CoO scenario where there are small groups.

The existence of small groups in the \textsc{US-CoO} scenario warrants a more detailed investigation of each algorithm's \textsf{FR}. Focusing on small (less than 10 arrivals) and large groups (more than 100), we find that \textsc{CBP} dramatically reduces the proportion of sample paths with $\textsf{FR}_{g}(\bomega)<1$ for small groups. Among large groups, \textsc{CBP} results in a larger proportion of sample paths with $\textsf{FR}_{g}(\bomega)<1$ than \textsc{ABP}, though the magnitude is small in these cases (see
Table \ref{tab:group-CoO-regret-v-group} in Appendix \ref{app:emp-algos} for more details). This demonstrates a trade-off: {\textsc{ABP}} guarantees group-fairness for large groups, whereas {\textsc{CBP}} introduces inefficiencies to help all groups meet their minimum requirement.

\begin{figure}[h]
    \centering    \includegraphics[width=\textwidth]{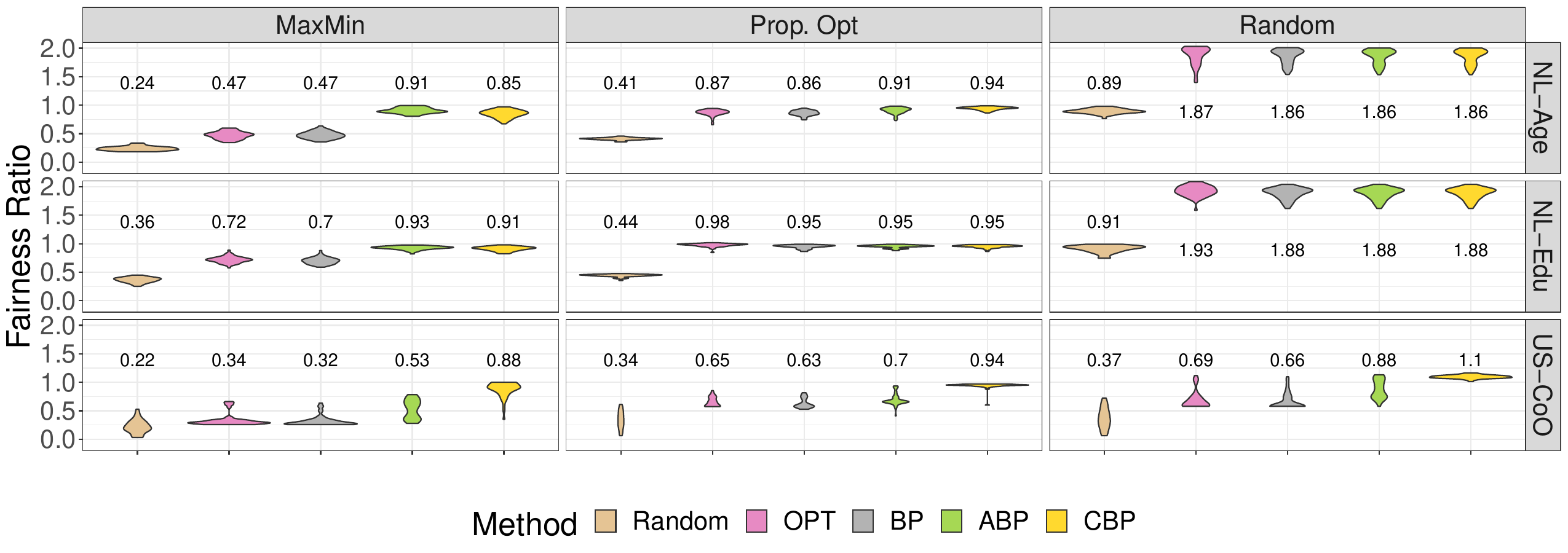}
\caption{Distribution of \textsf{FR}$(\bomega)$ under the three fairness rules for each scenario across 50 bootstrapped arrival sequences. The labels show \textsf{FR}---the average value of \textsf{FR}$(\bomega)$. We note that, unlike ABP and CBP, the Random, OPT, and BP algorithms are independent of the chosen fairness rule.}
    \label{fig:fairness-results}
\end{figure}

\noindent\textbf{Global objective. } 
The last goal of our data-driven exploration is to compare the global objective of the approaches and benchmarks (see Table~\ref{tab:efficiency}).
By comparing \ref{eq:outcome-benchmark} to OPT we can characterize the  loss in total employment necessary to accommodate group fairness. 
With just one exception, \ref{eq:outcome-benchmark} always obtains at least 99\% of $\opt$, showing that group fairness may come at a low cost. This demonstrates our earlier claim (see Section \ref{ssec:contributions}) that real-world instances do not reflect World (A) in Figure \ref{fig:simple-example}, i.e., ensuring group fairness does not necessitate a significant decrease in the global objective. In some cases, this occurs because the minimum requirements are met ``for free'', i.e., the fairness constraints are naturally fulfilled when maximizing total employment (see Figure~\ref{fig:fairness-results}). A second and more nuanced reason relates to patterns in the employment score vectors. When these are highly correlated, yet slightly offset, swapping the assignment of two cases in different groups may yield small changes in total employment but large changes in average group employment.

All of the online methods ({\textsc{BP}}, \textsc{ABP}, and \textsc{CBP}) average at least 95\% of the total employment level under \textsc{OPT}. The ordering between these three is unsurprising since \textsc{BP} optimizes solely for total employment while \textsc{ABP} and \textsc{CBP} introduce inefficiencies with the goal of increasing degrees of fairness. However, \textsc{ABP} loses no more than 2\% efficiency compared to \textsc{BP}, and \textsc{CBP} loses no more than 3\% efficiency compared to \textsc{BP}. These small losses in efficiency should be considered alongside the improvement in fairness that they achieve. 
\begin{table}[H]
\centering
\begin{tabular}{llllllllllll}
  \hline
  & \multicolumn{3}{c}{NL-Age} && \multicolumn{3}{c}{NL-Edu} && \multicolumn{3}{c}{US-CoO}\\
  \cline{2-4} \cline{6-8} \cline{10-12}
Rule: & \footnotesize MaxMin & \footnotesize PrOpt & \footnotesize Random && \footnotesize MaxMin & \footnotesize PrOpt & \footnotesize Random && \footnotesize MaxMin & \footnotesize PrOpt & \footnotesize Random \\ 
  \hline
\textsc{Rand} & 0.46 & 0.46 & 0.46 && 0.46 & 0.46 & 0.46 && 0.67 & 0.67 & 0.67 \\ 
\textsc{CBP} & 0.95 & 0.97 & 0.98 && 0.96 & 0.96 & 0.98 && 0.96 & 0.98 & 0.99 \\ 
\textsc{ABP} & 0.96 & 0.98 & 0.98 && 0.98 & 0.98 & 0.98 && 0.98 & 0.99 & 0.99 \\ 
\textsc{BP} & 0.98 & 0.98 & 0.98 && 0.98 & 0.98 & 0.98 && 0.99 & 0.99 & 0.99 \\ 
\footnotesize{OFFLINE$_\mathcal{F}$} & 0.97 & 1.00 & 1.00 && 0.99 & 1.00 & 1.00 && 0.99 & 1.00 & 1.00 \\ 
   \hline
\end{tabular}
\caption{Average ratio of total employment achieved under each algorithm compared to~\textsc{OPT}.}
\label{tab:efficiency}
\end{table}

\subsection{Practical Implications} \label{sec:practical-implications}
Our work aims to contribute to the algorithmic toolbox of GeoMatch, a decision support tool developed by the Immigration Policy Lab (IPL) at Stanford University. GeoMatch builds on methods first proposed in \cite{bansak2018improving}. 
Currently deployed or in development in multiple countries, the GeoMatch tool is tailored to each country context and provides location recommendations to human case workers who make the final geographic placement decision. In what follows we discuss the practical requirements that guided the development of our approach. 

Partner organizations would interact with the algorithms proposed in the paper in two ways: 1) the partner provides the definition for the groups and the fairness rule and 2) case workers at the partner organization receive location recommendations. Thus, both interactions should be relatively simple and explainable. The methods in this paper do not change interaction (2) relative to the status quo (although we note that the explanation for the recommendations may need to be slightly modified for these algorithms compared to the currently deployed algorithm of \cite{elisabethpaper}; however, the main complexity in explaining recommendations stems from the core idea of opportunity costs in online resource allocation more generally. The added complexity of ABP and BP is thus marginal within the space of online algorithms). We have made interaction (1) as straightforward as possible by introducing a simple framework for fairness rules, and algorithms that can be adapted to a large set of fairness rules.

\section{Conclusions}\label{sec:conclusions}
% !TEX root = main.tex
In this paper, we introduced the first formal approach to incorporating group fairness into the dynamic refugee assignment problem. Our approach explicitly addresses the practical needs of refugee resettlement stakeholders by offering flexibility to policy makers when  specifying the notions of group fairness desiderata they wish to attain. Moreover, we developed new assignment algorithms  that are guaranteed to achieve strong global and group-level performance guarantees for \emph{any} feasible definition of  group fairness chosen by the policy maker.  Finally, through extensive numerical experiments on multiple real-world data, we showed that our online algorithms can achieve desired group-level fairness outcomes with very small changes in global performance. 

Both our theoretical and empirical results demonstrate the performance of our two proposed algorithms, \textsc{ABP} and \textsc{CBP}. While \textsc{BP} achieves a higher overall average employment rate than our proposed methods, this results in unfair outcomes for certain groups under various fairness rules. In some settings, \textsc{BP} gets ``lucky'' and achieves group fairness, for example when the minimum requirements are quite loose as in the \textsc{NL-Edu} scenario under the Random fairness rule. \textsc{ABP} improves upon \textsc{BP} in terms of fairness, working well when all groups are relatively large and obtaining an overall average employment score within 2\% of \textsc{BP} on real-world data. \textsc{CBP} sacrifices some efficiency ($\sim2\%$) relative to \textsc{ABP} in order to achieve fairer outcomes for all groups regardless of size. Therefore, it is a natural choice in settings with small groups.

Our work opens up a number of promising directions for future research: 
\begin{itemize}
    \item In this paper, we make the assumption that each individual belongs in a unique group; this makes sense when the group definition is based on a single attribute (country of origin, age, educational level, gender) and is important for our algorithms to know how to assign reserves. That said, in practice, refugees and asylum seekers have multiple attributes and their profile is defined as an intersection of the respective subgroups. Designing algorithms that account for such \emph{intersectional} groups is a fundamental direction for future work.    
    \item In this paper, we assume that a case consists of a single individual and our results can easily extend to the case where cases have multiple individuals who all belong to the same group. That said,  many families consist of members of different gender, educational level, age, etc. Since families must be assigned to the same location, our methods do not readily apply to settings with  within-case group-heterogeneity and this is an important extension. 
    \item Although our framework is flexible enough to incorporate various fairness rules, these rules need to be defined as a minimum requirement on the average performance on the group. This precludes fairness rules that aim to optimize, say, for the performance achieved by the $10$th-percentile of the group, which may be a meaningful statistic as it disallows a few good outcomes to outweigh the experience of the majority of the group members. Providing results that account for {more general} fairness rules can thus be a useful direction for future work. 
\end{itemize}

Finally, we hope that future work may extend our algorithmic techniques for achieving group fairness that is specified by an \emph{ex-post minimum requirement} (Section~\ref{ssec:formal_model}) to various other public policy settings in which individuals must be assigned dynamically (e.g., in health care and housing). More broadly, this work is a concrete example of how AI can be used to achieve complex and socially beneficial policy goals---an area with significant opportunity for future research.

\section*{Acknowledgement}
We acknowledge funding from the Charles Koch Foundation, Google.org, Stanford Impact Labs, the National Science Foundation, and J-PAL for this work. This work was also supported by Stanford University’s Human-Centered Artificial Intelligence (HAI) Hoffman-Yee Grant ``Matching Newcomers To Places: Leveraging Human-Centered AI to Improve Immigrant Integration''. The funders had no role in the data collection, analysis, decision to publish, or preparation of the manuscript. The authors are also grateful to
the \emph{Simons Institute for the Theory of Computing} as part of this work was done during the Fall’22
semester-long program on \emph{Data Driven Decision Processes}.

\bibliographystyle{alpha}
\bibliography{references}

\appendix

\newpage
\section{Supplementary material for  numerical results (Section~\ref{sec:numerical})}\label{app:numerical}
% !TEX root = main.tex
\subsection{Estimation}\label{app:emp-estimation}
The outcomes, $w_{t,j}$ for case $t$ at location $j$, are estimated according to the methodology proposed in \cite{bansak2018improving}. Thus, we refer the reader to \cite{bansak2018improving} for more details, but highlight the key information in what follows.

First, separate machine learning (ML) models are trained for each location $j$ using historical administrative data for refugees that were previously assigned to location $j$. The covariates used in the ML models consist of individual-level features such as country of origin, age, language skills, education, family size, sex, etc. When a new case $t$ arrives, the predicted outcome vector $(w_{t,1},...,w_{t,M})$ is generated by first applying each of the $M$ machine learning models to every individual within case $t$. Then, the individual-level predictions are aggregated to a case-level prediction for every location by applying a particular mapping. For example, $w_{t,j}$ may be the average predicted outcome among all individuals in case $t$. Alternatively, when the outcomes are binary (such as whether an individual has found employment), $w_{t,j}$ could be the probability of at least one individual in case $t$ achieving a positive outcome. This paper uses the latter mapping.

In what follows, we provide structural properties of the estimates $w_{t,j}$ in both the US and Netherlands contexts, as these properties can help provide context for the numerical results shown in Section \ref{sec:numerical}.

For the US data, for fixed $t$, the average span---i.e., the difference between the maximum and minimum values---of $w_{t,j}$ is 0.67 with standard deviation 0.12, and for fixed $j$ the average span of $w_{t,j}$ is 0.56 with standard deviation 0.2. The span of average probability of employment across locations is 0.6. For the NL data, for fixed $t$, the average span of $w_{t,j}$ is 0.43 with standard deviation 0.22, and for fixed $j$ the average span of $w_{t,j}$ is 0.72 with standard deviation 0.18. Across locations, the span of the average probability of employment is 0.3.

There is also correlation between the predicted outcomes for cases within the same group. This is because all of the group definitions in the paper (country of original, age, and education) are covariates used to in the ML models. For the US data (where groups are defined by country of origin), the average correlation within a group is 0.82, compared to the average correlation across the entire dataset of 0.76. For the NL data with groups defined by age, the average correlation within a group is 0.52 (with standard deviation .03 across groups) compared  to 0.46 overall. When groups are defined by education, the average correlation within groups is 0.51 with standard deviation 0.04 across groups.

Figure \ref{fig:slack} shows the slackness of each fairness rule (corresponding to Definition \ref{def:slackness}) considered, across all groups used in the empirical analysis. For the Random and Proportional Optimized fairness rules, slackness appears to be bounded away from zero as expected. However, for the MaxMin fairness rule, slackness is very close to zero.
%\wwcomment{Can we update the figure titles' to ``xxx Slackness''? Also update ``Maxmin'' to ``MaxMin''. ``Maxmin'' seems to be a consistent typo in all figures.}
\begin{figure}
    \centering
    \includegraphics[width=\linewidth]{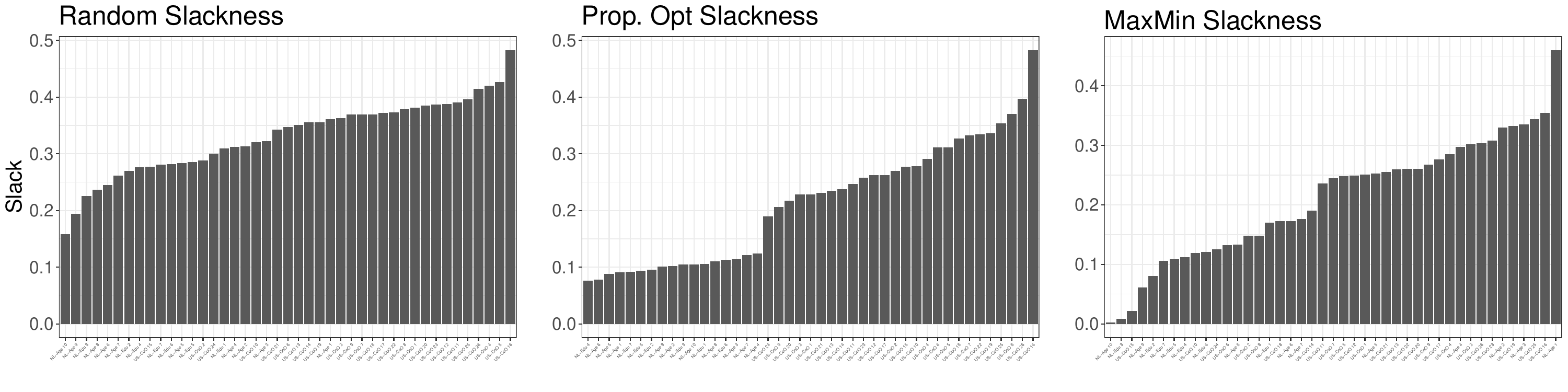}
    \caption{Slackness for each fairness rule, shown at the group-level. Specifically, each bar corresponds to a value of $\epsilon_{g,\mathcal{F}}(\mathcal{P})$ where $\mathcal{P}$ is the true distribution of 2016 arrivals.}
    \label{fig:slack}
\end{figure}

\subsection{Algorithmic set-up} \label{app:emp-setup}
We now describe the details of the implementations of all online methods---\textsc{BP}, \textsc{ABP}, and \textsc{CBP}--- used in Section \ref{sec:numerical}.

\paragraph{\textsc{ABP} and \textsc{CBP}.}
As described in Section \ref{sec:bp}, the parameters $\boldsymbol{\lambda}^*$ and $\boldsymbol{\mu}^*$ used in \textsc{ABP} and \textsc{CBP} (Algorithms \ref{algo:bpc} and \ref{algo:conservative}) are found by solving
$\min_{\boldsymbol{\lambda}\geq 0, \boldsymbol{\mu}\geq 0} \expect{L(\bolds{\mu},\bolds{\lambda})}$ 
%\wwcomment{$\geq$ instead of $>$ for the constraint?}
where
\begin{equation*}
\begin{aligned}
L(\boldsymbol{\mu},\boldsymbol{\lambda})&=\sum_{t=1}^T \max_{j \in \Sloc} ((1+\lambda_{g(t)})w_{t,j} - \mu_j) + \sum_{j \in \Sloc} \mu_j s_j -\sum_{g \in \Sgro} \lambda_g O_gN(g,T).
\end{aligned}
\end{equation*}

Let the superscript $k$ index a sample path of arrivals. The minimization problem is approximated by solving the following sample-average approximation across $K$ sample paths: 
%\wwcomment{again, $\geq$ instead of $>$?}

\begin{equation}\label{prob:saa}
\begin{aligned}
\min_{\boldsymbol{\lambda}\geq 0, \boldsymbol{\mu}\geq 0} &\sum_{k=1}^K \left( \sum_{t=1}^T z^k_t + \sum_{j\in\Sloc} \mu_j s_j -\sum_{g \in \Sgro} \lambda_g O^k_gN^k(g,T)\right) \\
\text{s.t.  } &  z^k_t\geq ((1+\lambda_{g^k(t)})w^k_{t,j} - \mu_j) \; \forall j\in\Sloc, \; t\in \{1,...,T\}
\end{aligned}
\end{equation}
where $O^k_g$ and $N^k(g,T)$ can be pre-computed offline for each sample path of arrivals and then treated as constants in Problem \ref{prob:saa}. Using the auxiliary variables $z_t^k$, Problem \ref{prob:saa} is a linear program and thus computationally tractable. 

After calculating $\boldsymbol{\lambda}^*$ and $\boldsymbol{\mu}^*$, Algorithm \ref{algo:bpc} can be implemented as described in Section \ref{sec:bp}. For Algorithm \ref{algo:conservative}, we similarly compute $\boldsymbol{\lambda}^*$ and $\boldsymbol{\mu}^*$.

Additionally, for Algorithm \ref{algo:conservative}, the value of $\Psi(g,t,\beta)$ is set in the following data-driven way: First, we find a confidence interval $[l_g, u_g]$ for the number of remaining arrivals of each group $g$; this uses the fact that the number of remaining arrivals of each group $g$ follows a Binomial distribution with mean $(T-t)\bar{p}_g$. 
Specifically,  we set $l_g = (T-t)  \bar{p}_g - Z_g(0.9)$ 
and $u_g = (T-t)\bar{p}_g + Z_g(0.9)$ where $Z_g(x)$ denotes the inverse CDF of the group-specific Binomial distribution. 
Intuitively, for each $n\in[l_g,u_g]$, we want to know: \emph{If there are $n$ remaining arrivals in group~$g$ and they are all assigned to their highest-score (greedy) location, what is the chance this group will not meet its minimum requirement?}
For each $n\in[l_g,u_g]$, we independently sample $n$ group $g$ arrivals $K$ times. For each of the $K$ length-$n$ vectors, $\{\mathbf{w}_{\tau}^{k(g)}\}_{\tau=1,...,n}$, we calculate $s_g^k := \sum_{\tau=1}^n(\max_j w^{k(g)}_{\tau j} - \bar{O}_g)$ for $k=1,...,K$, where $\bar{O}_g$ is a 95th percentile of $O_g(\omega)$ across 100 random instances. Finally, we calculate the 10th percentile of $[s_g^1,...,s_g^K]$ and use this value, $q(n,g)$, to set $\Psi(g,t,\beta)=\min_{n\in [l_g,u_g]}q(n,g)$. This effectively lower bounds $\sum_{\tau\in \mathcal{A}(g,T)\setminus \mathcal{A}(g,t)}(\max_j w_{\tau j} - \bar{O}_g)$ as desired.

\paragraph{\textsc{BP}.}
To determine the parameters $\boldsymbol{\mu}^*$ used in \textsc{BP}, we compute a sample average of the optimal dual values of the continuous relaxation of \textsc{OPT}. This is also equivalent to the approach described above for \textsc{ABP} without a fairness constraint (\emph{e.g.} achieved by setting $O_g=0$ for all groups $g$).

For ABP, CBP, and BP, if the first-best (desired) location has no remaining capacity, we assign the current case to the next-best location that has remaining capacity. For BP and ABP, this is found by maximizing $w_{t,j} - \mu_j$ or $(1+\lambda_{g(t)})w_{t,j} - \mu_j$, respectively, over locations with positive remaining capacity. For CBP, we follow the same next-best assignment as ABP if condition \eqref{eq:cond-predict} holds. If \eqref{eq:cond-predict} does not hold, we greedily select a location with positive remaining capacity, i.e. the location with positive remaining capacity that maximizes $w_{t,j}$

\subsection{Additional figures for Section \ref{sec:emp-fairness-algos}}\label{app:emp-fairness-algos}

Figure \ref{fig:opt-performance} shows $\textsc{FR}_g(\bomega)$ for each group under \textsc{OPT} using our three fairness rules. This can be compared to Figure \ref{fig:benchmark-fairness-performance} which shows $\textsc{FR}_g(\bomega)$ for \textsc{Rand} and \textsc{BP}.

\begin{figure}[h!]
    \centering
    \includegraphics[width = \textwidth]{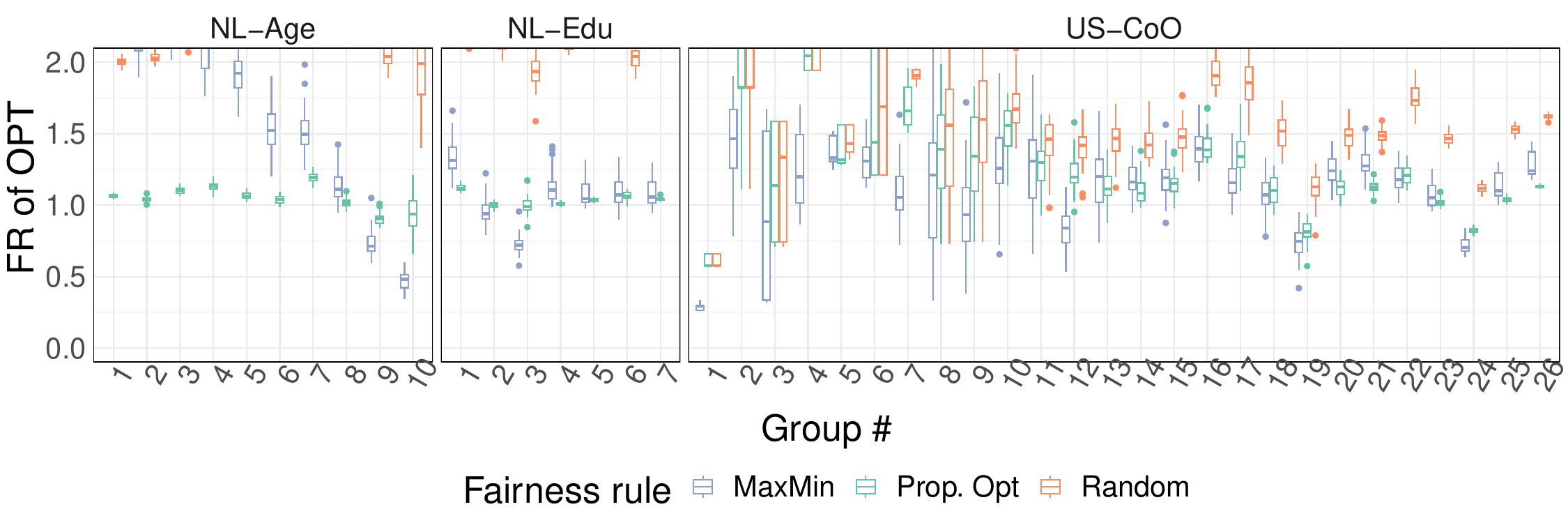}
    \caption{$\textsc{FR}_g(\bomega)$ achieved by \textsc{OPT} with respect to the three fairness targets for each scenario.}
    \label{fig:opt-performance}
\end{figure}

\subsection{Additional figures and tables for Section \ref{sec:emp-algos}} \label{app:emp-algos}

Figure \ref{fig:abp-cbp-performance} shows group-level results for the ABP and CBP algorithms, analogous to Figure \ref{fig:benchmark-fairness-performance} that was shown for \textsc{Rand} and \textsc{BP}. 

\begin{figure}[h!]
    \centering
    \includegraphics[width = \textwidth]{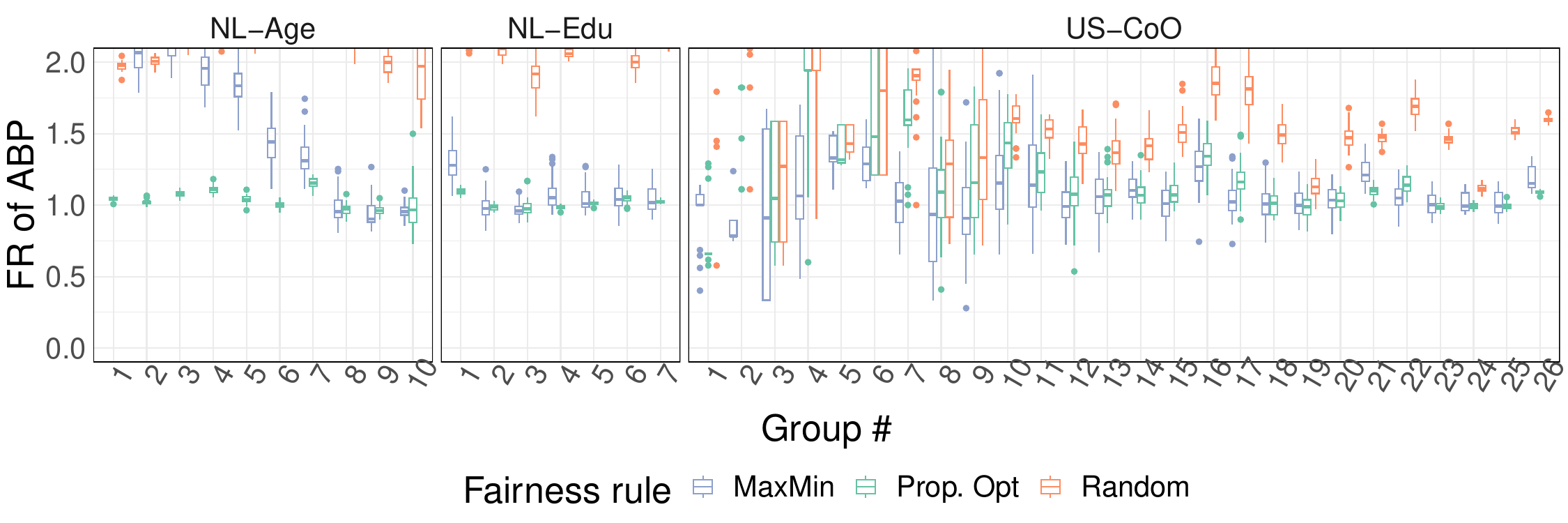}
    \includegraphics[width = \textwidth]{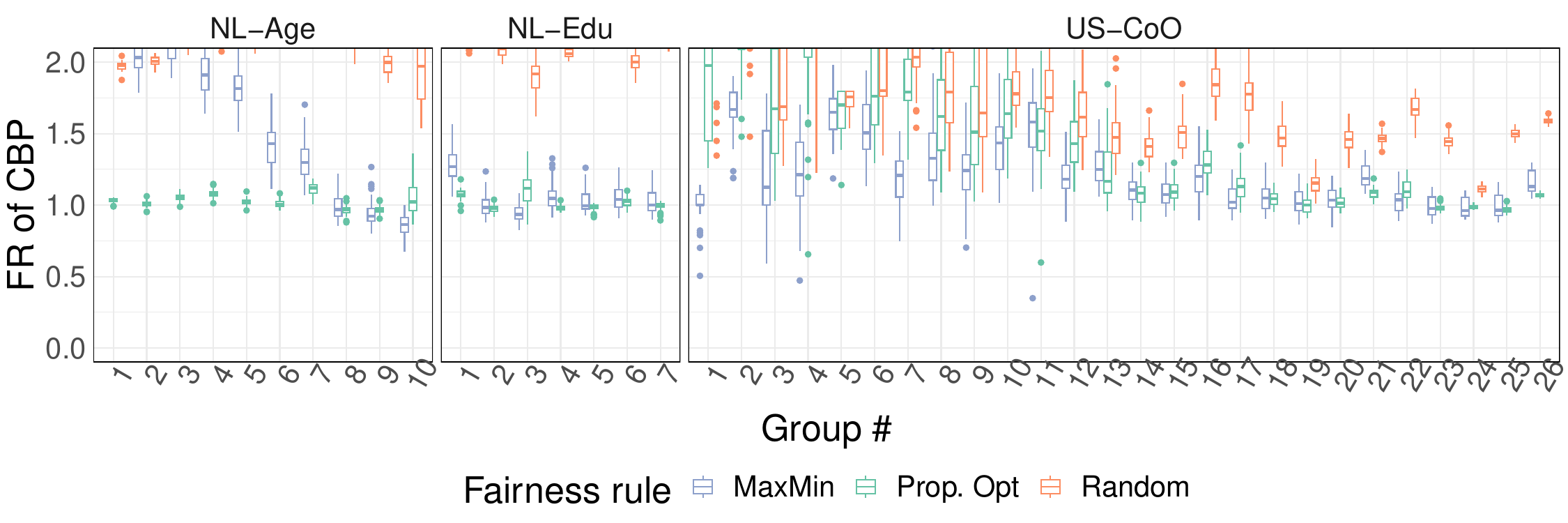}
    \caption{$\textsc{FR}_g(\bomega)$ achieved by ABP and CBP with under our three fairness rules.}
    \label{fig:abp-cbp-performance}
\end{figure}

\begin{table}[ht]
\centering
\begin{tabular}{rrrrrrrrrrrrr}
  \hline
  && \multicolumn{3}{c}{\% positive regret} && \multicolumn{3}{c}{Avg fairness ratio}&& \multicolumn{3}{c}{Min fairness ratio}  \\
  \cline{3-5} \cline{7-9} \cline{11-13}
 Group & Size & \textsc{BP} &  \textsc{ABP} & \textsc{CBP} && \textsc{BP} &  \textsc{ABP} & \textsc{CBP} && \textsc{BP} &  \textsc{ABP} & \textsc{CBP}  \\ 
  \hline
1 & small & 1.00 & 0.91 & 0.00 && 0.59 & 0.71 & 1.85 && 0.58 & 0.58 & 1.26 \\ 
  2 & small & 0.00 & 0.00 & 0.00 && 2.30 & 1.77 & 2.17 && 1.11 & 1.11 & 1.48 \\ 
  3 & small & 0.41 & 0.39 & 0.00 && 1.11 & 1.12 & 1.68 && 0.58 & 0.58 & 1.03 \\ 
  4 & small & 0.02 & 0.02 & 0.02 && 2.13 & 2.13 & 2.16 && 0.90 & 0.60 & 0.66 \\ 
  5 & small & 0.00 & 0.00 & 0.00 && 1.39 & 1.39 & 1.65 && 1.28 & 1.28 & 1.14 \\ 
  6 & small & 0.00 & 0.00 & 0.00 && 1.83 & 1.84 & 2.01 && 1.21 & 1.21 & 1.29 \\ 
  7 & small & 0.00 & 0.00 & 0.00 && 1.66 & 1.64 & 1.88 && 1.26 & 1.00 & 1.32 \\ 
  8 & small & 0.41 & 0.43 & 0.00 && 1.11 & 1.09 & 1.67 && 0.68 & 0.41 & 1.09 \\ 
  9 & small & 0.54 & 0.35 & 0.00 && 1.12 & 1.26 & 1.60 && 0.53 & 0.65 & 1.03 \\ 
  10 & small & 0.00 & 0.02 & 0.00 && 1.44 & 1.43 & 1.68 && 1.11 & 0.87 & 1.18 \\ 
  11 & small & 0.00 & 0.15 & 0.02 && 1.39 & 1.25 & 1.58 && 1.18 & 0.96 & 0.60 \\ 
  12 & small & 0.00 & 0.28 & 0.00 && 1.24 & 1.08 & 1.46 && 1.02 & 0.54 & 1.09 \\ 
  13 & small & 0.26 & 0.28 & 0.04 && 1.07 & 1.07 & 1.23 && 0.88 & 0.87 & 0.96 \\ 
  14 & medium & 0.14 & 0.12 & 0.18 && 1.08 & 1.08 & 1.08 && 0.93 & 0.90 & 0.88 \\ 
  15 & medium & 0.00 & 0.14 & 0.06 && 1.17 & 1.08 & 1.10 && 1.01 & 0.95 & 0.96 \\ 
  16 & medium & 0.00 & 0.00 & 0.00 && 1.36 & 1.36 & 1.30 && 1.12 & 1.07 & 1.07 \\ 
  17 & medium & 0.00 & 0.10 & 0.08 && 1.34 & 1.18 & 1.14 && 1.02 & 0.90 & 0.95 \\ 
  18 & medium & 0.04 & 0.48 & 0.24 && 1.12 & 1.01 & 1.05 && 0.96 & 0.90 & 0.95 \\ 
  19 & medium & 1.00 & 0.54 & 0.46 && 0.83 & 0.98 & 1.00 && 0.72 & 0.82 & 0.91 \\ 
  20 & medium & 0.02 & 0.36 & 0.44 && 1.12 & 1.03 & 1.02 && 0.98 & 0.89 & 0.94 \\ 
  21 & medium & 0.00 & 0.00 & 0.00 && 1.11 & 1.10 & 1.08 && 1.02 & 1.01 & 1.01 \\ 
  22 & medium & 0.00 & 0.00 & 0.10 && 1.20 & 1.15 & 1.10 && 1.07 & 1.02 & 0.98 \\ 
  23 & large & 0.34 & 0.66 & 0.86 && 1.01 & 0.99 & 0.98 && 0.96 & 0.94 & 0.94 \\ 
  24 & large & 1.00 & 0.56 & 0.82 && 0.82 & 0.99 & 0.99 && 0.78 & 0.95 & 0.96 \\ 
  25 & large & 0.06 & 0.74 & 0.94 &&1.03 & 0.99 & 0.97 && 0.99 & 0.95 & 0.93 \\ 
  26 & large & 0.00 & 0.00 & 0.00 && 1.12 & 1.09 & 1.07 && 1.09 & 1.06 & 1.05 \\ 
   \hline
\end{tabular}
\caption{Percentage of instances (out of 50) with a fairness ratio below 1, average fairness ratio, and minimum fairness ratio. ``Small'' groups are those with fewer than 10 arrivals. ``Medium'' groups are those with between 10-100 arrivals, and ``large'' groups have over 100 arrivals. Note: This table is based on US data using the Proportional Optimized fairness rule.}
    \label{tab:group-CoO-regret-v-group}
\end{table}

\subsection{Evaluation beyond i.i.d. based on a single historical trace (Section~\ref{sec:emp-algos})} \label{app:real-world}
In this section we demonstrate the performance of our algorithms in a setting that most closely resembles the real world. In reality, non-stationarities arise due to changing arrival patterns of refugees with certain characteristics (e.g., one country of origin might have a spike of arrivals in a given month), and due to changing economic conditions in the host country. Our theoretical guarantees extend to such non-stationary arrivals; see the discussions in Appendix~\ref{app:nonstationary}.

We treat the 2016 arrivals---in the order that they actually arrived---as the test cohort. Therefore, the arrival sequence no longer satisfies an i.i.d. assumption. All parameters used in the algorithms (e.g., $\boldsymbol{\lambda}^*$, $\boldsymbol{\mu}^*$, and $\mathbb{E}[O_g]$) are determined using the 2015 cohort of arrivals. When running the algorithms, the only knowledge that the algorithm has about the 2016 arrivals is the total number of 2016 arrivals and the capacity vector for 2016. These assumptions reflect the information known in reality.

There is one subtle complexity. Suppose there are certain groups which are present in the 2016 arrival cohort that were not present in the 2015 arrival cohort. This means that we cannot estimate $\lambda_g^*$ for a group $g$ in this situation. This is not an issue in the \textsc{NL-Edu} or \textsc{NL-Age} scenarios, but is an issue in the \textsc{US-CoO} scenario. We propose handling the issue in the following way. If there is a large group $g$ that arrives in 2016 and not in 2015, it is likely that the policymaker will have prior knowledge about this group. For example, there may have been a policy that prohibited this group from arriving in 2015, but policy changes allow for their arrival in 2016. With this prior knowledge, we could use data-driven approaches to estimate $\lambda_g^*$ for this group. For example, we could generate employment probability predictions for simulated arrivals from this group and add them to the 2015 arrival cohort before computing the parameters of the algorithms.

Now consider a second situation in which refugees from a group $g$ arrive in 2016, but had no arrivals in 2015, and the policymakers did not have prior knowledge about this change. In this case, it is likely that the group is quite small. For example, if a particular group has, on average, only a handful of arrivals, it is likely there are certain years where the group will have zero arrivals, and policymakers might not pay special attention to this. In this situation, we propose either a) using older data that includes arrivals from this group when determining the parameters of the algorithm or b) assigning arrivals to this group to their greedy locations. Because these groups are likely to be quite small, this will not greatly impact the performance of the algorithm. For the numerical results of this section, we take this approach, and assign any arrivals from a group that did not appear in 2015 to their greedy location (among those with remaining capacity).

\begin{figure}[h!]
    \centering
    \includegraphics[width = \textwidth]{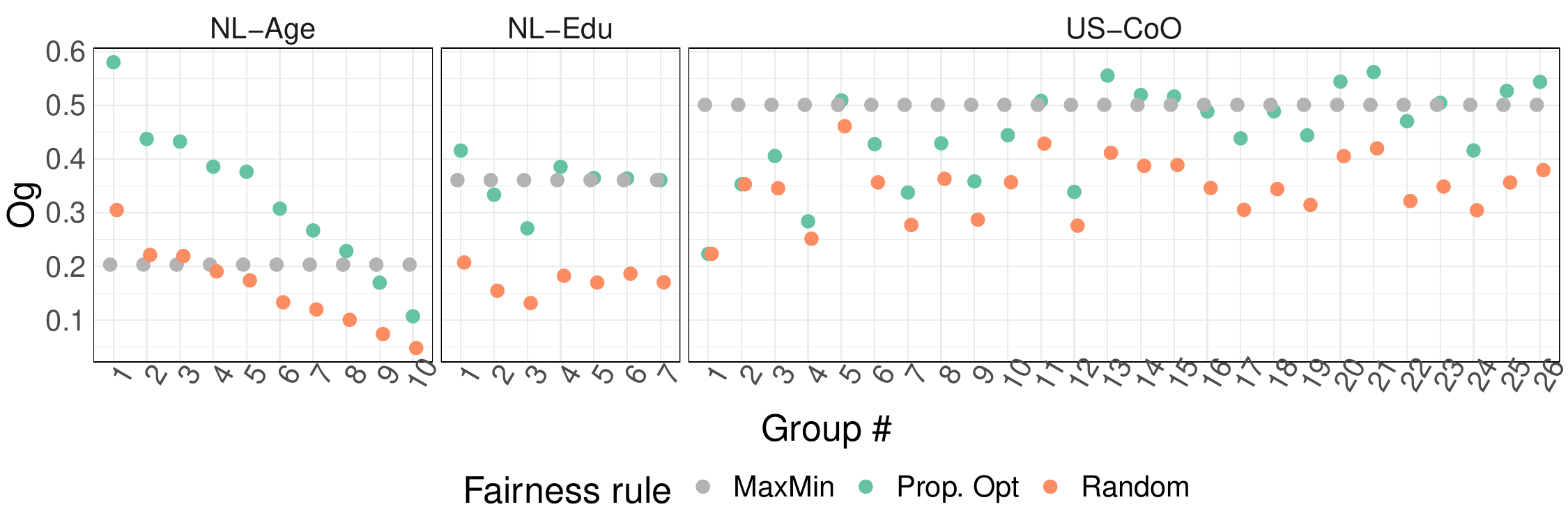}
    \caption{Minimum requirements for the 2016 arrival cohort.}
    \label{fig:fairness-targets-real}
\end{figure}

\begin{table}[ht]
\centering
\begin{tabular}{llllllllllll}
  \hline
  & \multicolumn{3}{c}{NL-Edu} && \multicolumn{3}{c}{NL-Age} && \multicolumn{3}{c}{US-CoO}\\
  \cline{2-4} \cline{6-8} \cline{10-12}
Rule: & \footnotesize MaxMin & \footnotesize PrOpt & \footnotesize Random && \footnotesize MaxMin & \footnotesize PrOpt & \footnotesize Random && \footnotesize MaxMin & \footnotesize PrOpt & \footnotesize Random \\ 
  \hline
\textsc{Rand} & 0.46 & 0.46 & 0.46 && 0.46 & 0.46 & 0.46 && 0.67 & 0.67 & 0.67 \\ 
\textsc{CBP} & 0.94 & 0.94 & 0.96 && 0.92 & 0.94 & 0.96 && 0.88 & 0.91 & 0.92 \\ 
\textsc{ABP} & 0.94 & 0.96 & 0.96 && 0.96 & 0.96 & 0.96 && 0.91 & 0.92 & 0.92 \\ 
\textsc{BP} & 0.96 & 0.96 & 0.96 && 0.96 & 0.96 & 0.96 && 0.92 & 0.92 & 0.92 \\ 
\textsc{Offline}$_{\mathcal{F}}$ & 0.98 & 1.00 & 1.00 && 0.99 & 1.00 & 1.00 && 0.98 & 1.00 & 1.00 \\
   \hline
\end{tabular}
\caption{Efficiency of methods in the ``real-world'' scenario.}
\end{table}

\begin{figure}
    \centering
    \includegraphics[width=\textwidth]{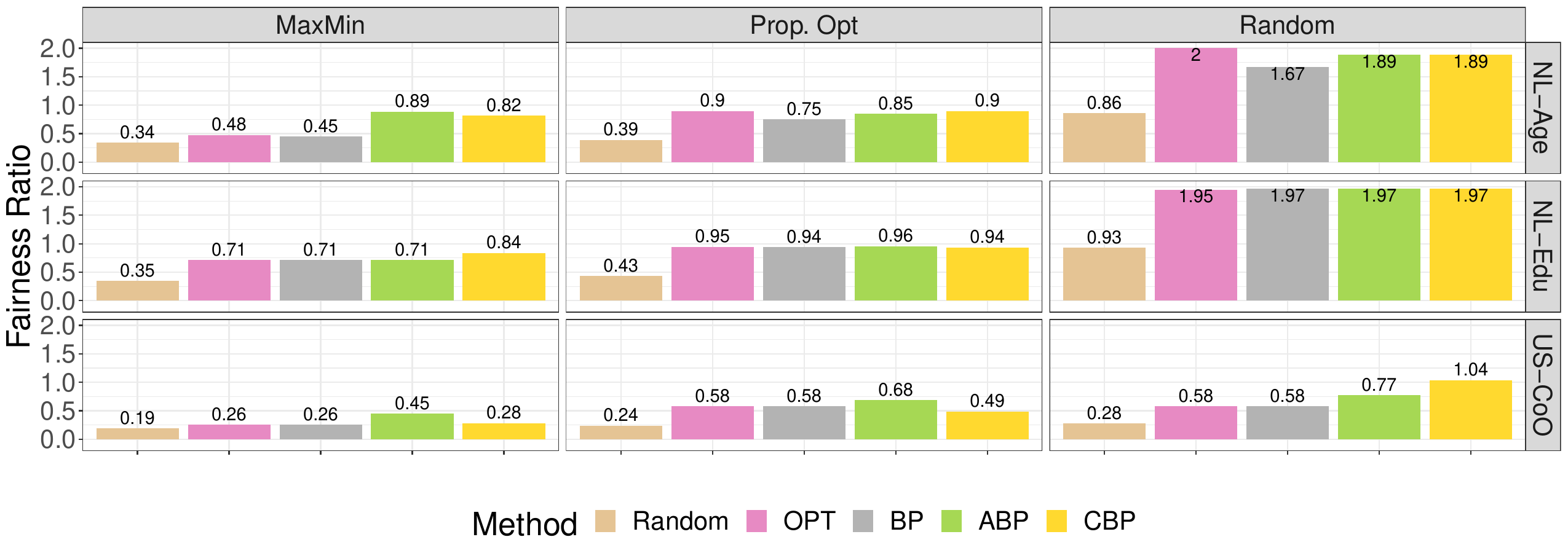}
    \caption{Fairness ratio (aggregated by min) obtained in the ``real-world'' scenario. Note: Unlike ABP and CBP, the Random, OPT, and BP algorithms are independent of the chosen fairness rule.}    \label{fig:avg_regret_nonstationary-max}
\end{figure}
\clearpage
\section{Infeasibility of Resolving Methods (Section~\ref{ssec:rel_work})}\label{app:intro}
% !TEX root = main.tex
It is natural to wonder whether resolving techniques from the literature on online resource allocation are immediately applicable to our setting with group fairness constraints. In this appendix, we demonstrate through an example that vanilla applications of these techniques may fail if they require solving infeasible optimization problems. We consider the \textsc{Bayes Selector} algorithm \cite[algorithm 1]{vera2021bayesian}. The $\textsc{MinDiscord}$ algorithm of \cite{elisabethpaper}, which was developed in the context of refugee assignment, follows a similar idea to assign refugees but without group fairness constraints.

Upon the arrival of case $t$, the DM observes the first $t$ arrivals $\hat{\bolds{\theta}}_1,\ldots,\hat{\bolds{\theta}}_t$ and has assigned the initial $t-1$ arrivals whose assignment vectors are given by $\hat{\boldsymbol{z}}_1,\ldots,\hat{\boldsymbol{z}}_{t-1}$. Then the DM samples a list of $K$ future sample paths $\{\bomega^k_t = (\bolds{\theta}^k_{t+1},\ldots,\bolds{\theta}^k_T)\}_{k \in [K]}$ with $K$ being a hyper-parameter. The sampling is independent across paths and each feature vector $\btheta^k_{\tau}$ is sampled from the known probability distribution $\mathcal{P}_{\tau}$ for $\tau > t$. For each concatenated sample path $\hat{\bomega}^k_t = (\hat{\bolds{\theta}}_1,\ldots,\hat{\bolds{\theta}}_t,\bolds{\theta}^k_{t+1},\ldots,\bolds{\theta}^k_T)$, the DM computes the optimal fractional assignment $\boldsymbol{\tilde{z}}^k$ to the offline problem $O^\star_{\frule}(\hat{\bomega}^k_t)$ in \eqref{eq:outcome-benchmark}, with an additional constraint that assignments of the first $t-1$ cases are consistent, i.e., $\boldsymbol{\tilde{z}}^k_{\tau} = \hat{\boldsymbol{z}}_{\tau}$ for $\tau < t$. Case $t$ is assigned to the location with the highest probability among these optimal assignments, that is, location $J(t) \in \arg\max_{j \in \set{M}} \sum_{k \in [K]} \tilde{z}^k_{t,j}$. 

Since the group fairness constraints depend on sample paths, the offline problem $O^\star_{\frule}(\hat{\bomega}_t^k)$ may be infeasible because conducted assignment decisions are induced by other sample paths and no future assignment decisions can satisfy the group fairness constraints for this sample path. This is not an issue when there are only capacity constraints as these constraints do not vary across sample paths. The following proposition formalizes this intuition.
\begin{proposition}
For any even $T \geq 1000$, there exists a setting such that with probability at least $0.0261$, for the second arrival, \textsc{Bayes Selector} with any hyper-parameter $K \geq 1$ must solve an infeasible offline problem during its run-time.
\end{proposition}

\begin{proof}
Suppose $T$ is an even number at least $1000$. Consider the following setting with $2$ locations $\{1,2\}$, and three groups $\{A,B,C\}$. Locations have capacity $s_1 = T/2, s_2 = T/2$. Group arrival probabilities are $p_A = 0.4,p_B=0.6-1/T,p_C=1/T$.  Scores of arrivals to locations are determined by their groups such that $w_{t,1}=1,w_{t,2}=0.5$ if $t$ is in group $A$; $w_{t,1}=0.5,w_{t,2}=0$ if $t$ is in group $B$; and $w_{t,1}=w_{t,2}=0$ if $t$ is in group $C$. The fairness rule is MaxMin (Example~\ref{ex:maxmin-fairness-rule}). 

Fix the hyper-parameter $K \geq 1$. Let us simulate the algorithm from case $1$. With probability $p_A \geq 0.4$, the first case is from group $A$; denote its feature by $\hat{\btheta}_1$. Consider the algorithm for case $1$. Denote the set of paths sampled by the DM by $\Omega_1 = \{\bomega_1^k=(\btheta^k_2,\ldots,\btheta^k_T)\}_{k \in [K]}$. Define $\set{S}_1$ as a set of sample paths from case $2$ to case $T$ such that there is a group $C$ arrival and the number of group $A$ arrivals is at most $0.5T - 1$, that is, 
$\set{S}_1 = \{ (\btheta_2,\ldots,\btheta_T)\colon \sum_{t=2}^T \indic{g(\btheta_t)=C} \geq 1,\sum_{t=2}^T \indic{g(\btheta_t)=A} \leq 0.5T-1\}$. For any $\bomega_1^k \in \set{S}_1$, consider the concatenated sample path $\hat{\bomega}_1^k = (\hat{\btheta}_1,\bomega_2^k)$. We know there is a group $C$ arrival and the number of group $A$ arrivals is at most $T / 2$. Since the fairness rule is MaxMin, having a group $C$ arrival removes fairness constraints for all groups as cases from group $C$ always have score zero. The optimal fractional assignment $\tilde{\bolds{z}}^k$ of $O^\star_{\mathrm{maxmin}}(\hat{\bomega}_1^k)$ then assigns all group $A$ arrivals to location $1$ to maximize the total score. Therefore, if $\bomega_1^k \in \set{S}_1$, we must have $\tilde{z}^k_{1,1} = 1$ since the first case is of group $A$.

Let $\set{E}_1$ be the event that more than half of the sample paths in $\Omega_1$ are in $\set{S}_1$. If $\set{E}_1$ holds, we have $\sum_{k =1}^K \tilde{z}_{1,1}^k > 0.5 > \sum_{k =1}^K \tilde{z}_{1,j}^k$ for $j = 2,3$ and thus the algorithm will assign the first case to location $1$. To lower bound the probability of $\set{E}_1$, we first upper bound the probability that $\bomega_1^k \not \in \set{E}_1$ for any $k$ by
\begin{align*}
\Pr\{\bomega_1^k \not \in \set{E}_1\} &\leq \Pr\left\{g(\btheta_t^k) \neq C,\forall t \geq 2\right\} + \Pr\left\{\sum_{t=2}^T \indic{g(\btheta_t^k)=A} \geq 0.5T\right\} \\
&\leq (1-p_C)^{T-1} + \exp\left(\frac{-2(0.1T)^2}{T}\right) \leq 0.4 + \exp(-20)\leq 0.41 
\end{align*}
which uses Hoeffding Inequality (Fact~\ref{fact:hoeffding}) and the fact that $p_C = 1/T$ and $T \geq 1000$. As a result, $\Pr\{\bomega_2^k \in \set{S}_1\} \geq 0.59$. Since sample paths are generated independently and any one sample path is more likely than not in $\set{S}_1$, we can lower bound the probability that more than half of sample paths of $\Omega_1$ are in $\set{S}_1$ by at least $0.25$ (the worst case is $K=2$).

Consider the second arrival, i.e., one with feature $\hat{\btheta}_2$. Condition on $\set{E}_1$ and the event that the second arrival is not from group $C$. 
We know $\Pr\{g(\hat{\btheta}_1)=A,\set{E}_1,g(\hat{\btheta}_2) \neq C\} = p_A(1-p_C)\Pr\{\Omega_1\} \geq 0.09$ since they are independent. Let $\Omega_2 = \{\bomega_2^k\}_{k \in [K]}$ be the paths sampled by the algorithm. Define $\set{S}_2$ to be the set of sample paths from case $3$ to case $T$ such that there is no group $C$ arrival and the number of group $B$ arrivals is at least $T / 2$, i.e., $\set{S}_2 = \{(\btheta_3,\ldots,\btheta_T)\colon \sum_{t=3}^T \indic{g(\btheta_t)=C}=0,\sum_{t=3}^T \indic{g(\btheta_t)=B} \geq 0.5T\}$. Suppose that for some $k$, $\bomega_2^k \in \set{S}_2$. Since there is no group $C$ arrival, the MaxMin fairness rule requires assigning as many group-$B$ arrivals as possible to location $1$ to increase the average score of group $B$, which is always no greater than that of group $A$. As a result, since the number of group $B$ arrivals is at least $T/2$ for sample paths in $\set{S}_2$, any \textbf{feasible} assignment to the offline problem $O_{\mathrm{maxmin}}^\star(\hat{\bomega}_2^k)$ must assign all location $1$ to group $B$ arrivals to satisfy its fairness constraint. However, as the first case of group $A$ is assigned to location $1$ and this decision is irrevocable, the offline problem $O_{\mathrm{maxmin}}^\star(\hat{\bomega}_2^k)$ must be infeasible when $\bomega_2^k \in \set{S}_2$.

We next lower bound the probability that $\bomega_2^1 \in \set{S}_2$ by
\begin{align*}
\Pr\{\bomega_2^1 \in \set{S}_2\} &\geq \Pr\left\{g(\btheta_t^1) \neq C,\forall t \geq 3\right\} - \Pr\left\{\sum_{t=3}^T \indic{g(\btheta_t^k)=B} < 0.5T\right\} \\
&\geq (1-p_C)^{T-1} - \exp\left(\frac{-2(0.08T)^2}{T}\right) \geq 0.3 - \exp(-10) \geq 0.29.
\end{align*}
As a result, condition on the event that $\{g(\hat{\btheta}_1)=A\} \cap \set{E}_1 \cap \{g(\hat{\btheta}_2) \neq C\} \cap \{\bomega_2^1 \in \set{S}_2\}$, which happens with probability at least $0.09 * 0.29 = 0.0261$ because of independence, $\textbf{Bayes Selector}$ with any hyper-parameter $K \geq 1$ must solve an infeasible optimization problem. 
\end{proof}

\section{Extension to Non-Stationary Feature Distributions (Section~\ref{ssec:formal_model})}\label{app:nonstationary}
% !TEX root = main.tex
Our results extend to a non-stationary setting where groups' arrival probabilities can change over time, and the change is not too drastic across the entire horizon. Specifically, in this updated model, refugee case $t$ has a feature $\btheta_t \in \Theta$ with a distribution $\set{P}_t$ for any $t \in [T]$. We still make the assumption that the feature distribution $\set{P}_t$ maps refugee features to continuously distributed employment scores. For the distribution process $\{\set{P}_t\}$, we let $p_g(t) = p_g(\set{P}_t)$ be the probability that refugee case $t$ is from group $g \in \Sgro$ and define the average arrival probability of group $g$ across the time horizon as $\bar{p}_g(\{\set{P}_t\}) \triangleq \frac{1}{T}\sum_{t\in[T]} p_g(t)$.  Although we allow $p_g(t)$ to evolve over time, we assume that the score $w_{t,j}$ is drawn identically across time conditioned on a group; this is formalized below.
\begin{assumption}
For $g \in \Sgro, t,t' \in [T]$, conditioned on $g(t) = g(t')=g$, we have $(w_{t,j})_{j \in \set{M}}$ has the same distribution as $(w_{t',j})_{j \in \set{M}}$.
\end{assumption}
Under this assumption, the change in group arrival probabilities determines the degree of non-stationarity of the arrival process. To quantify such non-stationarity, we define the non-stationarity parameter $U$ by the maximum multiplicative gap between a group's arrival probability and its average arrival probability over the time horizon, i.e., 
\[
U(\{\set{P}_t\}) \triangleq \min\left\{u \geq 1 \colon p_g(\set{P}_\tau) / \bar{p}_g(\{\set{P}_t\}) \in [1 / u, u], \forall g \in \Sgro, \tau \in [T]\right\}.
\]
The key takeaway from this section is that all our global regret and ex-post $g-$regret guarantees in Lemmas~\ref{lem:bid-price-global}, \ref{lem:bid-price-group}, \ref{lem:cons-global} and \ref{lem:cons-group} remain the same except for new multiplicative factors that scale polynomially in $U$. 

Numerically, to provide 
a sense for the level of non-stationarity, we estimated $U$ for the real-world setting discussed in Appendix~\ref{app:real-world} by splitting 2015 and 2016 into quarters. In the Netherlands, we found that $U \leq 1.7$ in both years with one exception: with groups defined by Education, there is one group in one quarter of 2016 that causes $U$ to be equal to $7.5$. For the US, there are many small groups that may have no arrivals in a given quarter and thus $U$ cannot not be estimated. Among the groups where it could be estimated, we find that $U\leq 4.5$ both years.

\section{Fairness Rules: Objectives and Assumptions (Sections~\ref{ssec:model_objectives}, \ref{sec:fairness-rule})}\label{app:model}
% !TEX root = main.tex
\subsection{Measurability of fairness rule examples (Section~\ref{ssec:fairness_rule_examples})}\label{app:measurable}
For a fairness rule $\frule$, our result requires verification of measurability of both the requirement $O_{g,\frule}(\bomega)$ and the associated optimal score $O^\star_{\frule}(\bomega)$. In the sample space, a sample path $\bomega$ consists of a discrete sequence $\bolds{gr}(\bomega) = (g(t,\bomega)) \in [G]^T$ and a real vector sequence $\bolds{w}(\bomega)=(w_{t,j}(\bomega)) \in [0,1]^{T \times M}$. Since there are finitely many group arrival sequences, establishing measurability only requires verifying measurability of $O_{g,\frule}(\bomega)$ and $O^{\star}_{\frule}(\bomega)$ when the group arrival sequence $\bolds{gr}(\bomega)$ is fixed. As a result, let us consider the sample space $\Omega_{\bolds{gr}}$ where the group arrival sequence $\bolds{gr}(\bomega)$ is fixed to $\bolds{gr}$. In this scenario, both the required score for a fixed group $g$, $O_{g,\frule}(\bomega)$, and the optimal score $O^\star_{\frule}(\bomega)$ become functions of the score sequence $\bolds{w}$ and thus are functions over the $T \times M$ Euclidean space. For ease of notations, let $O_{g,\frule}^{\bolds{gr}}(\bolds{w}) = O_{g,\frule}((\bolds{gr},\bolds{w}))$,  $O^{\star,\bolds{gr}}_{\frule}(\bolds{w})=O^{\star}_{\frule}((\bolds{gr},\bolds{w}))$ and denote $\set{D} = (0,1)^{T \times M}$.

It remains to show that $O_{g,\frule}^{\bolds{gr}}(\bolds{w})$ and $O_{\frule}^{\star,\bolds{gr}}(\bolds{w})$ are Lebesgue measurable functions over $\set{D}$ for the fairness rules under consideration. Note that although the original domain is $[0,1]^{T \times M}$ but not $\set{D}$, we are fine to only consider measurability over $\set{D}$ since $[0,1]^{T \times M} \setminus \set{D}$ is of measure $0$. To establish the measurability, we utilize the following result which is a direct application of Theorem 1.1 in  \cite{martin1975continuity}. 
\begin{fact}\label{fact:continuity-lp}
For $\bolds{w} \in \set{D}$, consider the linear program $M(\bolds{w}) = \max_{\bolds{x} \in \mathbbm{R}^{T \times M}} \bolds{c}'(\bolds{w})\bolds{x}$ subject to $\bolds{A}(\bolds{w})\bolds{x} \leq \bolds{b}(\bolds{w})$ and $\bolds{x} \geq 0$. Assume that for every $\bolds{w} \in \set{D}$, 1) $\bolds{c}(\bolds{w}),\bolds{A}(\bolds{w}),\bolds{b}(\bolds{w})$ are continuous; 2) $M(\bolds{w})$ exists and is finite 3) the set of optimal solutions is bounded. Then we have $M(\bolds{w})$ is upper semicontinuous at every $\bolds{w}$. If the set of optimal solutions for the dual is also bounded, then $M(\bolds{w})$ is continuous.
\end{fact}
Based on Fact~\ref{fact:continuity-lp}, the following lemma shows that $O^{\star,\bolds{gr}}_{\frule}(\bolds{w})$ is measurable as long as $O_{g,\frule}^{\bolds{gr}}(\bolds{w})$ is continuous over $\set{D}$. 
\begin{lemma}\label{lem:measurable}
For any group arrival sequence $\bolds{gr}$, suppose $O_{g,\frule}^{\bolds{gr}}(\bolds{w})$ is continuous for every group $g \in \Sgro$ and induces feasible \eqref{eq:outcome-benchmark} over $\bolds{w} \in \set{D}$. Then $O^{\star,\bolds{gr}}_{\frule}(\bolds{w})$ is measurable over $\set{D}$.
\end{lemma}
\begin{proof}
Fix $\bolds{w} \in \set{D}$. Recall that $O^{\star,\bolds{gr}}_{\frule}(\bolds{w})$ is defined by \eqref{eq:outcome-benchmark}. In the form of Fact~\ref{fact:continuity-lp}, it is equivalent to view $\bolds{c}(\bolds{w}) = \bolds{w}$ as a vector and $\bolds{x}$ as the vectorized assignment decision variables $\bolds{z}$ in \eqref{eq:outcome-benchmark}. Let $n_g$ be the number of cases for group $g$ under the sequence $\bolds{gr}$. For constraints, $\bolds{b}(\bolds{w})$ is a $(M+G+T)$-dimensional vector with $\bolds{b}(\bolds{w})_j = s_j$ for $j \in \Sloc$, $\bolds{b}(\bolds{w})_g = -n_gO_{g,\frule}^{\bolds{gr}}(\bolds{w})$ for $g \in \Sgro$ and $\bolds{b}(\bolds{w})_t = 1$ for $t \in [T]$. Accordingly, \eqref{eq:outcome-benchmark} naturally induces the coefficient matrix $\bolds{A}(\bolds{w})$ which is linear in $\bolds{w}$. The objective function $M(\bolds{w})$ is exactly $TO^{\star,\bolds{gr}}_{\frule}(\bolds{w})$. We next verify the conditions in Fact~\ref{fact:continuity-lp}. First, $\bolds{c}(\bolds{w}), \bolds{A}(\bolds{w})$ are both linear in $\bolds{w}$ and thus continuous. In addition, $\bolds{b}(\bolds{w})$ is continuous since by assumption $O_{g,\frule}^{\bolds{gr}}(\bolds{w})$ is continuous for every $g \in \Sgro$. Second, by assumption \eqref{eq:outcome-benchmark} is feasible and thus $M(\bolds{w}) = TO^{\star,\bolds{gr}}_{\frule}(\bolds{w})$ exists. Since the total score is non-negative and at most $T$, it is also finite. Finally, the set of optimal solutions is bounded since the set of feasible solutions is bounded. Therefore, we have $M(\bolds{w})$ and also $O^{\star,\bolds{gr}}_{\frule}(\bolds{w})$ are upper semicontinuous at $\bolds{w}$ by Fact~\ref{fact:continuity-lp}. Now, since $O^{\star,\bolds{gr}}_{\frule}(\bolds{w})$ is upper semicontinuous at every point of $\set{D}$, it is also measurable over $\set{D}$.
\end{proof}
It remains to show that $O_{g,\mathrm{random}}^{\bolds{gr}}(\bolds{w}), O_{g,\mathrm{pro}}^{\bolds{gr}}(\bolds{w}),O_{g,\mathrm{maxmin}}^{\bolds{gr}}(\bolds{w})$ are all measurable. We show this in the following lemma.
\begin{lemma}\label{lem:continuous}
Fix a group arrival sequence $\bolds{gr}$ and a group $g$. We have that $O_{g,\mathrm{random}}^{\bolds{gr}}(\bolds{w}), O_{g,\mathrm{pro}}^{\bolds{gr}}(\bolds{w})$, $O_{g,\mathrm{maxmin}}^{\bolds{gr}}(\bolds{w})$ are continuous over $\set{D}$.
\end{lemma}
\begin{proof}
By definition, $O_{g,\mathrm{random}}^{\bolds{gr}}(\bolds{w}) = \frac{1}{|\Sarr(g,T)|}\sum_{t \in \Sarr(g,T)} \sum_{j\in\Sloc} \frac{s_j w_{t,j}}{\sum_{j'\in \Sloc} s_{j'}}$ where $\Sarr(g,T)$ is a fixed set since $\bolds{gr}$ is fixed. Therefore, $O_{g,\mathrm{random}}^{\bolds{gr}}(\bolds{w})$ is a linear function of $\bolds{w}$ and thus is continuous.

Recall that $\set{Z}_g = \{\bolds{z} \in [0,1]^{\Sarr(g,T) \times M} \colon \sum_{j\in\Sloc}z_{t,j} \leq 1,\forall t\in \Sarr(g,T);~\sum_{t \in \Sarr(g,T)} z_{t,j} \leq s_{j,g},\forall j \in \Sloc\}$ for $O_{g,\mathrm{pro}}^{\bolds{gr}}(\bolds{w})$. When the group arrival sequence is fixed, $\set{Z}_g$ is a fixed polytope. Let $\set{E}_g$ be the finite set of extreme points of $\set{Z}_g$. We have $O_{g,\mathrm{pro}}^{\bolds{gr}}(\bolds{w}) = \frac{1}{N(g,T)}\text{max}_{\bolds{z} \in \set{Z}_g} \sum_{t\in \Sarr(g,T)}\sum_{j \in \Sloc} w_{t,j}z_{t,j} = \frac{1}{N(g,T)}\text{max}_{\bolds{z} \in \set{E}} \sum_{t\in \Sarr(g,T)}\sum_{j \in \Sloc} w_{t,j}z_{t,j}$. Since the last one is the maximum of finitely many linear function of $\bolds{w}$, we conclude that $O_{g,\mathrm{pro}}^{\bolds{gr}}(\bolds{w})$ is continuous.

Finally, let us consider $O_{g,\mathrm{maxmin}}^{\bolds{gr}}(\bolds{w})$. For ease of notations, define $\Sgro'$ as the set of non-empty groups in the fixed group arrival sequence $\bolds{gr}$. Then by definition, $O_{g,\mathrm{maxmin}}^{\bolds{gr}}(\bolds{w})$ is given by $\max_{\bolds{z} \in \set{Z}}\min_{g' \in \Sgro'}\frac{1}{N(g',T)}\sum_{t \in \Sarr(g',T)}\sum_{j=1}^M w_{t,j}z_{t,j}$. It can be equivalently stated by the following program:
\begin{equation}\label{eq:maxmin}
\begin{aligned}
&\text{max}_{\bolds{z} \in \mathbb{R}^{T\times M}_+,\varphi \in \mathbb{R}_+} \quad\varphi\\
s.t.& \quad \varphi \leq \frac{1}{N(g',T)}\sum_{t \in \Sarr(g',T)}\sum_{j\in \Sloc} w_{t,j}z_{t,j},\quad \forall g' \in \Sgro' \\
&\quad \sum_{t \in [T]} z_{t,j} \leq s_j,\quad \forall j \in \Sloc, \quad
\quad \sum_{j=1}^M z_{t,j} \leq 1, \quad \forall t \in [T]
\end{aligned}
\end{equation}
The first three conditions in Fact~\ref{fact:continuity-lp} hold naturally. To show that $O_{g,\mathrm{maxmin}}^{\bolds{gr}}(\bolds{w})$ is continuous, it remains to prove that the dual of \eqref{eq:maxmin} has a bounded set of optimal solutions for a fixed $\bolds{w}$. Let $\{u_g\},\{x_j\},\{y_t\}$ denote the corresponding dual variables for constraints in \eqref{eq:maxmin}. Its dual is given by
\begin{equation}\label{eq:maxmin-dual}
\begin{aligned}
&\text{min}_{\bolds{u} \in \set{R}_+^{|\Sgro'|},\bolds{x}_+^M,\bolds{y}_+^T} \quad \sum_{j \in \Sloc} x_j s_j + \sum_{t \in [T]} y_t\\
s.t.& \quad \sum_{g' \in \Sgro'} u_{g'} \geq 1, \quad -\frac{w_{t,j}u_{g(t)}}{N(g(t),T)}+x_j+y_t \geq 0,\quad \forall t \in [T],j\in \Sloc
\end{aligned}
\end{equation}
Denote $\varphi^\star$ by the optimal value to \eqref{eq:maxmin} which exists and is finite. Let $(\bolds{u}^\star, \bolds{x}^\star,\bolds{y}^\star)$ be any optimal solution to the dual \eqref{eq:maxmin-dual}. Recall that $\bolds{w} \in \set{D}$ and thus $w_{t,j} > 0$ for every $t \in [T],j\in \Sloc$. We next show that $(\bolds{u}^\star, \bolds{x}^\star,\bolds{y}^\star) \in \left[0,\frac{2T\varphi^\star}{\min_{t \in [T],j\in \Sloc} w_{t,j}}\right]^{|\Sgro'|+M+T}$. First, its elements are non-negative by definition. Second, since it is an optimal solution, we have $\sum_{j \in \Sloc} x^\star_j s_j + \sum_{t \in [T]} y_t = \varphi^\star$. As a result, $x^\star_j, y_t^\star \leq \varphi^\star$ for any $j \in \Sloc, t\in [T]$. Now take any group $g' \in \Sgro'$. Since it is non-empty, there exists a case $t'$ such that $g(t') = g'$. Take any location $j' \in \Sloc$. By the constraint for $\tau$ and location $j'$ in \eqref{eq:maxmin-dual}, we have $-\frac{w_{t',j'}u_{g'}^\star}{N(g',T)}+x_{j'}^\star+y_{t'}^\star \geq 0$. Since $x_{j'}^\star,y_{t'}^\star \leq \varphi$ and $1 \leq N(g',T) \leq T$, we have $u^\star_{g'} \leq \frac{2T\varphi^\star}{w_{t',j'}}$. Extending the result to every group in $\Sgro'$ shows that the set of optimal solutions for the dual program \eqref{eq:maxmin-dual} is bounded. Using Fact~\ref{fact:continuity-lp} concludes the proof by proving $O_{g,\mathrm{maxmin}}^{\bolds{gr}}(\bolds{w})$ is continuous over $\set{D}$.
\end{proof}
Combining Lemmas~\ref{lem:measurable} and \ref{lem:continuous} gives the desired measurability result.
\begin{proposition}\label{prop:expost-feasible}
The Random, Proportionally Optimized, and MaxMin fairness rules give measurable required score $O_{g,\frule}$ and optimal score $O^\star_{\frule}$.
\end{proposition}
\begin{proof}
Let $\frule$ be any of the three fairness rules. By definition of the rule, \eqref{eq:outcome-benchmark} is always feasible. Fix a group arrival sequence $\bolds{gr}$. Lemma~\ref{lem:continuous} shows that $O_{g,\frule}^{\bolds{gr}}$ is continuous over $\set{D}$ for every group $g \in \Sgro$. Lemma~\ref{lem:measurable} then gives the measurability of $O^{\star,\bolds{gr}}$ over $\set{D}$. Now since there is only finitely many group arrival sequences (from $[G]^T$), we have $O^{\star}_{\frule}(\bomega), \{O_{g,\frule}(\bomega)\}$ are measurable over $[G]^T \times \set{D}$. We then finish the proof by noting that they are also measurable over the original sample space (with Borel $\sigma-$algebra) since $[0,1]^{T\times M} \setminus \set{D}$ is of measure $0$.
\end{proof}

\subsection{Sensitivity of Random and Proportionally Optimized (Proposition~\ref{prop:irr-rules})}\label{app:irr-rules}

The following result shows that Random is $1-$sensitive.
\begin{lemma}\label{lem:irr-random}
Random is $(1,\delta)$-sensitive for any $\delta \geq 0$.
\end{lemma}
\begin{proof}
Fix a group $g \in \Sgro$ and a sample path $\bomega = (\btheta_t)_{t \in [T]} \in \Omega$. By definition of the Random fairness rule (Example~\ref{ex:random-fairness-rule}), we have the total requirement for group $g$ to be $Q_{g,\mathrm{random}}(\bomega) = \frac{\sum_{t\colon g(\btheta_t) = g} \sum_{j \in \set{M}} w_j(\btheta_t)s_j}{T}$. Now suppose another sample path $\tilde{\bomega} = (\tilde{\btheta}_t)_{t \in [T]}$ such that $\btheta_t = \tilde{\btheta}_t$ for all $t \in [T]$ except for $t = \tau$. We have that 
\begin{align*}
|Q_{g,\mathrm{random}}(\bomega) - Q_{g,\mathrm{random}}(\tilde{\bomega})| &= \frac{\left|\indic{g(\btheta_{\tau}) = g}\sum_{j \in \set{M}}w_j(\btheta_{\tau})s_j - \indic{g(\tilde{\btheta}_{\tau}) = g}\sum_{j \in \set{M}}w_j(\tilde{\btheta_{\tau}})s_j\right|}{T} \\
&\leq \frac{\sum_{j \in \set{M}} s_j \left|\indic{g(\btheta_{\tau}) = g}w_j(\btheta_{\tau}) - \indic{g(\tilde{\btheta}_{\tau}) = g}w_j(\tilde{\btheta}_{\tau})\right|}{T} \\
&\leq \frac{\sum_{j \in \set{M}} s_j}{T} = 1.
\end{align*}
Moreover, if $g(\btheta_{\tau}) \neq g$ and $g(\tilde{\btheta}_{\tau}) \neq g$, the first equality shows that $|Q_{g,\mathrm{random}}(\bomega) - Q_{g,\mathrm{random}}(\tilde{\bomega})| = 0 \leq \sqrt{p_g}$. Since the bound holds for any sample path in the sample space $\Omega$, we thus complete the proof by invoking Definition~\ref{def:irr-rule}.
\end{proof}
We next show the property for the Proportionally Optimized fairness rule.
\begin{lemma}\label{lem:irr-pro}
Proportionally Optimized is $(2,\delta)$-sensitive for any $\delta \geq 0$.
\end{lemma}
\begin{proof}
Fix $g \in \Sgro$ and a sample path $\bomega = (\btheta_t)_{t \in [T]}$. Following the definition of the fairness rule in Example~\ref{ex:prop-opt-fairness-rule}, recall $\set{\tilde{Z}}_g(\bomega) = \{\bolds{\tilde{z}} \in [0,1]^{\Sarr(g,T,\bomega) \times M} \colon \sum_{j\in\Sloc}\tilde{z}_{t,j} \leq 1,\forall t\in \Sarr(g,T,\bomega);~\sum_{t \in \Sarr(g,T,\bomega)} \tilde{z}_{t,j} \leq s_{j,g},\forall j \in \Sloc\}$ is the set of feasible assignments for group $g$ and sample path $\bomega$. For any $\bolds{z} \in \set{\tilde{Z}}_g(\bomega)$, define $\mathrm{OBJ}(\bolds{z},\bomega) = \sum_{t \in \set{A}(g,T,\bomega)}\sum_{j \in \set{M}}w_{j}(\btheta_t)z_{t,j}$ be the total scores of group $g$ under this assignment. By definition of the proporitionally optimized rule, we have $Q_{g,\mathrm{pro}}(\bomega) = \max_{\bolds{z} \in \tilde{Z}_g(\bomega)} \mathrm{OBJ}(\bolds{z},\bomega)$.

Consider another sample path $\tilde{\bomega} = (\tilde{\btheta}_t)_{t \in [T]}$ such that $\btheta_t = \tilde{\btheta}_t$ except for $t = \tau$. Let $\bolds{z}^{\star} \in \tilde{Z}_g(\bomega)$ be the optimal solution of $\mathrm{OBJ}(\bolds{z},\bomega)$ and let $\tilde{\bolds{z}}^\star \in \arg\max_{\bolds{z} \in \tilde{Z}_g(\tilde{\bomega})} \mathrm{OBJ}(\bolds{z},\tilde{\bomega})$ be that under sample path $\tilde{\bomega}$. We know $Q_{g,\mathrm{pro}}(\bomega) = \mathrm{OBJ}(\bolds{z}^\star, \bomega)$ and $Q_{g,\mathrm{pro}}(\tilde{\bomega}) = \mathrm{OBJ}(\tilde{\bolds{z}}^\star, \tilde{\bomega})$. It remains to show that $|\mathrm{OBJ}(\bolds{z}^\star, \bomega) - \mathrm{OBJ}(\tilde{\bolds{z}}^\star, \tilde{\bomega})| \leq 2$. We consider four cases.
\begin{itemize}
\item $g(\bolds{\theta}_\tau) \neq g, g(\tilde{\btheta}_\tau) \neq g$. Then $\mathrm{OBJ}(\bolds{z}^\star, \bomega) = \mathrm{OBJ}(\tilde{\bolds{z}}^\star, \tilde{\bomega})$ since the feasible sets and the objective functions are the same under $\bomega$ and This establishes the first requirement in Definition~\ref{def:irr-rule} since $|Q_{g,\mathrm{pro}}(\bomega) - Q_{g,\mathrm{pro}}(\tilde{\bomega})| = 0 \leq 2\sqrt{p_g}.$
\item $g(\bolds{\theta}_\tau) = g(\tilde{\btheta}_\tau) = g$. Then $\bolds{z}^\star$ and $\tilde{\bolds{z}}^\star$ are both feasible assignments for group $g$ under $\bomega$ and $\tilde{\bomega}$. As a result,
\begin{align*}
Q_{g,\mathrm{pro}}(\bomega) - Q_{g,\mathrm{pro}}(\tilde{\bomega}) = \mathrm{OBJ}(\bolds{z}^\star, \bomega) - \mathrm{OBJ}(\tilde{\bolds{z}}^\star, \tilde{\bomega}) &\leq \mathrm{OBJ}(\bolds{z}^\star, \bomega) - \mathrm{OBJ}(\bolds{z}^\star, \tilde{\bomega}) \\
&= \sum_{j \in \set{M}} (w_j(\btheta_\tau) - w_j(\tilde{\btheta}_\tau))z^{\star}_{\tau,j} \leq 1,
\end{align*}
where the first inequality is because $\tilde{\bolds{z}}^\star$ is the optimal solution for sample path $\tilde{\bomega}$. Similarly, we can show $Q_{g,\mathrm{pro}}(\bomega) - Q_{g,\mathrm{pro}}(\tilde{\bomega}) \geq -1$ by swapping the two sample paths in the above argument. Therefore, for this case, $|\mathrm{OBJ}(\bolds{z}^\star, \bomega) - \mathrm{OBJ}(\tilde{\bolds{z}}^\star, \tilde{\bomega})| \leq 1$.
\item $g(\btheta_\tau) = g, g(\tilde{\btheta}_\tau) \neq g$. It implies that $\Sarr(g,T,\tilde{\bomega}) \subseteq \Sarr(g,T,\bomega), N(g,T,\tilde{\bomega}) = N(g,T,\bomega) - 1$ and thus $\tilde{\bolds{z}}^\star \in \tilde{Z}_g(\tilde{\bomega}) \subseteq \tilde{Z}_g(\bomega)$ is a feasible assignment for sample path $\bomega$. In addition, since cases of group $g$ in sample path $\tilde{\bomega}$ keep the same scores in $\bomega$, we have
\[
\mathrm{OBJ}(\tilde{\bolds{z}}^\star, \tilde{\bomega}) = \mathrm{OBJ}(\tilde{\bolds{z}}^\star, \bomega) \leq \mathrm{OBJ}(\bolds{z}^\star, \bomega).
\]
It remains to show $\mathrm{OBJ}(\bolds{z}^\star, \bomega) - \mathrm{OBJ}(\tilde{\bolds{z}}^\star, \tilde{\bomega}) \leq 2$. Let $n = N(g,T,\bomega)$ and $\bolds{z}' = \frac{n-1}{n}\bolds{z}^\star$. Moreover, let $\bolds{z}'_{\tau,j} = 0$ for all $j \in \set{M}$. Then we have $\bolds{z}' \in \tilde{Z}_g(\tilde{\bomega})$ is a feasible assignment for sample path $\tilde{\bomega}$. As a result,
\begin{align*}
\mathrm{OBJ}(\bolds{z}^\star, \bomega) - \mathrm{OBJ}(\tilde{\bolds{z}}^\star, \tilde{\bomega}) &\leq \mathrm{OBJ}(\bolds{z}^\star, \bomega) - \mathrm{OBJ}(\bolds{z}', \tilde{\bomega}) \\
&= \sum_{j \in \set{M}} w_j(\btheta_\tau)z^\star_{\tau,j} + \sum_{t \neq \tau}\indic{g(\btheta_t) = g}\sum_{j \in \set{M}} w_j(\btheta_t)\left(z^\star_{t,j} - z'_{t,j}\right) \\
&= \sum_{j \in \set{M}} w_j(\btheta_\tau)z^\star_{\tau,j} + \frac{1}{n}\sum_{t \neq \tau}\indic{g(\btheta_t) = g}\sum_{j \in \set{M}} w_j(\btheta_t)z^\star_{t,j} \\
&\leq 1 + \frac{n-1}{n} \leq 2,
\end{align*}
where the first inequality is because $\tilde{\bolds{z}}^\star$ is the optimal solution under $\tilde{\bomega}$; the first equality is by the assumption that $g(\btheta_\tau) = g, g(\tilde{\btheta}_\tau) \neq g$ and the property that $\bomega$ only differs with $\tilde{\bomega}$ for case $\tau$; the second equality is by the definition of $\bolds{z}'$ and the last inequality is because the score of each case is at most $1$. We thus show $0 \leq \mathrm{OBJ}(\bolds{z}^\star, \bomega) - \mathrm{OBJ}(\tilde{\bolds{z}}^\star, \tilde{\bomega}) \leq 2$ for this case;
\item $g(\btheta_\tau) \neq g, g(\tilde{\btheta}_\tau) = g$. This is symmetric with the last case. Following the same logic gives $-2 \leq \mathrm{OBJ}(\bolds{z}^\star, \bomega) - \mathrm{OBJ}(\tilde{\bolds{z}}^\star, \tilde{\bomega}) \leq 0$.
\end{itemize} 
Summarizing the above four cases shows $|\mathrm{OBJ}(\bolds{z}^\star, \bomega) - \mathrm{OBJ}(\tilde{\bolds{z}}^\star, \tilde{\bomega})| \leq 2$ and thus $|Q_{g,\mathrm{pro}}(\bomega) - Q_{g,\mathrm{pro}}(\tilde{\bomega})| \leq 2$. We also show that if $g(\btheta_{\tau}) \neq g$ and $g(\tilde{\btheta}_{\tau})$, we have $|Q_{g,\mathrm{pro}}(\bomega) - Q_{g,\mathrm{pro}}(\tilde{\bomega})| = 0 \leq 2\sqrt{p_g}$, which completes the proof.
\end{proof}
\begin{proof}[Proof of Proposition~\ref{prop:irr-rules}]
We prove the result by combining Lemmas~\ref{lem:irr-random}, and \ref{lem:irr-pro}.
\end{proof}

\subsection{Sensitivity of the MaxMin fairness rule (Proposition~\ref{prop:irr-maxmin})}\label{app:irr-maxmin}
\begin{proof}
Fix $\delta \in [TGe^{4-p_{\min}T / 8}, 1)$ and let $\beta = (\delta / (e^4\times G T))^{1/4}.$ Then $\beta \in (0,e^{-1}]$. Define event $\set{B}$ such that $N(g,T) \in \left[p_g T - 2\sqrt{2p_g T\ln\frac{1}{\beta}}, p_g T + 5\sqrt{p_g T\ln\frac{1}{\beta}}\right]$ for all group $g \in \Sgro$. Applying Lemma~\ref{lem:chernoff} and a union bound gives $\Pr\{\set{B}\} \geq 1 - 2G\beta^4\geq 1 - \delta/T$. 

Fix a sample path $\bomega = (\btheta_t)_{t \in [T]} \in \set{B}$. By the assumption that $\delta \geq TGe^{4-p_{\min}T / 8}$ and the definition of $\beta$, we have $p_{\min}T \geq 32\ln\frac{1}{\beta}$. In addition, for every group $g \in \Sgro$, \[N(g,T,\bomega) \geq p_g T - 2\sqrt{2p_g T\ln\frac{1}{\beta}} \geq \frac{1}{2}p_g T \geq 16,\]
and thus there is no empty group.

Let $\tilde{\bomega} = (\tilde{\btheta_t})_{t \in [T]}$ be another sample path such that $\btheta_t = \tilde{\btheta}_t$ except for $t = \tau$. This sample path does not have to be in $\set{B}$, but we must have $N(g,T,\tilde{\bomega}) \geq N(g,T,\bomega) - 1 \geq 15$ for any group $g$. Recall that the MaxMin rule (Example~\ref{ex:maxmin-fairness-rule}) solves the maximum $\varphi$ such that $\sum_{t \in \Sarr(g,T,\bomega)}\sum_{j \in \set{M}} w_{j}(\btheta_t)z_{t,j} \geq N(g,T,\bomega)\varphi$ for all $g \in \Sgro$ for some feasible assignment $\bolds{z} \in \tilde{\set{Z}}$. Denote the optimal value for sample path $\bomega$ by $\mathrm{maxmin}(\bomega)$ and that for $\tilde{\bomega}$ by $\mathrm{maxmin}(\tilde{\bomega})$. 

We next show that $|\mathrm{maxmin}(\tilde{\bomega}) -\mathrm{maxmin}(\bomega)| \leq \frac{2}{\min_{g \in \Sgro} 
\max(N(g,T,\bomega), N(g,T,\tilde{\bomega}))}$. Suppose that one optimal assignment for $\tilde{\bomega}$ is $\tilde{\bolds{z}}$. We have 
\[
\mathrm{maxmin}(\bomega) \geq \min_{g \in \Sgro} \frac{\sum_{t \in \Sarr(g,T,\bomega)} \sum_{j \in \set{M}} w_j(\btheta_t)\tilde{z}_{t,j}}{N(g,T,\bomega)}.
\]
Take $g'$ to be a group that achieves the minimum value for the right hand side. Since $N(g,T,\tilde{\bomega}) \geq 1$,
\[
\mathrm{maxmin}(\tilde{\bomega}) \leq \frac{\sum_{t \in \Sarr(g',T,\tilde{\bomega})} \sum_{j \in \set{M}} w_j(\tilde{\btheta}_t)\tilde{z}_{t,j}}{N(g,T,\tilde{\bomega})}.
\]
As a result,
\begin{align}
\mathrm{maxmin}(\tilde{\bomega}) - \mathrm{maxmin}(\bomega) &\leq \left|\frac{\sum_{t \in \Sarr(g',T,\tilde{\bomega})} \sum_{j \in \set{M}} w_j(\tilde{\btheta}_t)\tilde{z}_{t,j}}{N(g',T,\tilde{\bomega})} - \frac{\sum_{t \in \Sarr(g',T,\bomega)} \sum_{j \in \set{M}} w_j(\btheta_t)\tilde{z}_{t,j}}{N(g',T,\bomega)}\right| \label{eq:maxmin-0}\\
&\hspace{-1in}\leq \frac{N(g',T,\bomega)\left|\sum_{t \in \Sarr(g',T,\tilde{\bomega})} \sum_{j \in \set{M}} w_j(\tilde{\btheta}_t)\tilde{z}_{t,j} - \sum_{t \in \Sarr(g',T,\bomega)} \sum_{j \in \set{M}} w_j(\btheta_t)\tilde{z}_{t,j}\right|}{N(g',T,\tilde{\bomega})N(g',T,\bomega)} \label{eq:maxmin-1}\\
&+ \frac{\left|N(g',T,\bomega)-N(g',T,\tilde{\bomega})\right|\sum_{t \in \Sarr(g',T,\bomega)} \sum_{j \in \set{M}} w_j(\btheta_t)\tilde{z}_{t,j}}{N(g',T,\tilde{\bomega})N(g',T,\bomega)}. \label{eq:maxmin-2}
\end{align}
To bound \eqref{eq:maxmin-1}, we upper bound 
\begin{equation}\label{eq:maxmin-diff}
\left|\sum_{t \in \Sarr(g',T,\tilde{\bomega})} \sum_{j \in \set{M}} w_j(\tilde{\btheta}_t)\tilde{z}_{t,j} - \sum_{t \in \Sarr(g',T,\bomega)} \sum_{j \in \set{M}} w_j(\btheta_t)\tilde{z}_{t,j}\right|
\end{equation}
by the following three scenarios:
\begin{itemize}
\item If $\Sarr(g',T,\bomega) = \Sarr(g',T,\tilde{\bomega})$ and $\tau \neq \Sarr(g',T,\bomega)$, then cases of group $g$ have the same scores in $\bomega$ and $\tilde{\bomega}$ and thus \eqref{eq:maxmin-diff} is zero;
\item If $\Sarr(g',T,\bomega) = \Sarr(g',T,\tilde{\bomega})$ and $\tau \in \Sarr(g',T,\bomega)$, then $\eqref{eq:maxmin-diff} \leq \sum_{j \in \set{M}} \tilde{z}_{\tau,j}|w_j(\btheta_\tau)-w_j(\tilde{\btheta}_\tau| \leq 1$;
\item If $\Sarr(g',T,\bomega) \neq \Sarr(g',T,\tilde{\bomega})$, then one is a subset of the another and has exactly one fewer case. The scores of common cases are the same. As a result, we still have $\eqref{eq:maxmin-diff} \leq 1$.
\end{itemize}
Summarizing the above three scenarios gives $\eqref{eq:maxmin-diff} \leq 1$ and thus $\eqref{eq:maxmin-1} \leq \frac{N(g',T,\bomega)}{N(g',T,\tilde{\bomega})N(g',T,\bomega)} \leq \frac{1}{N(g',T,\tilde{\bomega})}$. In addition, for \eqref{eq:maxmin-2}, we know $\left|N(g',T,\bomega)-N(g',T,\tilde{\bomega})\right| \leq 1$ and $\sum_{t \in \Sarr(g',T,\bomega)} \sum_{j \in \set{M}} w_j(\btheta_t)\tilde{z}_{t,j} \leq N(g',T,\bomega)$. Therefore, we again have $\eqref{eq:maxmin-2} \leq \frac{1}{N(g',T,\tilde{\bomega})}$. This gives \[\mathrm{maxmin}(\tilde{\bomega}) - \mathrm{maxmin}(\bomega) \leq \eqref{eq:maxmin-1} + \eqref{eq:maxmin-2} \leq \frac{2}{N(g',T,\tilde{\bomega})}.\] 
By symmetry in the right hand side of \eqref{eq:maxmin-0}, we also have $\mathrm{maxmin}(\tilde{\bomega}) - \mathrm{maxmin}(\bomega) \leq \frac{2}{N(g',T,\bomega)}$. Therefore, $\mathrm{maxmin}(\tilde{\bomega}) - \mathrm{maxmin}(\bomega) \leq \frac{2}{\max(N(g',T,\bomega),N(g',T,\tilde{\bomega}))} \leq \frac{2}{\min_{g \in \Sgro} 
\max(N(g,T,\bomega), N(g,T,\tilde{\bomega}))}$. Following the same argument but swapping $\bomega, \tilde{\bomega}$ then gives 
\begin{equation}\label{eq:maxmin-bound}
\left|\mathrm{maxmin}(\bomega) - \mathrm{maxmin}(\tilde{\bomega})\right| \leq \frac{2}{\min_{g \in \Sgro} 
\max(N(g,T,\bomega), N(g,T,\tilde{\bomega}))} \leq \frac{2}{\min_{g \in \Sgro} N(g,T,\bomega)}.
\end{equation}
Let $\hat{g}$ be the group with minimum $N(g,T,\bomega)$. Then for any group $g \in \Sgro$, the difference in total requirement for the two sample paths is upper bounded by
\begin{align}
\left|Q_{g,\mathrm{maxmin}}(\bomega) - Q_{g,\mathrm{maxmin}}(\tilde{\bomega})\right| &= \left|\mathrm{maxmin}(\bomega)N(g,T,\bomega) - \mathrm{maxmin}(\tilde{\bomega})N(g,T,\tilde{\bomega})\right| \nonumber\\
&\hspace{-1in}\leq \left|\mathrm{maxmin}(\bomega) - \mathrm{maxmin}(\tilde{\bomega})\right|N(g,T,\bomega) + \mathrm{maxmin}(\tilde{\bomega})|N(g,T,\bomega) - N(g,T,\tilde{\bomega})|\nonumber\\
&\leq \frac{2N(g,T,\bomega)}{N(\hat{g},T,\bomega)}+ 1 \tag{By \eqref{eq:maxmin-bound}} \\
&\leq \frac{2(p_g T + 5\sqrt{p_g T \ln \frac{1}{\beta}})}{p_{\hat{g}} T - 2\sqrt{2p_{\hat{g}} T \ln \frac{1}{\beta}}} + 1 \tag{By the setting of $\set{B}$ and the assumption that $\bomega \in \set{B}$} \\
&\leq \frac{12p_g T}{0.5p_{\hat{g}}T} + 1 \leq \frac{25p_g}{p_{\min}} \tag{$p_gT, p_{\hat{g}}T \geq p_{\min}T \geq 32\ln\frac{1}{\beta}$},
\end{align}
which completes the proof since $\left|Q_{g,\mathrm{maxmin}}(\bomega) - Q_{g,\mathrm{maxmin}}(\tilde{\bomega})\right| \leq \frac{25p_g}{p_{\min}} \leq \sqrt{p_g}\chi$ with $\chi = 25p_{\min}^{-1}$.
\end{proof}

\subsection{Impossibility of distribution-dependent vanishing regret for arbitrary fairness rules (Section~\ref{ssec:criteria})}\label{app:hard-fairness-approx-target}
The section shows that one cannot hope for distribution-dependent vanishing regret for arbitrary fairness rules. In particular, we construct a problem setting such that for any large enough $T$ the ex-post $g-$regret of any non-anticipatory algorithm must be bounded away from zero with a constant probability.

Formally, let $\Phi(x)$ be the standard normal distribution. We need the following statement of Berry-Esseen Theorem adopted from Theorem 3.4.17 of \cite{durrett2019probability}.
\begin{fact}\label{fact:berry-esseen}
Let $X_1,X_2,\ldots$ be i.i.d. with $\expect{X_i} = 0,\expect{X_i^2} = \sigma^2$ and $\expect{|X_i|^3} = \rho < \infty$. If $F_n(x)$ is the distribution of $(X_1 + \cdots + X_n) / \sigma \sqrt{n}$ then for every $x$, $|F_n(x) - \Phi(x)| \leq 3\rho/(\sigma^3 \sqrt{n})$.
\end{fact}

\begin{proposition}\label{prop:hard-fairness-example-fairrule}
There exists a constant $C > 0$ and a problem setting such that, for any $T = 100K$ with sufficiently large integer $K$ and any non-anticipatory algorithm, with probability at least $C$, we must have both
\begin{itemize}
\item high ex-post $g-$regret: $\max_{g \in \set{G}} \left(O_{g,\frule}(\bomega)-\alpha_g(\bomega)\right) \geq 0.235$;
\item low ex-post multiplicative approximation: $\min_{g \in \set{G}} \frac{\alpha_g(\bomega)}{O_{g,\frule}(\bomega)} \leq 0.53$.
\end{itemize}
\end{proposition}
\begin{proof}
Fix $T = 100K$ for $K \geq \frac{12}{\Phi(-25)^2}$ and let $C = \frac{(1-4e^{-4})\Phi(-25)}{2} > 0$. Consider a setting with $M = 2$ locations each with capacity $\frac{T}{2}$. There are two groups, each with equal arrival probabilities $p_1 = p_2 = 0.5$. For every case $t$, the employment scores are given by $w_{t,1} = 1, w_{t,2} = 0$. The fairness rule $\frule$ is defined as follows. For a sample path $\bomega$, the DM observes the number of arrivals of each group and assigns as many arrivals as possible of the larger group to the good location to set the required score. That is, if group $1$ has more or equal arrivals than group $2$, then $O_{1,\frule}(\bomega) = \frac{T}{2N(1,T)}, O_{2,\frule}(\bomega) = 0$; otherwise, $O_{1,\frule}(\bomega) = 0,O_{2,\frule}(\bomega) = \frac{T}{2N(2,T)}$. 

Fix a non-anticipatory algorithm. The algorithm could be randomized whose randomness is independent of the case arrival process. We assume the probability space is extended to include the randomness of the algorithm. Define $\tilde{R}(\bomega)$ to be $\max\left( O_{1,\frule}(\bomega)-\alpha_1(\bomega), O_{2,\frule}(\bomega) - \alpha_2(\bomega)\right)$ and let the multiplicative approximation (fairness ratio)
be $\mathrm{FR}(\bomega) = \max_{g \in \set{G}} \frac{O_{g,\frule}(\bomega)-\alpha_g(\bomega)}{O_{g,\frule}(\bomega)}.$
We next show that with probability at least $C$, we have $\tilde{R}(\bomega) \geq 0.235$ and $\mathrm{FR}(\bomega) \leq 0.53$.

Recall that $N(g,t)$ is the number of arrivals of group $g$ up to round $t$. Let $Q = 0.97T = 97K$. By Hoeffding's Inequality (Fact~\ref{fact:hoeffding}), we have $\Pr\{|N(1,Q) - \frac{Q}{2}| > 20\sqrt{K}\} \leq 2\exp\left(\frac{-400K}{97K}\right) \leq 2e^{-4}$. Similarly, we have $\Pr\{|N(2,Q)-\frac{Q}{2}| > 20\sqrt{K}\} \leq 2e^{-4}$. Define 
\begin{equation}\label{eq:close-arrival}
\set{S} = \left\{|N(1,Q)-\frac{Q}{2}| \leq 20\sqrt{K}\right\} \cap \left\{|N(2,Q)-\frac{Q}{2}| \leq 20\sqrt{K}\right\}.
\end{equation}
Then by union bound we have $\Pr\{\set{S}\} \geq 1 - 4e^{-4}$. In addition, under $\set{S}$ we have $|N(1,Q) - N(2,Q)| \leq 40\sqrt{K}$. Define $X$ by the difference between numbers of arrivals of the two groups in rounds $Q+1$ to $T$. That is, $X= N(1,T) - N(1,Q) - (N(2,T) - N(2,Q))$. We know $X$ is the sum of $T - Q = 3K$ Rademacher random variables. Then by Berry-Esseen Theorem (Fact~\ref{fact:berry-esseen}), we have $|\Pr\{\frac{X}{\sqrt{3K}} \leq a\} - \Phi(a)| \leq \sqrt{\frac{3}{K}}$ for any real value $a$. Taking $a = -25$ leads to $\Pr\{X \leq -41\sqrt{K}\} \geq \Phi(-25) - 
\sqrt{\frac{3}{K}}\geq \frac{1}{2}\Phi(-25)$ since $K \geq \frac{12}{\Phi(-25)^2}$. By the symmetry of $X$, we have $\Pr\{X \geq 41\sqrt{K}\} = \Pr\{X \leq -41\sqrt{K}\} \geq \frac{1}{2}\Phi(-25)$.

Now define $A_1$ by the number of cases in group $1$ that are assigned to location $1$ by round $Q$. Note that both $A_1$ and $\set{S}$ are independent of $X$ since they are in the filtration of rounds up to $Q$ and are independent of arrivals in rounds $Q+1$ to $T$. Condition on $\set{S}$ defined in \eqref{eq:close-arrival}. Consider two scenarios:
\begin{itemize}
    \item $A_1 \geq 0.235T, X \leq -41\sqrt{K}$. In this scenario, we know there are more arrivals of group $2$ since $N(2,T) - N(1,T) \geq 41\sqrt{K} - 40\sqrt{K} > 0$. By the definition of the fairness rule $\frule$, we have $O_{2,\frule}(\bomega)N(2,T) = 0.5T$. However, since $A_1 \geq 0.235T$, at least $ 0.235T$ capacity from location $1$ are assigned to group $1$. We must have $\alpha_2(\bomega)N(2,T) \leq \frac{T}{2}-0.235T = 0.265T$, so $\tilde{R}(\bomega) \geq \frac{1}{N(2,T)}\left(O_2(\bomega)N(2,T)-\alpha_2(\bomega)N(2,T)\right) \geq \frac{0.235T}{T} = 0.235$. In addition, \[\mathrm{FR}(\omega) \leq \frac{\alpha_2(\bomega)N(2,T)}{O_2(\bomega)N(2,T)} \leq \frac{0.265T}{0.5T} = 0.53.\]
    \item $A_1 < 0.235T, X \geq 41\sqrt{K}$. Similar to the above scenario, we have that $N(1,T) - N(2,T) > 0$ and thus $O_{1,\frule}(\bomega)N(1,T) = 0.5T$. But since only $A_1$ cases from group $1$ are assigned to location $1$ up to round $Q$ and there are at most $T-Q=3K$ cases after round $Q$, the total score of group $1$ is at most $0.235T+0.03T=0.265T$. Therefore, under this scenario, $\tilde{R}(\bomega) \geq \frac{1}{N(1,T)}\left(O_{1,\frule}(\bomega)N(1,T)-\alpha_1(\bomega)N(1,T)\right) \geq \frac{1}{T}\left(0.5T-0.265T\right) = 0.235$. In addition, \[\mathrm{FR}(\omega) \leq \frac{\alpha_1(\bomega)N(1,T)}{O_1(\bomega)N(1,T)} \leq \frac{0.265T}{0.5T} = 0.53.\]
\end{itemize}

Summarizing the above two scenarios, we obtain that
\begin{align*}
\Pr&\left\{\tilde{\set{R}}(\bomega) \geq 0.235, \mathrm{FR}(\bomega) \leq 0.53\right\}\\&\geq\Pr\{\set{S}\}\Pr\left\{\tilde{\set{R}}(\bomega) \geq 0.235, \mathrm{FR}(\bomega) \leq 0.53 \mid\set{S}\right\} \\
&\overset{(a)}{\geq} \Pr\{\set{S}\}\left(\Pr\left\{A_1 \geq 0.235T, X \leq -41\sqrt{K} \mid \set{S}\right\} + \Pr\left\{A_1 < 0.235T, X \geq 41\sqrt{K} \mid \set{S}\right\}\right) \\
&\overset{(b)}{=} \Pr\{\set{S}\}\left(\Pr\left\{A_1 \geq 0.235T\mid \set{S}\right\}\Pr\left\{X \leq -41\sqrt{K}\right\} + \Pr\left\{A_1 < 0.235T \mid \set{S}\right\}\Pr\left\{X \geq 41\sqrt{K}\right\}\right) \\
&\overset{(c)}{=} \Pr\{\set{S}\}\frac{\Phi(-25)}{2}\left(\Pr\left\{A_1 \geq 0.235T\mid \set{S}\right\}+ \Pr\left\{A_1 < 0.235T \mid \set{S}\right\}\right) \\
&\overset{(d)}{=} \Pr\{\set{S}\}\frac{\Phi(-25)}{2} \geq \frac{(1-4e^{-4})\Phi(-25)}{2}
\end{align*}
where (a) is by the law of total probability and the fact that the two events have no intersection; (b) is by the independence between $X$ and the filtration up to round $Q$; (c) is by the result we obtained before that $|X| \geq 41\sqrt{K}$ happens with probability at least $\frac{\Phi(-25)}{2}$; (d) is because the two events are complement. We thus finish the proof.
\end{proof}

\subsection{Impossibility of distribution-independent vanishing regret for the MaxMin fairness rule (Section~\ref{ssec:criteria})}\label{app:max-min-imp}
We next demonstrate the difficulty of achieving distribution-independent vanishing regret for the MaxMin fairness rule. For any large enough number of arrivals $T$ and any algorithm, we can construct a feature distribution such that either the ex-post $g-$regret or the global regret is bounded away from zero.
\begin{proposition}
Fix an integer $c \geq 2$ and a non-anticipatory algorithm. For any $T = 6ck \geq 169\ln(32c)$ with an integer $k$, there exists a setting with MaxMin fairness rule, such that with probability at least $\frac{1}{512c^2}\min\left(0.3^{\frac{\ln(16c)}{6c}},1-e^{\frac{-\ln(16c)}{6c}}\right) > 0$,
either of the following two events holds:
\begin{itemize}
\item There is a group $g$ with $\expect{N(g,T)} \geq T/2, O_g(\bomega) \geq \frac{1}{36c}$ but $\frac{\alpha_g(\bomega)}{O_g(\bomega)} < \frac{1}{c}$;
\item
$\frac{\sum_{t\in[T]}\sum_{j\in\set{M}} w_{t,j} z_{t,j}/T}{\min(\expectsub{\bomega'}{O^\star(\bomega')}, O^{\star}(\bomega))} < 1-\frac{1}{8c}$ but $\min(\expectsub{\bomega'}{O^\star(\bomega')}, O^{\star}(\bomega)) \geq \frac{1}{4}$.
\end{itemize}
\end{proposition}
\begin{proof}
Fix $c$ and the non-anticipatory algorithm. Let $T = 6kc$ for some integer $k$ such that $T \geq 169\ln(32c)$. Consider the following setting with two locations, $1$ and $2$, and three groups $A, B, C$, such that $s_1 = T/3, s_2 = 2T/3$ and $p_A = p_B = \frac{1}{2} - \frac{\ln(16c)}{2T}, p_C = \frac{\ln(16c)}{T}$. Cases of the same group have same scores across locations. Let $\tilde{w}_{g,j}$ denote the corresponding scores. Then we set $\tilde{w}_{A,1}=1,\tilde{w}_{A,2}=\frac{1}{12c}$, $\tilde{w}_{B,1}=\frac{1}{12c}, \tilde{w}_{B,2}=0$ and $\tilde{w}_{C,1}=\tilde{w}_{C,2}=0$. Observe that if $N(C,T) \geq 1$, then the MaxMin rule assigns $O_g = 0$ for all groups, and $O^\star \geq \min(1/3,N(A,T)/T)$. In addition, if $N(C,T) = 0,N(B,T)>0$, then $O_B = \frac{\min(T/3,N(B,T))}{12cN(B,T)}.$

By Hoeffding's Inequality (Fact~\ref{fact:hoeffding}), \[
\Pr\{N(A,T) \leq T/3\} = \Pr\{N(A,T) \leq \expect{N(A,T)}-(T/6-\ln(16c)/2)\} \leq \exp(-2T/7)\leq 1/(16c)
\]
where the last two inequalities use $T \geq 130\ln(16c)$. Since $\Pr\{N(C,T) = 0\} = (1-\ln(16c)/T)^T \leq e^{-\ln(16c)}=1/(16c)$, we have
\begin{align}
\expectsub{\bomega'}{O^\star(\bomega')} \geq \expectsub{\bomega'}{O^\star(\bomega')\indic{N(C,T) \geq 1}} &\geq \expectsub{\bomega'}{\min\left(\frac{1}{3},\frac{N(A,T)}{T}\right)\indic{N(C,T) \geq 1}} \nonumber \\
&\hspace{-3in}\geq \expectsub{\bomega'}{\min\left(\frac{1}{3},\frac{N(A,T)}{T}\right)}-\frac{1}{3}\Pr\{N(C,T)=0\} \geq \frac{1}{3}\left(\Pr\{N(A,T) \geq T/3\}-\Pr\{N(C,T)=0\}\right) \nonumber\\
&\geq \frac{1}{3}(1-1/(8c)). \label{eq:bound-eostar}
\end{align}
Consider the first $Q=k(6c-1)$ cases. By Hoeffding's Inequality and the fact that $T \geq 169\ln(32c)$, 
\[
\Pr\left\{N(A,Q) \leq \frac{T}{3}\right\} = \Pr\left\{N(A,Q) \leq \expect{N(A,Q)}-\frac{T-5\ln(16c)}{12}\right\} \leq \exp(-2T/(13^2))\leq \frac{1}{1024c^2}.
\]
Similarly, we have $\Pr\{N(B,Q) \leq T/3\} \leq \frac{1}{1024c^2}.$ In addition, $\Pr\{N(C,Q) = 0\} = (1-\ln(16c)/T)^Q \geq (1-\ln(16c)/T)^T \geq 0.3^{\ln(16c)}\geq e^{-2\ln(16c)}=\frac{1}{256c^2}$ where we use the fact that $(1-1/y)^y \geq 0.3$ for $y \geq 5$ and $T/\ln(16c) \geq 5$. Define event $\set{S} = \{N(A,Q) > T/3,N(B,Q)>T/3,N(C,Q)=0\}$. Then by union bound, $\Pr\{\set{S}\} \geq \frac{1}{512c^2}.$ 

Condition on $\set{S}$. Let $Z_1$ be the number of group $A$ cases assigned to location $1$ in the first $Q$ periods. Consider two scenarios:
\begin{itemize}
\item $Z_1 > (1-1/c)s_1=(1-1/c)T/3$. Define event $\set{S}_B = \{N(C,T)-N(C,Q)=0\}$. Then under $\set{S}$ and $\set{S}_B$, since there is no group $C$ arrival and $N(B,T) \geq N(B,Q) \geq T/3$, the MaxMin rule sets $O_B(\bomega) = \frac{T/3}{12cN(B,T)} \geq \frac{1}{36c}$. However, since at least $Z_1  > (1-1/c)s_1$ capacity of location $1$ is assigned to group $A$ arrivals, the average score of group $B$, $\alpha_B(\bomega)$, is strictly less than $s_1/(12c^2 N(B,T))$. Therefore, $\frac{\alpha_B(\bomega)}{O_B(\bomega)} < \frac{1}{c}$ under the event $\set{S} \cap \{Z_1 > (1-1/c)s_1\} \cap \set{S}_B$.
\item $Z_1 \leq (1-1/c)s_1$. Define event $\set{S}_{A} = \{N(C,T)-N(C,Q)\geq 1\}$. Then under $\set{S}$ and $\set{S}_A$, since there is a group $C$ arrival and $N(A,T) \geq N(A,Q) \geq T/3$, the optimal solution to \eqref{eq:outcome-benchmark} has $O^\star(\bomega) \geq \min(1/3,N(A,T)/T) = 1/3$ and thus $\min\left(\expectsub{\bomega'}{\bomega'}, O^\star(\bomega)\right) \geq \frac{1-1/(8c)}{3} \geq \frac{1}{4}$ by \eqref{eq:bound-eostar}. But since $Z_1 \leq (1-1/c)s_1$, we have
\begin{align*}
\frac{\sum_{t \in [T]}\sum_{j \in \set{M}} w_{t,j}z_{t,j}}{T} &\leq \frac{(Z_1+N(A,T)-N(A,Q))w_{A,1} + N(A,Q)w_{A,2} + N(B,T)w_{B,1}}{T} \\
&< \frac{(1-1/c)(T/3)+T/6c+T/12c}{T} \leq \frac{1-1/(4c)}{3} \\
&\hspace{-0.5in}\leq \min\left(\expectsub{\bomega'}{\bomega'}, O^\star(\bomega)\right)\frac{1-1/(4c)}{1-1/(8c)} \leq \min\left(\expectsub{\bomega'}{\bomega'}, O^\star(\bomega)\right)(1-1/8c).
\end{align*}
Define \[\set{E} = \left(\set{S} \cap \{Z_1 > (1-1/c)s_1\} \cap \set{S}_B\right) \cup \left(\set{S} \cap \{Z_1 \leq (1-1/c)s_1\} \cap \set{S}_A\right).\]
\end{itemize}
Conditioning on $\set{E}$, either of the two events in the proposition holds since they correspond to the above two scenarios. We next lower bound the probability of $\set{E}$ by
\begin{align*}
\Pr\{\set{E}\} &= \Pr\{\set{S} \cap \{Z_1 > (1-1/c)s_1\} \cap \set{S}_B\} + \Pr\{\set{S} \cap \{Z_1 \leq (1-1/c)s_1\} \cap \set{S}_A\} \\
&= \Pr\{\set{S} \cap \{Z_1 > (1-1/c)s_1\}\}\Pr\{\set{S}_B\} + \Pr\{\set{S} \cap \{Z_1 \leq (1-1/c)s_1\}\}\Pr\{\set{S}_A\} \\
&\geq \Pr\{\set{S}\}\min(\Pr\{\set{S}_B\},\Pr\{\set{S}_A\}) \geq \frac{\min(\Pr\{\set{S}_B\},\Pr\{\set{S}_A\})}{512c^2}
\end{align*}
where the first equality is because the two events have no intersection; the second equality is due to independence between arrivals in $[1,Q]$ and $[Q+1,T]$. Note that 
\[\Pr\{\set{S}_B\} = (1-p_C)^{T-Q} = (1-\ln(16c)/T)^{T/6c} = (1-\ln(16c)/T)^{(T/\ln(16c)) \times (\ln(16c)/6c)} \geq 0.3^{\ln(16c)/6c}
\]
where the last inequality is because $T/\ln(16c) \geq 169$ and $(1-1/y)^y \geq 0.3$ for any $y \geq 6$. In addition,
\begin{align*}
\Pr\{\set{S}_A\} = 1-(1-p_C)^{T-Q} = 1-(1-\ln(16c)/T)^{T/6c} &= 1-(1-\ln(16c)/T)^{(T/\ln(16c)) \times (\ln(16c)/6c)}\\
&\geq 1-e^{-\ln(16c)/6c}
\end{align*}
where the last inequality is because $(1-1/y)^y \leq e^{-1}$ for $y \geq 1$. As a result, 
\[
\Pr\{\set{E}\} \geq \frac{\min(\Pr\{\set{S}_B\},\Pr\{\set{S}_A\})}{512c^2} \geq \frac{\min(0.3^{\ln(16c)/6c},1-e^{-\ln(16c)/6c})}{512c^2},
\]
which completes the proof.
\end{proof}

\section{Distribution-Dependent Vanishing Regret (Section~\ref{sec:bp})}\label{app:bp}
% !TEX root = main.tex
\subsection{Distribution-dependent vanishing regret of \textsc{ABP} (Theorem \ref{thm:abp-dis-dep})}\label{app:abp-dis-dep}

\begin{proof}[Proof of Theorem  \ref{thm:abp-dis-dep}]
We verify \eqref{eq:dis-dep-reg} for \textsc{ABP} under the given fairness rule sequence $\{\set{F}(T)\}_{T \geq 1}.$ Fix $\xi > 0$ and a feature distribution $\set{P} \in \set{C}$. Recall that we denote the maximum regret of $\textsc{ABP}$ for fairness rule $\set{F}(T)$ when there are $T$ arrivals by $\set{R}_{\frule(T)}^{\max}(T,\mathcal{P},\bomega)$ in \eqref{def:max-regret} where $\bomega \in \Theta^T$ has distribution $\set{P}^T$. We now show $\lim_{T \to \infty} \Pr\{\set{R}_{\frule(T)}^{\max}(T,\mathcal{P},\bomega) > \xi\} = 0.$ Recall that $p_{\min} = \min_{g \in \Sgro} p_g(\set{P}).$ Define 
\[
u_T = \sqrt{\frac{\ln (e^3M/\delta(T))}{p_{\min}T}}\left(\frac{1}{\hat{s}_{\min}} + 31\chi(\set{P})\right) + 4\chi(\set{P})\delta(T),
\]
which is the maximum of the upper bounds in Lemmas~\ref{lem:bid-price-global} and \ref{lem:bid-price-group}. Then by Lemmas~\ref{lem:bid-price-global} and \ref{lem:bid-price-group} and the union bound, with probability $1 - 3(G + 1)\delta(T)$, we have $R_T \leq u_T$. Moreover, since $\hat{s}_{\min} > 0, p_{\min} > 0$, $\lim_{T \to \infty} \delta(T) = 0$ and $\lim_{T \to \infty} \ln(\delta(T))/T = 0$, we have $\lim_{T \to \infty} u_T = 0.$ Thus, there exists $\tau > 0$ such that $u_T < \xi$ for any $T \geq \tau$. As a result, for $T \geq \tau$, $\Pr\{R_T > \xi\} \leq \Pr\{\set{R}_{\frule(T)}^{\max}(T,\mathcal{P},\bomega) > u_T \} \leq 1 - 3(G+1)\delta(T)$, which implies $\lim_{T \to \infty} \Pr\{\set{R}_{\frule(T)}^{\max}(T,\mathcal{P},\bomega)> \xi\} = 0$ since $\lim_{T \to \infty} \delta(T) = 0.$
\end{proof}

\subsection{Lower bound of global and ex-post $g-$regrets (Tightness of Lemmas \ref{lem:bid-price-global}-\ref{lem:bid-price-group})} \label{app:abp-lowerbound}
We show that the $\tilde{O}(1/\sqrt{T})$ guarantees of $\textsc{ABP}$ for  global regret and ex-post $g-$regret for groups of linear sizes are optimal up to logarithmic factors. Specifically, there exists a constant $C > 0$ such that there is a setting where for any large enough $T$ and any non-anticipatory algorithm, with probability at least $C$, the global regret \textbf{\emph{and}} the maximum ex-post $g-$regret across large groups are at least $1/(4\sqrt{T})$. Indeed, we show this for global regret relative to $\min\big(O^\star_{\mathrm{maxmin}}(\bomega),\expectsub{\bomega'}{O^\star_{\mathrm{maxmin}}(\bomega')}\big)$ and ex-post $g-$regret relative to $\min\big(O_{g,\mathrm{maxmin}}(\bomega),\expectsub{\bomega'}{O_{g,\mathrm{maxmin}}(\bomega')}\big)$, both of which are weaker benchmarks than the ones global regret and ex-post group-$g$ regret use.

We first need the following lemma bounding the absolute value of the sum of $T$ independent Rademacher random variables, i.e., random variables that are either $-1$ or $1$ with equal probability. 
\begin{lemma}\label{lem:bound-abs-exp}
Let $X$ be the sum of $T$ independent Rademacher random variables. Then $\expect{|X|} \leq~\sqrt{T}$.
\end{lemma}
\begin{proof}
The proof follows standard random walk analysis. Suppose $X=X_1+\cdots+X_T$ where each $X_i$ is a Rademacher random variable. Then $\expect{X^2} = \sum_{i\leq T} \expect{X_i^2} + 2\sum_{1\leq i < j \leq T} \expect{X_iX_j} = T$. We then have
$\expect{|X|} \leq \sqrt{\expect{|X|^2}} = \sqrt{\expect{X^2}} = \sqrt{T}$.
\end{proof}
\begin{proposition}\label{prop:lower-bound}
Fix an even $T \geq (192\Phi(-16))^2$ where $\Phi(x)$ is the cumulative distribution function of the standard normal distribution.  There exists a setting with MaxMin fairness rule, such that for any non-anticipatory algorithm, with probability at least $\frac{(1-3e^{-4})\Phi(-16)}{2}$, the following both hold:
\begin{itemize}
\item $\min\big(O^\star_{\mathrm{maxmin}}(\bomega),\expectsub{\bomega'}{O^\star_{\mathrm{maxmin}}(\bomega')}\big)-\frac{1}{T}\sum_{t=1}^T\sum_{j \in \set{M}}w_{t,j}z_{t,j} \geq 1/(4\sqrt{T})$
\item $\min\big(O_{g,\mathrm{maxmin}}(\bomega),\expectsub{\bomega'}{O_{g,\mathrm{maxmin}}(\bomega')}\big)-\alpha_g \geq 1/(4\sqrt{T})$.
\end{itemize}
\end{proposition}
\begin{proof}
Fix an even $T \geq (192\Phi(-16))^2$. Consider a setting with two locations, $1$ and $2$, each with capacity $T/2$. There are three groups, $A,B,C$, with arrival probabilities $p_A=p_B=0.5-1/\sqrt{T},p_C=2/\sqrt{T}$. For any case $t$, if it is in group $A$, then $w_{t,1}=1,w_{t,2}=0$; if in group $B$, $w_{t,1}=0,w_{t,2}=1$; and if in group $C$, $w_{t,1}=w_{t,2}=1.$ That is, group $A$ favors location $1$; group $B$ favors location $2$ and group $C$ is indifferent between them. We use the MaxMin fairness rule (Example~\ref{ex:maxmin-fairness-rule}). For ease of notation, we  omit the dependence on $\mathrm{maxmin}$ in $O^\star_{\mathrm{maxmin}}$ and $O_{g,\mathrm{maxmin}}$.

We first make some observations about $O^\star(\bomega)$ and $O_g(\bomega)$. The MaxMin fairness rule  assigns as many group $A$ and $B$ arrivals as possible to their preferred locations,  and then assigns group~$C$ arrivals to any remaining capacities. Thus, for any $g$, $O_g(\bomega) = \min\left(1,\frac{0.5T}{\max(N(A,T,\bomega),N(B,T,\bomega))}\right)$ and $O^\star(\bomega) = \min\left(1,1.5 - \frac{\max(N(A,T,\bomega),N(B,T,\bomega))}{T}\right)$. To analyze their expectations, we construct an extended sample path $\tilde{\bomega}$ from $\bomega$, such that in this new path, each group $C$ arrival is independently assigned to either group $A$ or group $B$ with probability $1/2$. Denote the resulting number of group $A$ arrivals by $X_A(\tilde{\bomega})$ and that of group $B$ by $X_B(\tilde{\bomega})$. Note that $X_A(\tilde{\bomega}) \geq N(A,T,\bomega), X_B(\tilde{\bomega}) \geq N(B,T,\bomega)$. In addition, the random variables $X_A, X_B$ are each a Binomial random variable $\mathrm{Bin}(T,0.5)$ with $X_A + X_B = T$. Let $\tilde{X}_A = X_A - T/2$ and we have $2\tilde{X}_A$ to be the sum of $T$ independent Rademacher random variables. Then for any $g$, 
\begin{align*}
\expectsub{\bomega'}{O_g(\bomega')} = \expectsub{\bomega'}{\min\left(1,\frac{0.5T}{\max(N(A,T,\bomega'),N(B,T,\bomega'))}\right)} &\geq \expect{\frac{0.5T}{\max(X_A,X_B)}} \\
&\hspace{-2in}= \expect{\frac{0.5T}{0.5T+|\tilde{X}_A|}}\geq\expect{1-\frac{|\tilde{X}_A|}{0.5T}} \geq \expect{1-\frac{\sqrt{T}}{T}} = 1 - 1/\sqrt{T},
\end{align*}
where the second to last inequality is by Lemma~\ref{lem:bound-abs-exp}. Similarly,
\begin{align*}
\expectsub{\bomega'}{O^\star(\bomega')} = \min\left(1,1.5 - \frac{\max(N(A,T,\bomega),N(B,T,\bomega))}{T}\right) &\geq 1.5 - \frac{\max(X_A,X_B)}{T} \\
&\hspace{-2in} = 1.5 - \frac{0.5T + |\tilde{X}_A|}{T} \geq 1 - 1/(2\sqrt{T}).
\end{align*}
To show the desired lower bound, let $Q = T/2$. Define event $\set{S}$ as
\[
\set{S} = \left\{\min(N(A,Q),N(B,Q))-\expect{N(A,Q)} \geq -2\sqrt{T},N(C,Q)-\frac{2Q}{\sqrt{T}}\geq -4T^{1/4}\right\}.
\]
Applying Hoeffding Inequality (Fact~\ref{fact:hoeffding}) to $N(A,Q), N(B,Q)$ and the Chernoff bound (Fact~\ref{lem:chernoff}) to $N(C,Q)$ gives $\Pr\{\set{S}\} \geq 1 - 3e^{-8}$. Let $\hat{z}_{g,j}(t)$ be the number of group $g$ arrivals assigned to location $j$ among the first $t$ cases. 

Let $Y_A = N(A,T)-N(A,Q), Y_B = N(B,T)-N(B,Q)$ be the number of group $A$ and $B$ arrivals after case $Q$. We know $Y_A, Y_B$ are independent of $\set{S}$ and are each a Binomial random variable $\mathrm{Bin}(Q,0.5-1/\sqrt{T})$. Define $\tilde{Y}_A = \frac{Y_A - \expect{Y_A}}{\sqrt{(1/4-1/T)Q}}$. One can verify that $\expect{\tilde{Y}_A} = 0$ and $\var(\tilde{Y}_A) = 1$. By Berry-Esseen theorem (Fact~\ref{fact:berry-esseen}), for any $y$, $\left|\Pr\{\tilde{Y}_A \geq y\} - (1-\Phi(y))\right| \leq \frac{3}{(\sqrt{1/4-1/T})^3\sqrt{Q}} \leq \frac{96}{\sqrt{T}}$. As a result, 
\begin{equation}\label{eq:prb-YA}
\Pr\{Y_A \geq \expect{Y_A} + 4\sqrt{T}\} = \Pr\left\{\tilde{Y}_A \geq \frac{4\sqrt{T}}{\sqrt{(1/4-1/T)Q}}\right\} \geq \Pr\{\tilde{Y}_A \geq 16\} \geq \Phi(-16)-\frac{96}{\sqrt{T}} \geq \Phi(-16)/2
\end{equation}
since $T \geq (192\Phi(-16))^2$ by assumption. By symmetry, $\Pr\{Y_B \geq \expect{Y_B} + 4\sqrt{T}\} \geq \Phi(-16)/2$. 

Let $\set{S}_1 = \set{S} \cap \{\hat{z}_{C,1}(Q) \geq \sqrt{T}/2-2T^{1/4}\}$ and $\set{S}_2 = \set{S} \setminus \set{S}_1$. Denote the average score across all arrivals by $\alpha=\frac{1}{T}\sum_{t=1}^T\sum_{j \in \set{M}} w_{t,j}z_{t,j}$ and recall that $\alpha_g$ is the average score of group $g$. Consider two cases:
\begin{itemize}
\item $\set{S}_1$ holds and $Y_A \geq \expect{Y_A} + 4\sqrt{T}$. Then $N(A,T) = N(A,Q)+Y_A \geq \expect{N(A,Q)}-2\sqrt{T}+\expect{Y_A}+4\sqrt{T} \geq 0.5T+\sqrt{T}$. Since $N(A,T) \geq 0.5T \geq N(B,T)$, we must have $O_A=\frac{T/2}{N(A,T)} \leq \frac{T/2}{T/2+\sqrt{T}} \leq 1-1/\sqrt{T} \leq \expect{O_A}$. In addition, since $\sqrt{T}/2-2T^{1/4} \geq \sqrt{T}/4$ capacity of location $1$ is assigned to group $C$, the average score of group $A$ is at most $\frac{T/2-\sqrt{T}/4}{N(A,T)}$. Therefore, $\min\left(\expect{O_A},O_A\right) - \alpha_A = O_A - \alpha_A \geq \frac{\sqrt{T}/4}{N(A,T)} \geq 1/(4\sqrt{T})$. In addition, we know $O^\star = 1.5 - \frac{N(A,T)}{T} \leq 1-1/\sqrt{T} \leq \expect{O^\star}.$ But since group $C$ takes at least $\sqrt{T}/4$ capacity from location $1$, the average score of all cases, $\alpha$, is at most $\frac{(T/2)-\sqrt{T}/4+T-N(A,T)}{T}=1.5-\frac{N(A,T)}{T}-1/(4\sqrt{T})$. Therefore, $\min(\expect{O^\star},O^\star) - \alpha = O^\star - \alpha \geq 1 / (4\sqrt{T})$.
\item $\set{S}_2$ holds and $Y_B \geq \expect{Y_B} + 4\sqrt{T}$. Under $\set{S}$, we have $\hat{z}_{C,1}(Q)+\hat{z}_{C,2}(Q) = N(C,Q) \geq \sqrt{T}-4T^{1/4}$. Since $\hat{z}_{C,1}(Q) < \sqrt{T}/2-2T^{1/4}$ under $\set{S}_2$, we have $\hat{z}_{C,2}(Q) \geq \sqrt{T}/2-2T^{1/4}$. Following the same analysis above, we have $\min\left(\expect{O_B},O_B\right) - \alpha_B \geq 1/(4\sqrt{T})$ and $\min\left(\expect{O^\star},O^\star\right) - \alpha \geq 1/(4\sqrt{T})$.
\end{itemize}
Based on these two cases,
\begin{align*}
\Pr&\left\{\min(\expect{O^\star},O^\star)-\alpha \geq 1/(4\sqrt{T}), \max_{g \in \set{G}\colon \expect{N(g,T)} \geq 0.4T} \left(\min(\expect{O_g},O_g)-\alpha_g\right) \geq 1/(4\sqrt{T})\right\} \\
&\geq\Pr\left\{\left(\set{S}_1 \cap \{Y_A \geq \expect{Y_A}+4\sqrt{T}\}\right) \cup \left(\set{S}_2 \cap \{Y_B \geq \expect{Y_B}+4\sqrt{T}\}\right)\right\} \\
&\overset{\set{S}_1 \cap \set{S}_2 = \emptyset}{=} \Pr\{\set{S}_1 \cap \{Y_A \geq \expect{Y_A}+4\sqrt{T}\}\} + \Pr\{\set{S}_2 \cap \{Y_B \geq \expect{Y_B}+4\sqrt{T}\}\} \\
&\overset{\set{S}_1 \perp Y_A, \set{S}_2 \perp Y_B}{=} \Pr\{\set{S}_1\}\Pr\{Y_A \geq \expect{Y_A}+4\sqrt{T}\} + \Pr\{\set{S}_2\}\Pr\{Y_B \geq \expect{Y_B}+4\sqrt{T}\} \\
&\overset{\eqref{eq:prb-YA}}{\geq} \Pr\{\set{S}_1\}\Phi(-16)/2 + \Pr\{\set{S}_2\}\Phi(-16)/2 \\
&= (\Pr\{\set{S}_1\} + \Pr\{\set{S}_2\})\Phi(-16)/2=\Pr\{\set{S}\}\Phi(-16)/2 \geq \frac{(1-3e^{-8})\Phi(-16)}{2}.
\end{align*}
Recalling that $\alpha = \frac{1}{T}\sum_{t=1}^T\sum_{j \in \set{M}} w_{t,j}z_{t,j}$ completes the proof.
\end{proof}

\subsection{Guarantee of global regret for \textsc{ABP} (Lemma \ref{lem:bid-price-global})}\label{app:thm_bp_proof_global}
\begin{proof}[Proof of Lemma~\ref{lem:bid-price-global}]
Since the algorithm selects location $J^{\BP}(t)$ for every case $t \leq \Temp$, we have the average employment score $\sum_{t \in [T]}\sum_{j \in \Sloc} w_{t,j}z_{t,j}$ is lower bounded by $ \frac{1}{T}\sum_{t=1}^{\Temp} w_{t,J^{\BP}(t)} \geq \frac{1}{T}\sum_{t=1}^{T} w_{t,J^{\BP}(t)} - \frac{T-\Temp}{T}$. Due to the independence assumption on the arrival of the cases, by Hoeffding's Inequality (Fact~\ref{fact:hoeffding}), with probability at least $1 - \frac{\delta}{M+2}$, we have $\sum_{t=1}^{T} w_{t,J^{\BP}(t)} \geq \expect{\sum_{t=1}^{T} w_{t,J^{\BP}(t)}} - \sqrt{\frac{1}{2}T\ln\left(\frac{M+2}{\delta}\right)}$ which is lower bounded by $T\expect{O^\star} - \sqrt{\frac{1}{2}T\ln\left(\frac{M+2}{\delta}\right)}$ by Lemma~\ref{lem:bp-property-obj}. Lemma~\ref{lem:bpc-temp} shows that with probability at least $1-\frac{M\delta}{M+2}$, $\Temp \geq T-\Delta^{\BP}$. As a result of union bound, with probability at least $1 - \delta$, $\frac{1}{T}\sum_{t \in [T]}\sum_{j \in \Sloc} w_{t,j}z_{t,j} \geq \expect{O^\star} - \frac{\Delta^{\BP}+\sqrt{\frac{1}{2}T\ln\left(\frac{M+2}{\delta}\right)}}{T}$ and thus $\set{R}_{\set{F}}^{\ABP} = \expect{O^\star} - \frac{1}{T}\sum_{t \in [T]}\sum_{j \in \Sloc} w_{t,j}z_{t,j} \leq \frac{\Delta^{\BP}+\sqrt{\frac{1}{2}T\ln\left(\frac{M+2}{\delta}\right)}}{T}$. Note that if $M = 1$, the regret is always zero and thus we can assume $M \geq 2$ under which $\frac{M+2}{\delta} \leq \left(\frac{M}{\delta}\right)^2$ for $\delta < 1$ (the theorem automatically holds when $\delta \geq 1$). As a result, plugging in the value of $\Delta^{\BP}$, we have $\set{R}_{\set{F}}^{\ABP} \leq \sqrt{\frac{\ln(M/\delta)}{T}}\left(\frac{1}{\hat{s}_{\min}}+1\right)$ with probability at least $1-\delta$.
\end{proof}

\subsection{Guarantee of ex-post g-regret for \textsc{ABP} (Lemma \ref{lem:bid-price-group})}\label{app:thm_bp_proof_group}
To prove the lemma, we first define a weaker regret for which the benchmark is the minimum between the expected minimum requirement and the sample path requirement:
\[
\set{R}_{g,\frule}(\bomega) \triangleq \min\left(\expectsub{\omega'}{O_{g,\frule}(\omega')}, O_{g,\frule}(\omega)\right) - \alpha_g(\bomega).
\]
We show that for this weaker regret, \textsc{ABP} has vanishing regret for any ex-post feasible fairness rule. The proof of Lemma~\ref{lem:bid-price-group} then connects this weaker benchmark with the sample-path minimum requirement via Lemma~\ref{lem:bpc-og-conc}.
\begin{lemma}\label{lem:abp-weak-g-regret}
Fix an ex-post feasible fairness rule $\frule$. For any $\delta>0$ and $g \in \Sgro$, \textsc{Amplified Bid Price Control} has $\set{R}_{g,\set{F}} \leq 
 \sqrt{\frac{\ln(M/\delta)}{T}}\cdot\left(\frac{1}{\hat{s}_{\min}}+\sqrt{\frac{200}{p_g}}\right)$ with probability at least $1-\delta$.
\end{lemma}
\begin{proof}
For ease of notation, we set $\beta = \left(\frac{\delta}{M+2}\right)^{1/4}$ in the following analysis. Fix a group $g \in \Sgro$. To bound $\set{R}_{g,\set{F}}^{\ABP}$, consider the event $\set{S}_1$ where $\Temp \geq T-\Delta^{\BP}$. Condition on $\set{S}_1$, we have $|\set{A}(g,T)|\alpha_g \geq \sum_{t=1}^{T-\Delta^{\BP}} w_{t,J^{\BP}(t)}\indic{g(t) = g}$. Let $X_t = w_{t,J^{\BP}(t)}\indic{g(t) = g}$ for $t \in [T]$. Note that $\{X_t\}_{t \in [T]}$ are independent non-negative random variables with $\expect{X_t^2} \leq p_g$. We can apply a variant of Bernstein's Inequality (see Lemma~\ref{lem:bernstein-implied}) to get that with probability at least $1-\beta^4$, it holds
\[
\sum_{t=1}^{T-\Delta^{\BP}} X_t \geq \expect{\sum_{t=1}^{T-\Delta^{\BP}} X_t} - 2\sqrt{2p_gT\ln(\nicefrac{1}{\beta})}.
\]
Denote the above event by $\set{S}_2$. Note that 
\[
\expect{X_t} = \expect{w_{t,J^{\BP}(t)}\indic{g(t)=g}} = \expect{w_{t,J^{\BP}(t)} \mid g(t) = g}\Pr\{g(t) = g\}.
\]
Lemma~\ref{lem:bp-property-fair} shows that $\expect{w_{t,J^{\BP}(t)} \mid g(t) = g} \geq \expect{O_g} - \sqrt{\frac{1}{p_g T}}$. Therefore, we have $\expect{X_t} \geq p_g\left(\expect{O_g} - \sqrt{\frac{1}{p_g T}}\right)$. As a result of linearity of expectation, under $\set{S}_1 \cap \set{S}_2$, the total score of group $g$ is lower bounded by 
\begin{align*}
|\set{A}(g,T)|\alpha_g &\geq \sum_{t=1}^{T-\Delta^{\BP}} p_g\left(\expect{O_g} - \sqrt{\frac{1}{p_g T}}\right) - 2\sqrt{2p_g T\ln(\nicefrac{1}{\beta})} \\
&\geq (T-\Delta^{\BP})p_g\expect{O_g} - \sum_{t=1}^T p_g\sqrt{\frac{1}{p_g T}} - 2\sqrt{2p_g T\ln(\nicefrac{1}{\beta})} \\
&= (T-\Delta^{\BP})p_g\expect{O_g} - 3\sqrt{2p_g T\ln(\nicefrac{1}{\beta})}.
\end{align*}
By an implication of Chernoff bound (see Lemma~\ref{lem:chernoff}), with probability at least $1 - \beta^4$, the number of group $g$ cases is at most $|\set{A}(g,T)| \leq p_g T + \xi$ where $\xi = 5\sqrt{\max\left(p_g T, \ln(\nicefrac{1}{\beta})\right)\ln(\nicefrac{1}{\beta})}$. Denote this event by $\set{S}_3$. Then under $\set{S}_1 \cap \set{S}_2 \cap \set{S}_3$, we have 
\[
\alpha_g \geq \frac{(T-\Delta^{\BP})p_g\expect{O_g}-3\sqrt{2p_gT\ln(\nicefrac{1}{\beta})}}{p_g T + \xi}\geq \frac{(T-\Delta^{\BP})p_g\expect{O_g}}{p_g T + \xi} - 3\sqrt{\frac{2\ln(\nicefrac{1}{\beta})}{p_g T}}.
\]
Now note that 
\[
\frac{(T-\Delta^{\BP})p_g\expect{O_g}}{p_g T + \xi} = \left(1-\frac{\Delta^{\BP}p_g+\xi}{p_g T + \xi}\right)\expect{O_g} \geq \expect{O_g} - \frac{\Delta^{\BP}}{T} - \frac{\xi}{p_g T}
\]
and $\frac{\xi}{p_g T} \leq 5\sqrt{\frac{\ln(\nicefrac{1}{\beta})}{p_g T}}+\frac{5\ln(\nicefrac{1}{\beta})}{p_g T}$. As a result, under $\set{S}_1 \cap \set{S}_2 \cap \set{S}_3$, $\alpha_g$ is lower bounded by
\[
\expect{O_g} - \frac{\Delta^{\BP}}{T} - 5\sqrt{\frac{\ln(\nicefrac{1}{\beta})}{p_g T}}+\frac{5\ln(\nicefrac{1}{\beta})}{p_g T} - 3\sqrt{\frac{2\ln(\nicefrac{1}{\beta})}{p_g T}} \geq \expect{O_g} - \frac{\Delta^{\BP}}{T} - 10\sqrt{\frac{\ln(\nicefrac{1}{\beta})}{p_g T}} - \frac{5\ln(\nicefrac{1}{\beta})}{p_g T}
\]
As a result, $\expect{O_g} - \alpha_g \leq 
 \sqrt{\frac{2\ln(1/\beta)}{T}}\frac{1}{\hat{s}_{\min}} + 10\sqrt{\frac{\ln(1/\beta)}{p_gT}} + \frac{5\ln(1/\beta)}{p_g T}.$ 

 Note that $\set{R}_{g,\frule}$ is always upper bounded by $1$. Therefore,
 \begin{align*}
 \set{R}_{g,\frule} &\leq \sqrt{\frac{2\ln(1/\beta)}{T}}\frac{1}{\hat{s}_{\min}} + \min\left(10\sqrt{\frac{\ln(1/\beta)}{p_gT}} + \frac{5\ln(1/\beta)}{p_g T}, 1\right) \\
 &\leq \sqrt{\frac{2\ln(1/\beta)}{T}}\frac{1}{\hat{s}_{\min}} + 20\sqrt{\frac{\ln(1/\beta)}{p_gT}}= \sqrt{\frac{2\ln(1/\beta)}{T}}\left(\frac{1}{\hat{s}_{\min}} + \sqrt{\frac{200}{p_g}} \right)
 \end{align*}
where the second inequality is due to the fact that if $\frac{5\ln(1/\beta)}{p_g T} > 10\sqrt{\frac{\ln(1/\beta)}{p_g T}}$, we would have $10\sqrt{\frac{\ln(1/\beta)}{p_g T}} \geq 20$ and thus the minimum term is equal to $1$. Note that $\Pr\{\set{S}_1\} \geq 1 - M\beta^4$ by Lemma~\ref{lem:bpc-temp} and $\Pr\{\set{S}_2\}, \Pr\{\set{S}_3\} \geq 1 - \beta^4$. By the union bound, we have $\Pr\{S_1 \cap S_2 \cap S_3\} \geq 1 - (M+2)\beta^4$. As a result, for any $\beta > 0$, with probability at least $1 - (M+2)\beta^4$, we have $\set{R}_{g,\frule} \leq \sqrt{\frac{2\ln(1/\beta)}{T}}\left(\frac{1}{\hat{s}_{\min}} + \sqrt{\frac{200}{p_g}} \right)$.

Since $\beta = (\delta / (M+2))^{1/4}$ and we can safely assume $M \geq 2,\delta < 1$ (otherwise the theorem automatically holds), we have $\sqrt{\frac{2\ln(1/\beta)}{T}} \leq \sqrt{\frac{\ln(M/\delta)}{T}}$ and $1 - (M+2)\beta^4 = 1 - \delta$. As a result, with probability at least $1-\delta$, $\set{R}_{g,\frule} \leq \sqrt{\frac{\ln(M/\delta)}{T}}\left(\frac{1}{\hat{s}_{\min}} + \sqrt{\frac{200}{p_g}} \right)$.
\end{proof}

We prove Lemma~\ref{lem:bid-price-group} by combining Lemmas~\ref{lem:abp-weak-g-regret} and \ref{lem:bpc-og-conc}.
\begin{proof}[Proof of Lemma~\ref{lem:bid-price-group}]
Fix $\delta \in (0,1), g \in \Sgro$. Applying Lemma~\ref{lem:abp-weak-g-regret} for $\delta$, we have with probability at least $1-\delta$ that $\set{R}_{g,\frule}^{\ABP}(\bomega) \leq \sqrt{\frac{\ln(M/\delta)}{T}}\left(\frac{1}{\hat{s}_{\min}}+\sqrt{\frac{200}{p_g}}\right)$. By Lemma~\ref{lem:bpc-og-conc}, with probability $1 - 2\delta$,\begin{equation}\label{eq:bound-og-abp-expost}
O_{g,\frule}(\bomega) \leq \expectsub{\bomega'}{O_{g,\frule}(\bomega')}+16\chi\sqrt{\frac{\ln\nicefrac{e^2}{\delta}}{p_g T}}+ 4\chi T \delta.
\end{equation} Then by union bound, with probability $1-3\delta$, we have
\begin{align*}
\set{R}_{g,\frule}^{\mathrm{ex},\ABP}(\bomega) &= O_{g,\frule}(\bomega) - \alpha_g(\bomega) \leq \min\left(\expectsub{\bomega'}{O_{g,\frule}(\bomega')}, O_{g,\frule}(\bomega)\right)+16\chi\sqrt{\frac{\ln(e^2/\delta)}{p_g T}} + 2\chi \delta - \alpha_g(\bomega) \\
&= \set{R}_{g,\frule}^{\ABP}(\bomega)+16\chi\sqrt{\frac{\ln(e^2/\delta)}{p_g T}} + 2\chi \delta - \alpha_g(\bomega)\\
&\hspace{-0.8in}\leq \sqrt{\frac{\ln(M/\delta)}{T}}\left(\frac{1}{\hat{s}_{\min}}+\sqrt{\frac{200}{p_g}}\right) +16\chi\sqrt{\frac{\ln(e^2/\delta)}{p_g T}} + 2\chi \delta\leq \sqrt{\frac{\ln (e^2M/\delta)}{p_g T}}\left(\frac{1}{\hat{s}_{\min}} + 31\chi\right) + 2\chi \delta.
\end{align*}
\end{proof}

\subsection{Guarantee for \textsc{ABP} without capacity constraints (Lemmas~\ref{lem:bp-property-obj},~\ref{lem:bp-property-fair})}\label{app:bp-property}
The proofs rely on KKT conditions of solving the convex optimization problem for $\bolds{\mu}^\star,\bolds{\lambda}^\star$ with similar ideas from \cite{talluri1998analysis}. Recall the minimization problem in line~\ref{algoline:set-lagra-multi} in \textsc{Amplified Bid Price Control} that solves $\bolds{\mu}^\star,\bolds{\lambda}^\star.$ We can write out its Lagrangian defined by 
\begin{equation}\label{eq:lagrange-second}
\begin{aligned}
L(\bolds{\mu},\bolds{\lambda},\bolds{u}^1,\bolds{u}^2) &= \sum_{t=1}^T\expectsub{\bolds{\theta}_t}{\max_{j \in \Sloc} ((1+\lambda_{g(\bolds{\theta}_t)})w_j(\bolds{\theta}_t) - \mu_j)}+ \sum_{j \in \Sloc} \mu_j s_j -\sum_{g \in \Sgro} \lambda_g \expect{O_gN(g,T)}\\
&\mspace{32mu}-\sum_{j\in \Sloc} \mu_j u^1_j - \sum_{g \in \Sgro} \lambda_g u^2_g,
\end{aligned}
\end{equation}
where $\bolds{u}^1 \in \mathbb{R}_+^M, \bolds{u}^2 \in \mathbb{R}_+^G$ are associated Lagrange multipliers. Note that $L(\bolds{\mu},\bolds{\lambda},\bolds{u}^1,\bolds{u}^2)$ is continuously differentiable because the number of groups is finite and condition on a group, the score distribution is continuous. By the Karush–Kuhn–Tucker (KKT) conditions, a necessary condition for $\bolds{\mu}^\star,\bolds{\lambda}^\star$ to be optimal is that the partial derivatives, $\frac{\partial L}{\partial \bolds{\mu}}, \frac{\partial L}{\partial \bolds{\lambda}}$ are both zero.  We then have for every $j \in \Sloc, g \in \Sgro$,
\begin{equation}\label{eq:kkt-feasible}
-\sum_{t\in[T]}\Pr\{J^{\BP}(t)=j\}+s_j - u_j^{1} = 0;~ \sum_{t\in[T]}p_g\expect{w_{t,J^{\BP}(t)} \mid g(t) = g}-\expect{O_g N(g,T)}-u_g^2 = 0.
\end{equation}

\begin{proof}[Proof of Lemma~\ref{lem:bp-property-fair}]
Since $u_g^2 \geq 0$ and $\expect{w_{t,J^{\BP}(t)} \mid g(t) = g}$ are the same across $t$ (we assume that refugees' features are i.i.d.), we have $\expect{w_{t,J^{\BP}(t)} \mid g(t) = g} \geq \frac{1}{p_g T}\expect{O_g N(g,T)}$ for any $t$ by \eqref{eq:kkt-feasible} . Furthermore, by Lemma~\ref{lem:ep-qg-og}, $\frac{1}{p_g T}\expect{O_gN(g,T)} \geq \expect{O_g} - \sqrt{\frac{1}{p_g T}}$
and thus $\expect{w_{t,J^{\BP}(t)} \mid g(t) = g} \geq \expect{O_g} - \sqrt{\frac{1}{p_g T}}$ for any $t$.
\end{proof}
\begin{proof}[Proof of Lemma~\ref{lem:bp-property-obj}] Recall the Lagrangian relaxation in \eqref{eq:lagrangian}. By definition of $L(\bolds{\mu}^\star,\bolds{\lambda}^\star)$, we have
\begin{align}
\expect{L(\bolds{\mu}^\star,\bolds{\lambda}^\star)} &= \expect{\sum_{t=1}^T \max_{j \in \Sloc} ((1+\lambda^\star_{g(t)})w_{t,j} - \mu^\star_j) + \sum_{j \in \Sloc} \mu^\star_j s_j -\sum_{g \in \Sgro} \lambda^\star_g O_gN(g,T)} \nonumber \\
&= \sum_{t=1}^T \expect{(1+\lambda^\star_{g(t)})w_{t,J^{\BP}(t)}-\mu^\star_{J^{\BP}(t)}} + \sum_{j \in \Sloc} \mu^\star_j s_j - \sum_{g \in \Sgro} \lambda^\star_g \expect{O_gN(g,T)} \nonumber\\
&= \sum_{t \in [T]}\expect{w_{t,J^{\BP}(t)}} + \sum_{j \in \Sloc} \mu_j^\star (s_j - \sum_{t\in[T]}\Pr\{j = J^{\BP}(t)\}) \nonumber\\
&\mspace{32mu}- \sum_{g \in \Sgro} \lambda_g^\star \left(\expect{O_gN(g,T)}-\sum_{t\in[T]}p_g\expect{w_{t,J^{\BP}(t)} \mid g(t)=g}\right). \label{eq:lagrangian-expand}
\end{align}
where the second equality uses linearity of expectation; the third equality uses linearity of expectation again and the i.i.d. assumption of arrivals. Recall that $(\boldsymbol{\mu}^\star, \boldsymbol{\lambda}^\star)$ is an optimal solution to the constrained optimization of $\min \expect{L(\bolds{\mu},\bolds{\lambda})}$. We now apply the complementary slackness in KKT conditions on the Lagrangian multipliers \eqref{eq:lagrange-second}. For every location $j \in \Sloc$ and group $g \in \Sgro$, we have that $\mu_j^\star u_j^1 = 0, \lambda_g^\star u_g^2 = 0$. Combining this fact with \eqref{eq:kkt-feasible}, we have that $\forall j \in \Sloc, g \in \Sgro,t \in [T]$,
\begin{align*}
\mu_j^\star (s_j - \sum_{t\in[T]}\Pr\{j = J^{\BP}(t)\}) = 0,~\lambda_g^\star \left(\expect{O_gN(g,T)}-\sum_{t\in[T]} p_g\expect{w_{t,J^{\BP}(t)} \mid g(t)=g}\right)=0.
\end{align*}
Putting it back to \eqref{eq:lagrangian-expand} gives us that $\sum_{t\in[T]}\expect{w_{t,J^{\BP}(T)}} = \expect{L(\bolds{\mu}^\star,\bolds{\lambda}^\star)}$. Since the Lagrangian satisfies $L(\bolds{\mu}^\star,\bolds{\lambda}^\star) \geq TO^\star$ for all sample paths, we also have $\expect{L(\bolds{\mu}^\star,\bolds{\lambda}^\star)} \geq \expect{TO^\star}$, which finishes the proof.
\end{proof}

\subsection{Lower bound $\Temp$ under \textsc{Amplified Bid Price Control} (Lemma~\ref{lem:bpc-temp})}\label{app:bpc-temp}
We first give the following lemma on the probability that a case is sent to a particular location under $J^{\BP}$. It is proved using the KKT condition.
\begin{lemma}\label{lem:bp-property-cap}
For every $j \in \Sloc$ and $t \in [T]$, we have $\sum_{\tau=1}^t \Pr\{J^{\BP}(t) = j\} \leq \hat{s}_j t$.
\end{lemma}
\begin{proof}
Since refugees' features are i.i.d., we can define $q_g(j) = \Pr\{J^{\BP}(t) = j | g(t) = g\}$ which is the same for any $t$. Since $u_j^1 \geq 0$, we immediately have $\sum_{t=1}^T\sum_{g \in \Sgro} p_g q_g(j) \leq s_j$ by \eqref{eq:kkt-feasible}. Therefore, $\sum_{g \in \Sgro}q_g(j)\sum_{t=1}^T p_g \leq s_j$, which implies $\sum_{g \in \Sgro} q_g(j)p_g \leq \hat{s}_j$. Then for any fixed $t \in [T]$, we have
\[
\sum_{\tau=1}^t \Pr\{J^{\BP}(t) = j\} = \sum_{g\in \Sgro}q_g(j)\sum_{\tau=1}^t p_g \leq \sum_{g\in \Sgro}q_g(j)p_g t\leq \hat{s}_jt.
\]
\end{proof}
\begin{proof}[Proof of Lemma~\ref{lem:bpc-temp}]
Recall that $a_j(t)$ is the consumed capacity in location $j$ in the first $t$ periods. Then~$\Temp = \min_t \{\exists j \in \Sloc, a_j(t) = s_j(t)\}.$ Fix a location $j \in \Sloc$. Let us set $b_j(t) = 1$ if $J^{\BP}(t) = j.$ Since \textsc{Amplified Bid Price Control} follows $J^{\BP}$ in the first $\Temp$ periods, we have for all $t \leq \Temp$, $a_j(t) = \sum_{\tau=1}^t b_j(\tau).$ Notice that $\{b_j(\tau),\tau \in [T]\}$ are independent Bernoulli random variables with $\expect{b_j(\tau)} = \Pr\{J^{\BP}(t) = j\}$. Using Fact~\ref{fact:hoeffding}, for a fixed case $t$, we have
\begin{equation}\label{eq:deviation-capacity}
\Pr\left\{\sum_{\tau=1}^t b_j(\tau) > \sum_{\tau=1}^t \Pr\{J^{\BP}(t) = j\}+ \sqrt{\frac{1}{2}t\ln\left(\frac{M+2}{\delta}\right)}\right\} \leq \frac{\delta}{M+2}.
\end{equation}
Consider time $T_1 = T -\frac{1}{\hat{s}_{\min}}\sqrt{\frac{1}{2}T\ln\left(\frac{M+2}{\delta}\right)}.$ Applying \eqref{eq:deviation-capacity}, Lemma~\ref{lem:bp-property-cap} and the union bound over all $j \in \Sloc$, with probability at least $1-\frac{M\delta}{M+2},$ for every $j \in \Sloc$, we have
\begin{align*}
\sum_{\tau=1}^{T_1} b_j(\tau) \leq \hat{s}_jT_1+\sqrt{\frac{1}{2}T\ln\left(\frac{M+2}{\delta}\right)}&\leq s_j - \sqrt{\frac{1}{2}T\ln\left(\frac{M+2}{\delta}\right)} + \sqrt{\frac{1}{2}T\ln\left(\frac{M+2}{\delta}\right)} \\
&=s_j,
\end{align*}
As a result, with probability at least $1-\frac{M}{M+2}\delta$, we have $\Temp \geq T_1=T -\frac{1}{\hat{s}_{\min}}\sqrt{\frac{1}{2}T\ln\left(\frac{M+2}{\delta}\right)}.$
\end{proof}

\subsection{Concentration for fairness rules with low sensitivity (Lemma~\ref{lem:bpc-og-conc})} \label{app:bpc-og-conc}

The proof of the desired result is as follows: we first show a concentration bound of $Q_g$ when the fairness rule satisfies Definition~\ref{def:irr-rule} over the entire sample space via a refined McDiarmid's Inequality. We then extend it to the high-probability case by an extension argument in \cite{combes2015extension}. The final step is to connect the concentration of $Q_g$ to the concentration of $O_g.$

Our first lemma is a refined version of the McDiarmid Inequality \cite{mcdiarmid1998concentration}. It shows that if a random variable satisfies the sensitivity property in Definition~\ref{def:irr-rule} for the entire sample space, its value centers around its expectation within a confidence bound of size $\sqrt{p_g T}$.
\begin{lemma}\label{lem:refined-mcd}
Given a function $f$ defined over the sample space, suppose that for some $\chi > 0, g \in \Sgro$, 
\begin{equation}\label{eq:sensitivity}
|f(\bomega) - f(\tilde{\bomega})| \leq \left\{
\begin{aligned}
    \sqrt{p_g(\set{P})}\chi, & \text{ if }g(\btheta_t) \neq g\text{ and } g(\tilde{\btheta}_t) \neq g \\
    \chi, & \text{ otherwise,}
\end{aligned}\right.
\end{equation}
for any pair of $\bomega,\tilde{\bomega} \in \Omega$ that differ in at most one arrival $t$ with feature $\btheta_t$ and $\tilde{\btheta}_t$. Then we have 
\[
\Pr_{\bomega \sim \set{P}^T}\left\{f(\bomega) \geq \expectsub{\bomega'}{f(\bomega)} + 5\chi\sqrt{\max\left(2p_g(\set{P})T,\ln\frac{1}{\beta}\right)}\ln\frac{1}{\beta}\right\} \leq \beta^4~ \text{for any }\beta > 0.
\]
\end{lemma}
\begin{proof}
The proof applies the Bernstein type McDiarmid Inequality stated in Fact~\ref{fact:mcdiarmid-bernstein}. Recall the definition of function $h$ in Appendix~\ref{sec:mcdiarmid} where for any $1\leq t \leq T, \bolds{x} = (x_1,\ldots,x_T) \in \Omega$, we define $h_t(x_1,\ldots,x_{t}) = \expectsub{\bomega}{f(\bomega) \mid \btheta_i = x_i, i \leq t} - \expectsub{\bomega}{f(\bomega) \mid \btheta_i = x_i, i \leq t-1}$.

To apply Fact~\ref{fact:mcdiarmid-bernstein}, we first show that the value of $\mathrm{maxdev}^+$, defined by $\sup_{\bolds{x} \in \Omega}\max_{t} h_t(x_1,\ldots,x_{t})$, is upper bounded by $\chi.$ To see this, for any $\bolds{x} \in \Omega, t\in [T]$, we have
\begin{align}
h_t(x_1,\ldots,x_{t}) &= \expectsub{\btheta_t,\ldots,\btheta_T}{f(x_1,\ldots,x_{t-1},x_t,\btheta_{t+1},\ldots,\btheta_T) - f(x_1,\ldots,x_{t-1},\btheta_t,\btheta_{t+1},\ldots,\btheta_T)} \label{eq:rewrite-h}\\
&\leq \expectsub{\btheta_t,\ldots,\btheta_T}{\chi} = \chi,\nonumber
\end{align}
where the inequality is because the two sequences in the expectation differ for at most one case and thus the assumed sensitivity bound holds. 

We next need to bound the variance of $h_t$. Recall from Appendix~\ref{sec:mcdiarmid} that for $t \leq T$ and the sequence $\bolds{x}_{t-1} = (x_1,\ldots,x_{t-1})$, $\var_t(\bolds{x}_{t-1})$ is defined by the variance of $h_t((\bolds{x}_{t-1},X_t))$ where $X_t$ is a independent random variable with the same distribution of $\btheta_t$. Note that $\expectsub{X_t}{h_t(\bolds{x}_{t-1},X_t)} = 0$. By the rewrite of $h_t$ in \eqref{eq:rewrite-h}, we then have
\begin{align}
\var_t(\bolds{x}_{t-1}) &= \expectsub{X_t}{h_t(\bolds{x}_{t-1},X_t)^2} \nonumber\\
&=\expectsub{X_t}{\left(\expectsub{\btheta_t,\ldots,\btheta_T}{f(\bolds{x}_{t-1},X_t,\btheta_{t+1},\ldots,\btheta_T) - f(\bolds{x}_{t-1},\btheta_t,\btheta_{t+1},\ldots,\btheta_T)}\right)^2}. \label{eq:expand-var}
\end{align}
Denote the set of features corresponding to group $g$ by $\Theta_g = \{\btheta \in \Theta\colon g(\btheta) = g\}$. By the assumed sensitivity condition, if $X_t \not \in \Theta_g$ and $\btheta_t \not \in \Theta_g$, we must have $f(\bolds{x}_{t-1},X_t,\btheta_{t+1},\ldots,\btheta_T) - f(\bolds{x}_{t-1},\btheta_t,\btheta_{t+1},\ldots,\btheta_T) \leq \sqrt{p_g(\set{P})}\chi$. Otherwise, it is upper bounded by $\chi$. Using \eqref{eq:expand-var} and conditioning on either $X_t\in\Theta_g$ or $X_t\not\in\Theta_g$ gives
\begin{align*}
\var_t(\bolds{x}_{t-1}) &\leq \Pr\{X_t \in \Theta_g\}\chi^2 + \Pr\{X_t \not \in \Theta_g\}\left(\Pr\{\btheta_t \not \in \Theta_g\}\times \sqrt{p_g(\set{P})}\chi  + \Pr\{\btheta_t \in \Theta_g\}\chi\right)^2 \\
&\leq p_g(\set{P})\chi^2+2(1-p_g(\set{P}))(p_g(\set{P})+p_g(\set{P})^2)\chi^2 \leq p_g(\set{P})\chi^2\left(1 + 1 - p_g(\set{P})^2\right)  \leq 2p_g(\set{P})\chi^2.
\end{align*}
Therefore, $\hat{v} = \sup_{\bolds{x} \in \Omega} V(\bolds{x}) = \sup_{\bolds{x} \in \Omega}\sum_{t \in [T]} \var_t(\bolds{x}_{t-1}) \leq 2\sum_{t\in[T]}p_g(\set{P})\chi^2 = 2p_g(\set{P}) T \chi^2$.

For any $\beta > 0$, applying Fact~\ref{fact:mcdiarmid-bernstein} to $f$ with $d = 5\chi\sqrt{\max\left(2p_g(\set{P}) T,\ln\frac{1}{\beta}\right)\ln\frac{1}{\beta}}$ and using $\mathrm{maxdev}^+ \leq \chi,\hat{v} \leq 2p_g(\set{P}) T \chi^2$ gives
\begin{align}
\Pr_{\bomega \sim \set{P}^T}\{f(\bomega) \geq \expectsub{\bomega'}{f(\bomega')} + d\} &\leq \exp\left(-\frac{d^2}{2(\hat{v}+\mathrm{maxdev}^+ d / 3)}\right) \nonumber\\
&\leq \exp\left(\frac{-25\chi^2\max\left(2p_g(\set{P}) T,\ln\frac{1}{\beta}\right)\ln\frac{1}{\beta}}{4p_g(\set{P}) T \chi^2 + 4\chi^2\sqrt{\max\left(2p_g(\set{P}) T, \ln\frac{1}{\beta}\right)\ln\frac{1}{\beta}}}\right) \nonumber\\
&\leq \exp\left(\frac{-25\max\left(2p_g(\set{P}) T,\ln\frac{1}{\beta}\right)}{6\max\left(2p_g(\set{P}) T,\ln\frac{1}{\beta}\right)}\right) \leq \beta^4 \nonumber
\end{align}
where the second-to-last inequality is because $\sqrt{\max(a,b)b} \leq \max(a,b)$ for any $a,b \geq 0$.
\end{proof}
Our sensitivity condition in Definition~\ref{def:irr-rule} does not need to hold for the entire sample space as required by Lemma~\ref{lem:refined-mcd}. We next relax the requirement of Lemma~\ref{lem:refined-mcd}. To do so, for each group $g$ let us define a distance $d_g(\bomega, \tilde{\bomega})$ for any pair of $\bomega = (\btheta_1,\ldots,\btheta_T), \tilde{\bomega}=(\tilde{\btheta}_1,\ldots,\tilde{\btheta}_T) \in \Omega$ by 
\begin{equation}\label{eq:def-distance}
d_g(\bomega, \tilde{\bomega}) = \chi \sum_{t = 1}^T \indic{\btheta_t \neq \tilde{\btheta}_t}\left(1 + \left(\sqrt{p_g(\bolds{P})} - 1\right)\indic{g(\btheta_t) \neq g, g(\tilde{\btheta}_t) \neq g }\right).
\end{equation}
\begin{lemma}\label{lem:metric-space}
If $\chi > 0$ and $p_g(\set{P}) > 0$, then $d_g(\Omega, d)$ is a metric space.
\end{lemma}
\begin{proof}
To see why $(\Omega,d)$ is a metric space, we observe that (1) $d(\bomega,\bomega) = 0$ for any $\bomega \in \Omega$; (2) $d(\bomega,\tilde{\bomega}) > 0$ if $\bomega \neq \tilde{\bomega}$; (3) $d(\bomega, \tilde{\bomega}) = d(\tilde{\bomega}, \bomega)$. It remains to prove the triangle inequality: $d(\bomega,\hat{\bomega}) \leq d(\bomega, \tilde{\bomega}) + d(\tilde{\bomega}, \hat{\bomega})$ for any $\bomega, \tilde{\bomega}, \hat{\bomega} \in \Omega.$ To prove this property, suppose $\hat{\bomega} = (\hat{\btheta}_1,\ldots,\hat{\btheta}_T)$. The triangle inequality is implied by showing that for each $t \leq T$,
\begin{equation}\label{eq:distance}
\begin{aligned}
&\indic{\btheta_t \neq \hat{\btheta}_t}\left(1 + \left(\sqrt{p_g(\bolds{P})} - 1\right)\indic{g(\btheta_t) \neq g, g(\hat{\btheta}_t) \neq g }\right) \\
&\leq \indic{\btheta_t \neq \tilde{\btheta}_t}\left(1 + \left(\sqrt{p_g(\bolds{P})} - 1\right)\indic{g(\btheta_t) \neq g, g(\tilde{\btheta}_t) \neq g }\right) \\
&+ \indic{\tilde{\btheta}_t \neq \hat{\btheta}_t}\left(1 + \left(\sqrt{p_g(\bolds{P})} - 1\right)\indic{g(\tilde{\btheta}_t) \neq g, g(\hat{\btheta}_t) \neq g }\right).
\end{aligned}
\end{equation}
Consider two cases:
\begin{itemize}
\item When $\btheta_t \neq \tilde{\btheta}_t$ and $\tilde{\btheta}_t \neq \hat{\btheta}_t$, the above inequality is equivalent to show that $1 + (\sqrt{p_g}-1)\indic{g(\btheta_t) \neq g, g(\hat{\btheta}_t) \neq g } \leq 2 + (\sqrt{p_g}-1)\left(\indic{g(\btheta_t) \neq g, g(\tilde{\btheta}_t) \neq g} + \indic{g(\tilde{\btheta}_t) \neq g, g(\hat{\btheta}_t) \neq g}\right).$ This is true because the left hand side is at most $1$ and the right hand side is at least $1+\sqrt{p_g}$ unless we have $g(\btheta_t) \neq g, g(\tilde{\btheta}_t) \neq g$ and $g(\hat{\btheta}_t) \neq g$. But in the latter case, the right hand side is $2\sqrt{g}$ and the left hand side is $\sqrt{g}$.
\item Otherwise, we either have $\btheta_t = \tilde{\btheta}_t$ or $\tilde{\btheta}_t = \hat{\btheta}_t$, and \eqref{eq:distance} becomes an equality.
\end{itemize}
Summarizing the two cases shows the triangle inequality for $d$ and thus $(\Omega, d)$ is a metric space.
\end{proof}
Having established the metric space, we can then use an extension argument from \cite{McSHANE1934ExtensionOR,combes2015extension} to show that the result of Lemma~\ref{lem:refined-mcd} also holds when \eqref{eq:sensitivity} is guaranteed only for a high-probability event. A similar result is also obtained in \cite[Theorem~2]{warnke2016method}. 
\begin{lemma}\label{lem:ext-mcd}
Given a function $f$ defined over the sample space $\Omega = \btheta^T$ with $|f(\bomega)| \leq T$ for any $\bomega \in \Omega$, suppose that there is an event $\set{B} \subseteq \Omega$, a constant $\chi > 0$ and a group $g \in \Sgro$, such that $f$ satisfies \eqref{eq:sensitivity} for any pair of $\bomega,\tilde{\bomega} \in \set{B}$ that differ in at most one arrival. Then for any $\beta > 0$,
\[
\Pr_{\bomega \sim \set{P}^T}\left\{f(\bomega) \geq \expectsub{\bomega'}{f(\bomega') \mid \bomega' \in \set{B}} + 5\chi\sqrt{\max\left(2p_g(\set{P})T,\ln\frac{1}{\beta}\right)}\ln\frac{1}{\beta} + \Pr\{\set{B}^c\}T\chi\right\} \leq \Pr\{\set{B}^c\} + \beta^4.
\]
\end{lemma}

\begin{proof}
Our proof is similar to that of \cite[Proposition~2]{combes2015extension}, and we include it for completeness. One can verify that under this lemma's assumption, we have $|f(\bomega) - f(\tilde{\bomega})| \leq d_g(\bomega, \tilde{\bomega})$ for any $\bomega, \tilde{\bomega} \in \set{B}$, i.e., $f$ is 1-Lipschitz with respect to $d_g$ within $\set{B}$ where $d_g$ is the distance defined in \eqref{eq:def-distance}. Now define 
$f_{\circ}(\bomega) = 
\inf_{\bomega' \in \set{B}} \{f(\bomega') + d_g(\bomega',\bomega)\}.$ Since $(\Omega,d_g)$ is a metric space by Lemma~\ref{lem:metric-space} and $f$ is $1-$Lipschitz within $\set{B}$, we have $f_{\circ}(\bomega)$ is $1-$Lipschitz within $\Omega$ following the proof of \cite[Theorem~1]{McSHANE1934ExtensionOR}. As a result, $f_\circ$ satisfies condition \eqref{eq:sensitivity} required by Lemma~\ref{lem:refined-mcd}. Denoting $x = 5\chi\sqrt{\max\left(2p_g(\set{P})T,\ln\frac{1}{\beta}\right)}$, Lemma~\ref{lem:refined-mcd} then shows that $\Pr\{f_{\circ}(\bomega) \geq \expect{f_{\circ}(\bomega}) + x\} \leq \beta^4.$

Connecting back to the desired result, note that 
\begin{align}
\Pr\{f(\bomega) \geq \expect{f_{\circ}(\bomega)} + x\} &\leq \Pr\{f(\bomega) \geq \expect{f_{\circ}(\bomega)} + x, \bomega \in \set{B}\} + \Pr\{\bomega \not \in \set{B}\} \nonumber\\
&= \Pr\{f_{\circ}(\bomega) \geq \expect{f_{\circ}(\bomega)} + x, \bomega \in \set{B}\} + \Pr\{\bomega \not \in \set{B}\} \nonumber\\
&\leq \Pr\{f_{\circ}(\bomega) \geq \expect{f_{\circ}(\bomega)} + x\} + \Pr\{\bomega \not \in \set{B}\} \leq \beta^4 + \Pr\{\set{B}^c\},\label{eq:f-fo}
\end{align}
where the equality is because $f_{\circ}(\bomega) = f(\bomega)$ when $\bomega \in \set{B}.$ The last step towards the desired result is by noting $f_{\circ}(\bomega) \leq \expect{f(\bomega) \mid \bomega \in \set{B}} + \sup_{\bomega'} d_g(\bomega,\bomega') \leq \expect{f(\bomega) \mid \bomega \in \set{B}} + T\chi$ and thus 
\begin{align}
\expect{f_{\circ}(\bomega)} &= \expect{f_{\circ}(\bomega) \mid \bomega\in \set{B}}\Pr\{\set{B}\} + \expect{f_{\circ}(\bomega) \mid \bomega \not\in \set{B}}\Pr\{\set{B}^c\} \nonumber\\
&\leq  \expect{f_{\circ}(\bomega) \mid \bomega\in \set{B}}\Pr\{\set{B}\} + \left(\expect{f_{\circ}(\bomega) \mid \bomega\in \set{B}} + T\chi\right)\Pr\{\set{B}^c\} \nonumber\\
&= \expect{f_{\circ}(\bomega) \mid \bomega\in \set{B}} + T\chi\Pr\{\set{B}^c\} = \expect{f(\bomega) \mid \bomega\in \set{B}} + T\chi\Pr\{\set{B}^c\}.\label{eq:up-fo}
\end{align}
Plugging \eqref{eq:up-fo} into the left hand side of \eqref{eq:f-fo} completes the proof.
\end{proof}
Fix a fairness rule $\frule$; the following result helps connect $\expect{Q_{g,\frule}}=\expect{O_{g,\frule}N(g,T)}$ with $\expect{O_{g,\frule}}$.
\begin{lemma}\label{lem:ep-qg-og}
For $g \in \Sgro$, we have $\expect{O_{g,\frule}} - \sqrt{\frac{1}{p_g T}} \leq \frac{1}{p_g T}\expect{O_{g,\frule}N(g,T)} \leq \expect{O_{g,\frule}} + \sqrt{\frac{1}{p_g T}}$.
\end{lemma}
\begin{proof}
By Cauchy-Schwarz inequality, \[\left|\expect{O_{g,\frule} N(g,T)} - \expect{O_{g,\frule}}\expect{N(g,T)}\right| \leq \sqrt{\var(O_{g,\frule})\var(N(g,T))} \leq \sqrt{\var(N(g,T))} \leq \sqrt{p_g T}.\] 
Dividing both sides by $p_g T$ gives the desired result.
\end{proof}

The following lemma connects the concentration of $Q_{g,\frule}(\bomega)$ to the concentration of $O_{g,\frule}(\bomega)$.
\begin{lemma}\label{lem:qg-to-og}
For any $\beta \in (0,e^{-1})$, $q > 0$ and a group $g \in \Sgro$, with probability at least $$1 - \beta^4 - \Pr\{Q_{g,\frule}(\bomega) \geq \expectsub{\bomega'}{Q_{g,\frule}(\bomega')} + q\},$$ we have $O_{g,\frule}(\bomega) \leq \expectsub{\bomega'}{O_{g,\frule}(\bomega')}+6\sqrt{\frac{2\ln\nicefrac{1}{\beta}}{p_g T}}+\frac{2q}{p_g T}$.
\end{lemma}
\begin{proof}
Fix a group $g$. Suppose that $p_g T \geq 4\sqrt{2p_g T\ln\frac{1}{\beta}}$ and equivalently $p_g T \geq 32\ln\frac{1}{\beta}$. If this is not the case, $6\sqrt{\frac{2\ln\nicefrac{1}{\beta}}{p_g T}} \geq 1$ and the result trivially holds. Define event $\set{S} = \{\bomega \colon N(g,T,\bomega) \geq p_g T - 2\sqrt{2p_g T \ln \frac{1}{\beta}}\} \cap \{Q_{g,\frule}(\bomega) \leq \expectsub{\bomega'}{Q_{g,\frule}(\bomega')} + q\}$. By Lemma~\ref{lem:chernoff} and union bound, we have $\Pr\{\set{S}\} \geq 1 - \beta^4 - \Pr\{Q_{g,\frule}(\bomega) \geq \expectsub{\bomega'}{Q_{g,\frule}(\bomega')} + q\}$. Condition on $\bomega \in \set{S}$. We first have $N(g,T,\bomega) \geq p_g T - 2\sqrt{2p_g T \ln \frac{1}{\beta}} \geq \frac{p_g T}{2} > 0$ and thus $N(g,T,\bomega) \geq 1$ since it is an integer. By the definition of $Q_{g,\frule}(\bomega)$, the total requirement, we have
\begin{align}
O_{g,\frule}(\bomega) = \frac{Q_{g,\frule}(\bomega)}{N(g,T,\bomega)} &\leq \frac{\expectsub{\bomega'}{Q_{g,\frule}(\bomega')}+q}{p_g T-2\sqrt{2p_g T\ln\frac{1}{\beta}}} \tag{since conditioning on $\set{S}$}\\
&\leq \frac{\expectsub{\bomega'}{O_{g,\frule}(\bomega')}p_g T+\sqrt{p_g T}+q}{p_g T \left(1 - \sqrt{2\frac{2\ln\nicefrac{1}{\beta}}{p_g T}}\right)} \tag{By Lemma~\ref{lem:ep-qg-og}} \\
&= \frac{1}{1 - \sqrt{\frac{8\ln\nicefrac{1}{\beta}}{p_g T}}}\left(\expectsub{\bomega'}{O_{g,\frule}(\bomega')} + \sqrt{\frac{1}{p_g T}} + \frac{q}{p_g T}\right) \nonumber \\
&\hspace{-1in}\leq \left(1+4\sqrt{\frac{2\ln\nicefrac{1}{\beta}}{p_g T}}\right)\left(\expectsub{\bomega'}{O_{g,\frule}(\bomega')} + \sqrt{\frac{1}{p_g T}} + \frac{q}{p_g T}\right) \tag{$\frac{1}{1-x} \leq 1+2x$ for $0<x<1/2$}  \\
&\hspace{-1in}\leq \expectsub{\bomega'}{O_{g,\frule}(\bomega')} + 4\sqrt{\frac{2\ln\nicefrac{1}{\beta}}{p_g T}} + 2\sqrt{\frac{1}{p_g T}}+\frac{2q}{p_g T} \tag{$p_g T \geq 32\ln\frac{1}{\beta}$ and $\expectsub{\bomega'}{O_{g,\frule}(\bomega')} \leq 1$} \\
&\leq \expectsub{\bomega'}{O_{g,\frule}(\bomega')} + 6\sqrt{\frac{2\ln\nicefrac{1}{\beta}}{p_g T}}+\frac{2q}{p_g T} \nonumber \tag{$\ln 1/\beta \geq 1$}
\end{align}
which finishes the proof.
\end{proof}
We next prove Lemma~\ref{lem:bpc-og-conc} by combing Lemmas~\ref{lem:ext-mcd} and \ref{lem:qg-to-og}.
\begin{proof}[Proof of Lemma~\ref{lem:bpc-og-conc}]
Fix a group $g$ and let $\beta = (\delta/(e^4))^{1/4}$. We assume $\delta \leq 1$ since the result otherwise trivially holds. In this case we also have $\beta \in (0,e^{-1}].$ Since the fairness rule $\set{F}$ is $(\chi,\delta)-$sensitive, by Definition~\ref{def:irr-rule} there exists an event $\set{B}$ with $\Pr\{\set{B}\} \geq 1 -\delta/T$ such that \eqref{eq:sensitivity} holds for $Q_{g,\frule}$ over $\set{B}$. Applying Lemma~\ref{lem:ext-mcd} then gives
\begin{equation}\label{eq:bound-q-pr}
\Pr\left\{Q_{g,\frule}(\bomega) \geq \expectsub{\bomega'}{Q_{g,\frule}(\bomega') \mid \bomega' \in \set{B}} + 5\chi\sqrt{\max\left(2p_g(\set{P})T,\ln\frac{1}{\beta}\right)}\ln\frac{1}{\beta} + \delta \chi\right\} \leq \delta+\beta^4.
\end{equation}

\begin{align*}
\text{
Note that }\expectsub{\bomega'}{Q_{g,\frule}(\bomega')} &= \expectsub{\bomega'}{Q_{g,\frule}(\bomega') \mid \bomega' \in \set{B}}\Pr\{\set{B}\} + \expectsub{\bomega'}{Q_{g,\frule}(\bomega') \mid \bomega' \in \set{B}^c}\Pr\{\set{B}^c\} \\
&\hspace{-1in}= \expectsub{\bomega'}{Q_{g,\frule}(\bomega') \mid \bomega' \in \set{B}} + (\expectsub{\bomega'}{Q_{g,\frule}(\bomega') \mid \bomega' \in \set{B}^c} - \expectsub{\bomega'}{Q_{g,\frule}(\bomega') \mid \bomega' \in \set{B}})\Pr\{\set{B}^c\} \\
&\hspace{-1in}\geq \expectsub{\bomega'}{Q_{g,\frule}(\bomega') \mid \bomega' \in \set{B}} - T\Pr\{\set{B}^c\} 
\end{align*}
where the last inequality is because $Q_{g,\frule}$ is always bounded by $T$. Therefore, $\expectsub{\bomega'}{Q_{g,\frule}(\bomega') \mid \bomega' \in \set{B}} \leq \expectsub{\bomega'}{Q_{g,\frule}(\bomega')}+T\Pr\{\set{B}^c\}$. Combining it with \eqref{eq:bound-q-pr} gives
\[
\Pr\left\{Q_{g,\frule}(\bomega)\geq \expectsub{\bomega'}{Q_{g,\frule}(\bomega')} + 5\chi\sqrt{\max\left(2p_g(\set{P})T,\ln\frac{1}{\beta}\right)}\ln\frac{1}{\beta} + \delta(\chi+1)\right\} \leq \delta+\beta^4.
\]
Note that if $p_gT \leq \ln(1/\beta),$ the desired result trivially holds since $O_g(\bomega) \leq 1.$ We thus assume $p_g T > \ln(1/\beta).$ Applying Lemma~\ref{lem:qg-to-og} with $q = 5\chi\sqrt{\max\left(2p_g(\set{P})T,\ln\frac{1}{\beta}\right)}\ln\frac{1}{\beta} + \delta(\chi+1)$ gives that with probability at least $1 - 2\beta^4 - \delta \geq 1 - 2\delta$, we have 
\begin{align*}
O_{g,\frule}(\bomega) &\leq \expectsub{\bomega'}{O_{g,\frule}(\bomega')} + 6\sqrt{\frac{2\ln\nicefrac{1}{\beta}}{p_g T}}+\frac{q}{p_g T} \\
&\leq \expectsub{\bomega'}{O_{g,\frule}(\bomega')} +6\sqrt{\frac{2\ln\nicefrac{1}{\beta}}{p_g T}}+\frac{10\chi\sqrt{\max\left(p_gT,\ln\frac{1}{\beta}\right)}\ln\frac{1}{\beta}}{p_g T} + \frac{\delta(\chi + 1)}{p_g T} \\
&\leq \expectsub{\bomega'}{O_{g,\frule}(\bomega')} +16\chi\sqrt{\frac{2\ln\nicefrac{1}{\beta}}{p_g T}}+ 2\chi \delta,
\end{align*}
where the last inequality is because we focus on $p_g T \geq \ln(1/\beta)$ and $\chi \geq 1$. We complete the proof by noting $2\ln(1/\beta) \leq \ln(e^2/\delta)$ by the definition of $\beta$.
\end{proof}

\section{Distribution-Independent Vanishing Regret (Section~\ref{sec:cons})}\label{app:cons}
% !TEX root = main.tex
\subsection{Guarantee of ex-post $g-$regret for \textsc{CBP} (Lemma~\ref{lem:cons-group})}\label{app:proof-cons-group}
In this section, we assume $\delta \leq 1$ since otherwise the lemma vacuously holds. Define 
\begin{equation}\label{eq:set-of-T1}
T_1 = \frac{2^{52}\chi^{7.2}(M+G)^{2.6}}{\hat{s}^3_{\min}\varepsilon^{7.2}\delta^{2.6}}.
\end{equation}
By assumption, $\beta = \left(\frac{\delta}{12(M+G)T}\right)^{1/4}$ and we set in \eqref{eq:parameter-setting} that $\gamma_g = \min\left(1-\expect{O_g},16\chi\sqrt{\frac{2\ln(1/\beta)}{p_g T}}\right)$. We next provide a lemma concerning the comparison between a logarithm and a polynomial.
\begin{lemma}\label{lem:log-comp}
For $x \geq 10^9$, we have $\ln(x) \leq x^{1/6}$.
\end{lemma}
\begin{proof}
Let $q(x) = x^{1/6} - \ln(x)$. Then $q'(x) = \frac{1}{6}x^{-5/6}-\frac{1}{x}$ and $q'(x) > 0$ for any $x > 6^6 = 46656$. Therefore $q(x)$ increases in $(46656,+\infty)$. Since $q(10^9) > 0$, we know $x^{1/6} \geq \ln(x)$ for any $x \geq 10^9$.
\end{proof}

The next lemma concerns the scaling of different parameters.
\begin{lemma}\label{lem:gind-var-ineq}
If $T \geq T_1$, then $T \geq 36\dcbp$ and $\frac{2^{21}\chi^3 \ln^2(1/\beta)\sqrt{G}}{\varepsilon^3 \sqrt{T}} \leq 6T\beta^4$.
\end{lemma}
\begin{proof}
We first show the first bound. Fix $T \geq T_1$. By Lemma~\ref{lem:bound-cgamma}, we have $\dcbp \leq \frac{404\chi \ln(T/\delta)\sqrt{GT}}{\hat{s}_{\min}\varepsilon}$. It thus suffices to show $\sqrt{T} \geq \frac{14544\chi \ln(T/\delta)\sqrt{G}}{\hat{s}_{\min}\varepsilon}$. Note that $T \geq 1/\delta$ and $\ln(T) \leq T^{1/6}$ by Lemma~\ref{lem:log-comp}. Thus, it suffices to have $T^{1/3} \geq \frac{30000\chi \sqrt{G}}{\hat{s}_{\min}\varepsilon}$, which holds by the definition of $T_1$.

For the second bound, recall that $\beta = \left(\frac{\delta}{12(M+G)T}\right)^{1/4}$, and thus $1/\beta > 10^9$ since $T \geq T_1$. By Lemma~\ref{lem:log-comp}, we have $\ln^2(1/\beta) \leq (1/\beta)^{1/3}$. It then suffices to show $\left(\frac{2^7\chi}{\varepsilon}\right)^3 \leq T^{3/2}\beta^{13/3}$. Taking both sides to the power of $12/13$ and recalling the definition of $\beta$, the inequality is implied by showing $12\left(\frac{2^7 \chi}{\varepsilon}\right)^{36/13}(M+G) / \delta \leq T^{5/13}$, which holds true when $T \geq T_1$.
\end{proof}
\begin{proof}[Proof of Lemma~\ref{lem:cons-group}]
Fix a group $g \in \Sgro$. Since $T \geq T_1 \geq 3$, we have $\beta < e^{-1}$. Consider two cases:
\begin{itemize}
\item $\gamma_g \leq \varepsilon_g / 2$ and $p_g \geq \frac{144\ln(1/\beta)}{\varepsilon_g^2T}$. In this case, by Lemma~\ref{lem:fair-case-1} and the bound on $\dcbp$ in Lemma~\ref{lem:bound-cgamma}, the average score of group $g$ is at least  $\alpha_g\geq\expect{O_g} + \gamma_g - \frac{4848\chi \ln(T/\delta)}{\hat{s}_{\min}\varepsilon}\sqrt{\frac{G}{T}}$ with probability at least $1-(5T+5)(M+G)\beta^4$;
\item $\gamma_g > \varepsilon_g / 2$ or $p_g < \frac{144\ln(1/\beta)}{\varepsilon_g^2T}$. Recall that $\gamma_g = \min(1-\expect{O_g},16\chi\sqrt{\frac{2\ln(1/\beta)}{p_g T}})$. If $\gamma_g > \varepsilon_g / 2$, we must have $p_g \leq \frac{2^{11}\chi^2 \ln(1/\beta)}{\varepsilon^2 T}$. Since $\frac{144\ln(1/\beta}{\varepsilon_g^2 T} \leq \frac{2^{11}\chi^2 \ln(1/\beta)}{\varepsilon^2 T}$, we must have $p_g \leq \frac{2^{11}\chi^2 \ln(1/\beta)}{\varepsilon^2 T}$ for this case. Let $C = \frac{2^{21}\chi^3 \ln^2(1/\beta)}{\hat{s}_{\min}\varepsilon^3}\sqrt{\frac{G}{T}}$. By Lemma~\ref{lem:fair-case-2}, which applies since $T \geq 36\dcbp$ by Lemma~\ref{lem:gind-var-ineq}, with probability at least $1-(5T+5)(M+G)\beta^4-C$, $\alpha_g \geq \min(\expect{O_g}+\gamma_g,O_g)-\frac{4848}\chi \ln(T/\delta){\hat{s}_{\min}\varepsilon}\sqrt{\frac{G}{T}}$. Moreover, Lemma~\ref{lem:gind-var-ineq} shows that $C \leq 6T\beta^4$. 
\end{itemize}
Summarizing the above two cases, we then have with probability at least $1-(11T+5)(M+G)\beta^4 \geq 1 - 12T(M+G)\beta^4 \geq 1 - \delta$ (since $T \geq T_1$), $\alpha_g \geq \min(\expect{O_g}+\gamma_g,O_g)-\frac{4848\chi \ln(T/\delta)\sqrt{G/T}}{\hat{s}_{\min}\varepsilon}$. Since the fairness rule is $(\chi,\delta)-$sensitive, Lemma~\ref{lem:bpc-og-conc} shows that with probability $1-2\delta$, $O_g \leq \expect{O_g}+16\chi\sqrt{\frac{2\ln(e^2/\delta)}{p_g T}} + 2\chi \delta$. Since $O_g$ is always upper bounded by $1$, we then have with probability $1-2\delta$, $O_g \leq \expect{O_g}+\gamma_g + 2\chi \delta$ by the definition of $\gamma_g$ in \eqref{eq:parameter-setting}. Therefore, for a fixed group $g \in \Sgro$, by union bound, with probability at least $1-\delta - 2\delta = 1- 3\delta$, we have
\begin{align*}
\set{R}_{g,\frule}^{\mathrm{ex},\alg} = O_g - \alpha_g &\leq O_g -  \left(\min(\expect{O_g}+\gamma_g,O_g)-\frac{4848\chi \ln(T/\delta)\sqrt{G/T}}{\hat{s}_{\min}\varepsilon}\right)\\
&\leq 2\chi \delta +\frac{4848\chi \ln(T/\delta)\sqrt{G/T}}{\hat{s}_{\min}\varepsilon},
\end{align*}
which finishes the proof.
\end{proof}

\subsection{Lower bound on $T$ for high-probability g-regret (Assumption of Lemma~\ref{lem:cons-group})}
\label{app:small_groups_failure_example}
Recall that in Lemma~\ref{lem:cons-group}, we require $T \geq T_1$ where $T_1$ needs to scale with $1/\delta$ (see \eqref{eq:set-of-T1}). We next show that it is necessary by providing a lower bound of $T$ to obtain high probability low ex-post $g$-regret result.
\begin{proposition}\label{prop:lowerbound-T}
For any $\delta > 0$ and even $T \leq \frac{3}{40\delta}$, there is an instance such that for any non-anticipatory algorithm, with probability at least $\delta$, the ex-post $g$-regret of a group is at least $0.5$.
\end{proposition}
\begin{proof}
Fix $\delta$ and an even $T \leq \frac{3}{40\delta}$. Consider a setting with $M = 2, G=2, s_1 = \frac{T}{2}, p_1 = \frac{1}{T}, p_2 = 1-\frac{1}{T}$ and we use the proportionally optimized fairness rule. For a case $t$ from group $1$, its score is given as follows: with probability $\frac{1}{2}$, $w_{t,1} = 1, w_{t,0} = 0$ and otherwise $w_{t,1} = 0, w_{t,2} = 1$. Then with probability $T(1-\nicefrac{1}{T})^{T-1}\frac{1}{T} \geq 0.3$, there is exactly one case of group $1$ and we have $O_1 = 0.5.$ Now define event $\set{S}$ by the event that there is only one case of group $1$; this case is the last arrival; and the case has zero score for the only location $J_1$ with remaining capacity. Then $\Pr\{\set{S}\} = \sum_{j=1}^2 \Pr\{N(1,T-1) = 0,J_1 = j\}\Pr\{g(T)=1,w_{T,j}=0\} = \frac{0.5}{T}\Pr\{N(1,T-1) = 0\} \geq \frac{0.15}{T}$. In this scenario, $\alpha_1 = 0$ for any non-anticipatory algorithm. Therefore, with probability at least $\frac{0.15}{T} \geq \delta,$ it must incur $\alpha_1 = 0 = O_1 - 0.5.$
\end{proof}

\subsection{Distribution-independent vanishing regret of \textsc{CBP} (Theorem~\ref{thm:cbp-dis-dep})}\label{app:cor-cbp-dis-dep}
\begin{proof}
We verify \eqref{eq:dis-ind-reg} for \textsc{CBP} under the given fairness rule sequence $\{\set{F}(T)\}.$ Fix any $\xi > 0$. Recall that $\set{R}_{\frule}^{\max}(T,\mathcal{P},\bomega)$ is the maximum regret defined in \eqref{def:max-regret}. Define $u_T = \frac{4848\chi\ln(T/\delta(T))}{\hat{s}_{\min}\varepsilon}\sqrt{\frac{G}{T}} + 2\chi \delta(T)$, which is the maximum of the regret bounds in Lemmas~\ref{lem:cons-global} and \ref{lem:cons-group}. By Lemmas~\ref{lem:cons-global} and \ref{lem:cons-group}, as long as $T \geq T_1(\delta(T))$, where $T_1(\delta(T))$ is set by replacing $\delta$ in \eqref{eq:set-of-T1} with $\delta(T)$, we have $\Pr\{\set{R}_{\frule(T)}(T,\set{P},\bomega) > u_T\} \leq 1 - \delta(T)$ and $\Pr\left\{\max_{g \in \Sgro} \set{R}_{g,\frule(T)}^{\mathrm{ex}}(T,\set{P},\bomega) > u_T\right\} \leq 1 - 3G\delta(T).$ As a result, if $T \geq T_1(\delta(T)),$ we have $\Pr\{\set{R}_{\frule}^{\max}(T,\mathcal{P},\bomega) > u_T\} \leq (3G+1)\delta(T) \leq 4G\delta(T).$ This guarantee holds for any feature distribution $\set{P}$ within the distribution class $\set{C}$ by the assumptions of the corollary. Therefore, if $T \geq T_1(\delta(T)),$ 
\begin{equation}\label{eq:supineq}
\sup_{\set{P} \in \set{C}} \Pr_{\bomega \in \set{P}^T} \left\{\set{R}_{\frule}^{\max}(T,\mathcal{P},\bomega) > u_T\right\} \leq 4G\delta(T).
\end{equation}
Since $\lim_{T \to \infty} T\delta(T)^{2.6} \to \infty$ implies $\lim_{T \to \infty} \ln(T) - 2.6\ln(1/\delta(T)) = +\infty$, we have \[\limsup_{T \to \infty} \ln(T / \delta(T)) / (2\ln(T)) \leq 1.\]
Therefore, 
\[
\limsup_{T \to \infty} u_T \leq \limsup_{T \to \infty} \frac{9696\chi \ln(T)}{\hat{s}_{\min}\varepsilon}\sqrt{\frac{G}{T}}+2\chi\delta(T) = 0,
\]
where the last equality is by the assumption that $\lim_{T \to \infty} \delta(T) = 0$. Together with $u_T \geq 0$, this upper bound shows that $\lim_{T \to \infty} u_T = 0$. As a result, there exists some $T_2$ such that $u_T \leq \xi$ for any $T \geq T_2$. That is, if $T \geq T_1(\delta(T))$ and $T \geq T_2$, then by \eqref{eq:supineq}, 
\begin{equation}\label{eq:supp-f}
\sup_{\set{P} \in \set{C}} \Pr_{\bomega \in \set{P}^T} \left\{\set{R}_{\frule}^{\max}(T,\mathcal{P},\bomega) > \xi\right\} \leq 4G\delta(T).
\end{equation}
Now since $\lim_{T \to \infty} T \cdot \delta(T)^{2.6} = \infty,$ recalling the form of $T_1(\delta(T))$ from \eqref{eq:set-of-T1} gives 
\[\lim_{T \to \infty} T \cdot \frac{2^{52}\chi^{7.2}(M+G)^{2.6}}{\hat{s}_{\min}^3 \varepsilon^{7.2} T_1(\delta(T))} = \infty,\]
which implies $\lim_{T \to \infty} \frac{T}{T_1(\delta(T))} = \infty.$ As a result, there exists some $T_3$ such that $T \geq T_1(\delta(T))$ for any $T \geq T_3$. Thus, for any $T \geq T_4 \triangleq \max(T_2, T_3)$, we have \eqref{eq:supp-f} always holds true. We then verify \eqref{eq:dis-ind-reg} for \textsc{CBP} by noting that for any fixed $\xi > 0$, because of the existence of such $T_4$, we must have 
\[
\limsup_{T \to \infty} \sup_{\set{P} \in \set{C}} \Pr_{\bomega \in \set{P}^T} \left\{\set{R}_{\frule}^{\max}(T,\mathcal{P},\bomega) > \xi\right\} \leq \limsup_{T \to \infty} \delta(T) = \lim_{T \to \infty} \delta(T) = 0.
\]
Since $\Pr_{\bomega \in \set{P}^T} \left\{\set{R}_{\frule}^{\max}(T,\mathcal{P},\bomega) > \xi\right\} \geq 0$, we also have $\liminf_{T \to \infty} \sup_{\set{P} \in \set{C}} \Pr_{\bomega \in \set{P}^T} \left\{\set{R}_{\frule}^{\max}(T,\mathcal{P},\bomega) > \xi\right\} \geq 0$. As a result, the limit exists and $\lim_{T \to \infty} \sup_{\set{P} \in \set{C}} \Pr_{\bomega \in \set{P}^T} \left\{\set{R}_{\frule}^{\max}(T,\mathcal{P},\bomega) > \xi\right\} = 0$, which finishes the proof.
\end{proof}
\subsection{Distribution-independent vanishing regret without positive slackness (Theorem~\ref{thm:cbp-dis-dep})}\label{app:zero-slackness}
We next discuss how to adapt our results to a fairness rule without positive slackness. Suppose we have an ex-post feasible and $(\chi, \delta(T))-$sensitive fairness rule $\set{F}$, which may not have positive slackness.  For simplicity assume that $\delta(T) = T^{-\alpha}$ with $\alpha > 0$ and $2.6\alpha < 1$. We can run \textsc{CBP} with a fictional fairness rule $\tilde{\set{F}}$ that sets the minimum requirement of each group to be $O_{g, \tilde{\set{F}}}(\bomega) = O_{g,\set{F}}(\bomega) - T^{-\beta}$ where $\beta$ satisfies $\beta > 0$ and $7.2\beta < 1 - 2.6\alpha$. By its construction, the fictional fairness rule $\tilde{\set{F}}$ is ex-post feasible, $(\chi,\delta(T))$-sensitive, and has slackness $\varepsilon = T^{-\beta}$. As a result, Lemmas~\ref{lem:cons-global} and \ref{lem:cons-group} apply, showing that with probability at least $1 - (3G+1)\delta(T)$,
\begin{equation}\label{eq:guarantee-fictious}
\set{R}_{\tilde{\set{F}}}^{\alg} \leq \frac{404\chi\ln(T/\delta(T)}{\hat{s}_{\min}\sqrt{G}}T^{-0.5+\beta}, ~~\text{and}~~\set{R}_{g,\tilde{\set{F}}}^{\mathrm{ex},\alg} \leq \frac{4848\chi\ln(T/\delta(T))}{\hat{s}_{\min}\sqrt{G}}T^{-0.5+\beta}
\end{equation}
for any group $g \in \Sgro$ as long as $T \geq T_1$ where $T_1$ is defined in \eqref{eq:set-of-T1}. 
Since $\delta(T) = T^{-\alpha}$, $\varepsilon = T^{-\beta}$ and $7.2\beta + 2.6\alpha < 1$, we have $T \geq T_1$ for sufficiently large $T$.

Moreover, since $O_{g,\tilde{\frule}}(\bomega) < O_{g,\frule}(\bomega),$ the global objective under $\tilde{\frule}$ is no less than that under $\frule$, i.e., $\set{R}_{\set{F}}^{\alg} \leq \set{R}_{\tilde{\set{F}}}^{\alg}$. \eqref{eq:guarantee-fictious} then also implies the following guarantee of $\alg$ for the original fairness rule $\frule$:
\begin{align*}
\set{R}_{\set{F}}^{\alg} \leq \set{R}_{\tilde{\set{F}}}^{\alg} \leq \frac{404\chi\ln(T/\delta(T))}{\hat{s}_{\min}\sqrt{G}}T^{-0.5+\beta} ~~\text{and}~~\\
\set{R}_{g,\set{F}}^{\mathrm{ex},\alg} = \set{R}_{g,\tilde{\set{F}}}^{\mathrm{ex},\alg} + T^{-\beta} \leq \frac{4848\chi\ln(T/\delta(T))}{\hat{s}_{\min}\sqrt{G}}T^{-0.5+\beta} + T^{-\beta}.
\end{align*}
Since $\beta \in (0,0.5)$ and $\delta(T) = T^{-\alpha}$ with $2.6\alpha < 1$, the right hand side converges to zero as $T$ goes to infinity. Moreover, the probability of \eqref{eq:guarantee-fictious} being true, $1 - (3G+1)\delta(T)$, converges to $1$ as $T$ goes to infinity.  We can then verify \eqref{eq:dis-ind-reg} as in Appendix~\ref{app:cor-cbp-dis-dep}. As a result, a simple adaptation of $\alg$ obtains distribution-independent vanishing regret for a fairness rule as long as it is ex-post feasible and has low sensitivity, even if its slackness is not bounded away from zero. 

\subsection{Lower bound $\Temp$ under \textsc{Conservative Bid Price Control} (Lemma~\ref{lem:cons-temp})}\label{app:proof-cons-temp}
Let $a_j(t)$ be the amount of consumed capacity at location $j$ after the assignment of case $t$. Recall the definition of $\Temp$ as the index of the first case after whose assignment a location runs out of capacity. That is, $\Temp = \min\{t\colon \exists j \in \Sloc, a_j(t) = s_j\}$. The intuition of the proof is to separate $a_j(t)$ into two parts: capacity used by assignment $J^{\BP}$ and that used by greedy assignment. The first part is bounded using ideas similar to the proof of Lemma~\ref{lem:bpc-temp} for \textsc{Amplified Bid Price Control}. The second part is bounded by Lemma~\ref{lem:bound-step-greedy}.

To formalize the argument, let us define  $\{b_j(t), t \in [T],j\in\Sloc\}$ which concerns assignment $J^{\BP}$. For an arrival $t$, let $b_j(t) = 1$ if $J^{\BP}(t) =j$; otherwise $b_j(t) = 0$. To avoid the impact of capacity constraints, we consider a new system which has the same arrival process as the original one but capacity constraints can be violated at no cost. Formally, define $\hat{J}^{\CONS}(t)$ as the assignment for case $t$ in the new system. Note that for the first $\Temp$ cases, we have $J^{\alg}(t) = J^{\CONS}(t) = \hat{J}^{\CONS}(t)$ since the original system still has capacity in every location. Based on $\hat{J}^{\CONS}(t)$, we also similarly define $\hat{V}_g[t]$ by the total fairness excess score in the new system, given by $\sum_{\tau \in \Sarr(g,t)} (w_{\tau,\hat{J}^{\CONS}(\tau)} - \expect{O_g}-\gamma_g)$. Note that condition \eqref{eq:cond-predict} for $\hat{J}^{\CONS}(t)$ is evaluated on $\hat{V}_g[t-1]$. Now let us define $c_g(t) \in \{0,1\}$ such that $c_g(t) = 1$ if and only if case $t$ is of group $g$ and in the new system Algorithm~\ref{algo:conservative} takes a greedy step, i.e.,  condition~\eqref{eq:cond-predict} is violated. Define $\hat{a}_j(t) = \sum_{\tau=1}^t \left(b_j(\tau) + c_{g(\tau)}(\tau)\right)$, which is an upper bound on the capacity usage in the new system. Also define $\widehat{\Temp} = \min\{t\colon \exists j \in \Sloc, \hat{a}_j(t) \geq s_j\}$. The next lemma shows that $\widehat{\Temp}$ is a lower bound of $\Temp$.
\begin{lemma}\label{lem:low-temp}
For every sample path, we have $\Temp \geq \widehat{\Temp}$.
\end{lemma}
\begin{proof}
Fix a case $t'$ such that $a_j(t' - 1) < s_j$ for every $j \in \Sloc$. That is, before the assignment of case $t'$, every location has available capacity. As a result, for arrivals $1,\ldots,t'$, \textsc{Conservative Bid Price Control} assigns cases following $J^{\CONS}(t')$, and thus $J^{\alg}(t') =J^{\CONS}(t')= \hat{J}^{\CONS}(t')$. We then upper bound $a_j(t')$ for a location $j \in \Sloc$ by 
\begin{equation}\label{eq:upper-usage}
a_j(t') = \sum_{\tau=1}^{t'} \indic{J^{\alg}(\tau) = j} = \sum_{\tau=1}^{t'} \indic{\hat{J}^{\CONS}(\tau) = j} \leq \sum_{\tau=1}^{t'} (b_j(\tau)+c_{g(\tau)}(\tau))=\hat{a}_j(t'),
\end{equation}
where the inequality is due to the fact that as long as $\hat{J}^{\CONS}(\tau) = j$, either it is chosen because of an assignment $J^{\BP}$ or a greedy assignment, i.e., $b_j(\tau) = 1$ or $c_{g(\tau)}(\tau) = 1$. Notice that for every $j \in \Sloc$, we have $a_j(\Temp - 1) < s_j$ by the definition of $\Temp$. As a result, taking $t' = \Temp$ in \eqref{eq:upper-usage} gives $a_j(\Temp) \leq \hat{a}_j(\Temp)$ for every location $j$ and thus $\hat{a}_{j'}(\Temp) \geq s_{j'}$ for some location $j'$. By the non-decreasing nature of $\hat{a}_{j'}(t)$ in $t$, we conclude that $\Temp \geq \widehat{\Temp}$. 
\end{proof}
As a result of Lemma~\ref{lem:low-temp}, it suffices to show that with high probability $\widehat{\Temp} \geq T-\dcbp$ to prove Lemma~\ref{lem:cons-temp}. Recall that $\widehat{\Temp}$ is the first period $t$ that $\hat{a}_j(t)$ exceeds $f_j(1)$ for some $j \in \Sloc$. It thus requires us to upper bound $\hat{a}_j(t) = \sum_{\tau=1}^t b_j(\tau)+\sum_{\tau=1}^t c_{g(\tau)}(\tau)$. The first term on the right hand side corresponds to the number of assignments $J^{\BP}$. The second term is the number of greedy assignments. We upper bound these two terms separately in Lemma~\ref{lem:bound-step-bp} and Lemma~\ref{lem:bound-step-greedy}.
\begin{lemma}\label{lem:bound-step-bp}
With probability at least $1 - MT\beta^4$, we have $\sum_{\tau=1}^t b_j(\tau) \leq \hat{s}_j t + \sqrt{2t\ln(\nicefrac{1}{\beta})}$ for every $t \in [T], j \in \Sloc$.
\end{lemma}
\begin{proof}
The proof is similar to that of Lemma~\ref{lem:bpc-temp}. We include it here for completeness. Fix $t \in [T], j \in \Sloc$. Note that $b_j(1),\ldots,b_j(t)$ are independent Bernoulli random variables. By Lemma~\ref{lem:bp-property-cap}, we have $\sum_{\tau\leq t}\Pr\{b_j(\tau) = 1\} = \sum_{\tau\leq t}\Pr\{J^{\BP}(\tau) = j\} \leq \hat{s}_j t$. As a result, using Hoeffding's Inequality from Fact~\ref{fact:hoeffding} gives $\Pr\left\{\sum_{\tau=1}^t b_j(\tau) \geq \hat{s}_j t + \sqrt{2t\ln(\nicefrac{1}{\beta})}\right\} \leq \beta^4.$ Using a union bound over all $t \in [T]$ and $j \in \Sloc$ gives the desired result.
\end{proof}
Now let us upper bound the number of greedy steps.  Instead of considering all groups, we count on a group level and show that for each group there are a limited number of greedy assignments during the entire horizon. {Suppose the labels of group $g$ cases are $t^g_1,\ldots,t^g_{N(g,T)}$.} Condition on $N(g,T) = n$ where $1 \leq n \leq T$. Define $\mu_g^\GR = \expect{w_{1,J^{\GR}(1)} \mid g(1) = g}, \mu_g^{\BP} = \expect{w_{1,J^{\BP}(1)} \mid g(1) = g}$ as the expected score of a group $g$ arrival for a greedy and an assignment $J^{\BP}$ respectively (note that these expectations are the same across all cases condition on belong to a fixed group by the i.i.d. assumption of refugees' features). In addition, let $n_g^{\GR}(t) = \sum_{\tau=1}^{t} c_g(\tau)$ {be the number of greedy assignments for group $g$ in the first $t$ periods.} The following lemma shows a lower bound on $\hat{V}_g[t^g_k]$ for any $1 \leq k \leq n$. 
\begin{lemma}\label{lem:score-concentration}
For a fixed group $g$ with $\gamma_g \leq \frac{\varepsilon_g}{2}$, there exists an event $\set{S}_{\mathrm{score}}^g$ with probability at least $1 - 2T\beta^4$ such that  condition on $N(g,T) = n$ with $n \in [1,T]$ and $\set{S}_{\mathrm{score}}^g$, for any $1 \leq k \leq n$, we have 
\[
\hat{V}_g[t^g_k] \geq \frac{n^{\GR}_g(t^g_k)\varepsilon_g}{2} - k\left(\sqrt{\frac{1}{p_g T}}+\gamma_g\right)-2\sqrt{k\ln(\nicefrac{1}{\beta})}.
\]
\end{lemma}
\begin{proof}
Consider the sequence of scores for arrivals receiving greedy assignments $W^{\GR}$. Define $\bar{\mu}_{g,i}^{\GR}$ as the average of the first $i$ values in this sequence. Note that the assignment rule $\hat{J}^{\CONS}$ is independent of the score between case $t$ and any location. As a result, elements in $W^{\GR}$ are i.i.d. and have the same distribution as $\max_j w_{1,j}$ condition on $g(1) = g$. Recall that $\mu_g^{\GR} = \expect{\max_j w_{1,j} \mid g(1) = g}$. Then by Hoeffding's Inequality (Fact~\ref{fact:hoeffding}), for any fixed $i \leq T$, we have $\Pr\left\{\bar{\mu}_{g,i}^{\GR} < \mu_g^{\GR}-\sqrt{\frac{2\ln(\nicefrac{1}{\beta})}{i}}\right\} \leq \beta^4$. We can similarly define the sequence of scores for arrivals receiving assignment $J^{\BP}$ by $W^{\BP}$ and let $\bar{\mu}_{g,i}^{\BP}$ be the average of the first $i$ values. Also, recall that $\mu_g^{\BP} = \expect{w_{1,J^{\BP}(1)} \mid g(1) = g}$. Similar to the scenario of greedy scores, by Hoeffding's Inequality (Fact~\ref{fact:hoeffding}), for any fixed $1 \leq i \leq T$, we have $\Pr\left\{\bar{\mu}_{g,i}^{\BP} < \mu_g^{\BP}-\sqrt{\frac{2\ln(\nicefrac{1}{\beta})}{i}}\right\} \leq \beta^4$. Define
\[
\set{S}_{\mathrm{score}}^g = \left\{\forall i \leq T,~\bar{\mu}_{g,i}^{\GR} \geq  \mu_g^{\GR}-\sqrt{\frac{2\ln(\nicefrac{1}{\beta})}{i}}\right\} \cap \left\{\forall i \leq T,~\bar{\mu}_{g,i}^{\BP} \geq \mu_g^{\BP}-\sqrt{\frac{2\ln(\nicefrac{1}{\beta})}{i}}\right\}.
\]
We know $\Pr\left\{\set{S}_{\mathrm{score}}^g\right\} \geq 1 - 2T\beta^4$ by union bound. Let us now condition on $N(g,T) = n$. Fix $k$ such that $1 \leq k \leq n$ and consider the $k$th arrival of group $g$. For ease of notation, let $n^{\GR} = n_g^{\GR}(t_k^g)$ and $n^{\BP} = k - n^{\GR}$ be the number of greedy assignments and assignments $J^{\BP}$ for group $g$ cases in the first $t_k^g$ periods. To prove the desired result, note that $n^{\GR} \leq n$ and $n^{\BP} \leq n$. Therefore, under $\set{S}_{\mathrm{score}}^g$, we have
$\bar{\mu}_{n^{\GR}}^{\GR} \geq \mu_g^{\GR}-\sqrt{\frac{2\ln(\nicefrac{1}{\beta})}{n^{\GR}}}$ and 
$\bar{\mu}_{g,n^{\BP}}^{\BP} \geq \mu_g^{\BP}-\sqrt{\frac{2\ln(\nicefrac{1}{\beta})}{n^{\BP}}}$.
By definition, $\hat{V}_g[t_c^g] = \sum_{i=1}^{n^{\GR}} s_i^{\GR} + \sum_{i=1}^{n^{\BP}} s_i^{\BP} - k(\expect{O_g}+\gamma_g) = \bar{\mu}_{g,n^{\GR}}^{\GR}n^{\GR} + \bar{\mu}_{g,n^{\BP}}^{\BP}n^{\BP} - k(\expect{O_g}+\gamma_g)$. Then under $\set{S}_{\mathrm{score}}^g$ we have
\begin{align*}
\hat{V}_g[t_k^g] &= \bar{\mu}_{g,n^{\GR}}^{\GR}n^{\GR} + \bar{\mu}_{g,n^{\BP}}^{\BP}n^{\BP} - k(\expect{O_g}+\gamma_g) \\
&\geq \mu_g^{\GR}n^{\GR}-\sqrt{2n^{\GR}\ln(\nicefrac{1}{\beta})} + \mu_g^{\BP}n^{\BP}-\sqrt{2n^{\BP}\ln(\nicefrac{1}{\beta})} - k(\expect{O_g}+\gamma_g) \\
&\geq n^{\GR}\left(\mu_g^{\GR} - \expect{O_g}\gamma_g\right) + (k - n^{\GR})\left(\mu_g^{\BP} - \expect{O_g}-\gamma_g\right) - 2\sqrt{k\ln(\nicefrac{1}{\beta})} \\
&\geq \frac{n^{\GR}\varepsilon_g}{2} - k\left(\sqrt{\frac{1}{p_g T}}+\gamma_g\right) - 2\sqrt{k\ln(\nicefrac{1}{\beta})}\quad = \quad \frac{n_g^{\GR}(t_k^g)\varepsilon_g}{2} - k\left(\sqrt{\frac{1}{p_g T}}+\gamma_g\right) - 2\sqrt{k\ln(\nicefrac{1}{\beta})}
\end{align*}
where the first inequality uses the definition of $\set{S}_{\mathrm{score}}^g$; the second inequality uses $n^{\GR}+n^{\BP} = k$ and $\sqrt{a} + \sqrt{b} \leq \sqrt{2(a+b)}$ for any non-negative $a,b$; the last inequality uses the slackness property (Definition~\ref{def:slackness}) that $\mu_g^{\GR} \geq \expect{O_g} + \varepsilon_g$, the assumption that $\gamma_g \leq \varepsilon_g / 2$, and Lemma~\ref{lem:bp-property-fair} that $\mu_g^{\BP} \geq \expect{O_g} - \sqrt{\frac{1}{p_g T}}$. 
\end{proof}
We can then upper bound the total number of greedy steps. Recall that \[C(\bolds{\gamma}) = \frac{12}{\varepsilon}\sum_{g\colon \gamma_g \leq \nicefrac{\varepsilon_g}{2}}\gamma_gp_g T+6\sum_{g \colon \gamma_g > \nicefrac{\varepsilon_g}{2}}p_g T.\]
\begin{lemma}\label{lem:bound-step-greedy}
We have $\Pr\left\{\sum_{g \in \Sgro}\sum_{\tau=1}^T c_g(\tau) \leq \frac{18\sqrt{G T}\ln(\nicefrac{1}{\beta})}{\varepsilon} + \frac{51G\ln(\nicefrac{1}{\beta})}{\varepsilon^2}+C(\bolds{\gamma})\right\} \geq 1 - (2T+1)G\beta^4$.
\end{lemma}
\begin{proof}[Proof of Lemma~\ref{lem:bound-step-greedy}]
Fix a group $g$ with $\gamma_g \leq \varepsilon_g / 2$. We want to bound the number of greedy assignments. Condition on the number of group $g$ arrivals by $n$, i.e., condition on $N(g,T) = n \in [T]$. Take $k \in [0,n-1]$ such that the $(k + 1)$th arrival is the last case that receives a greedy assignment in group $g$. If there is no such arrival, we know $\sum_{\tau=1}^T c_g(t) = 0$. By definition of $k$, we have $\sum_{\tau=1}^T c_g(t) = n_g^{\GR}(t_k^g) + 1$. Therefore, it suffices to upper bound $n_g^{\GR}(t_k^g)$. Recall $\set{S}_{\mathrm{score}}^g$ the event in Lemma~\ref{lem:score-concentration}. We know $\Pr\{\set{S}_{\mathrm{score}}^g\} \geq 1 - 2T\beta^4$. In addition, under $\set{S}_{\mathrm{score}}^g$, we have
\begin{equation}\label{eq:lower-bound-vg}
\hat{V}_g[t_k^g] \geq \frac{n_g^{\GR}(t_k^g)\varepsilon_g}{2} - k\left(\sqrt{\frac{1}{p_g T}}+\gamma_g\right) - 2\sqrt{k\ln(\nicefrac{1}{\beta})} \geq \frac{n_g^{\GR}(t_k^g)\varepsilon_g}{2} - k\left(\sqrt{\frac{1}{p_g T}}+\gamma_g\right) - 2\sqrt{n\ln(\nicefrac{1}{\beta})}.
\end{equation}
Note that since arrival $k + 1$ receives a greedy assignment, by the definition of $J^{\CONS}$ (Line~\ref{algoline:cons-select} of Algorithm~\ref{algo:conservative}) and the setting of parameters in \eqref{eq:parameter-setting}, we have 
$\hat{V}_g[t_k^g] \leq \Cex + \Psi(g,t,\beta) + 1 =  6\ln(\nicefrac{1}{\beta})+\frac{\ln(\nicefrac{1}{\beta})}{\varepsilon_g} + 1 \leq \frac{8\ln(\nicefrac{1}{\beta})}{\varepsilon_g}.$ Combining this with \eqref{eq:lower-bound-vg} implies
$\frac{n_g^{\GR}(t_k^g)\varepsilon_g}{2} - n\left(\sqrt{\frac{1}{p_g T}}+\gamma_g\right) - 2\sqrt{n\ln(\nicefrac{1}{\beta})} \leq \frac{8\ln(\nicefrac{1}{\beta})}{\varepsilon_g}.$ As a result, condition on $N(g,T) = n$ and $\set{S}_{\mathrm{score}}^g$, we have
\begin{equation}\label{eq:bound-greedy-count}
\begin{aligned}
\sum_{\tau=1}^T c_g(\tau) = n_g^{\GR}(t_k^g) + 1 &\leq \frac{2n}{\varepsilon_g}\left(\sqrt{\frac{1}{p_g T}}+\gamma_g\right) + \frac{4\sqrt{n\ln(\nicefrac{1}{\beta})}}{\varepsilon_g} + \frac{16\ln(\nicefrac{1}{\beta})}{\varepsilon_g^2} + 1 \\
&\leq \frac{2n}{\varepsilon_g}\left(\sqrt{\frac{1}{p_g T}}+\gamma_g\right) + \frac{4\sqrt{n\ln(\nicefrac{1}{\beta})}}{\varepsilon_g} + \frac{17}\ln(\nicefrac{1}{\beta}){\varepsilon_g^2}.
\end{aligned}
\end{equation}
Define $\bar{n} = p_g T + 5\left(\sqrt{p_g T}+\ln(\nicefrac{1}{\beta})\right)$ and  $\set{S}_{\mathrm{size}}^g = \{N(g,T) \leq \bar{n}\}$. Using Lemma~\ref{lem:chernoff}, we know that $\Pr\{\set{S}_{\mathrm{size}}^g\} \geq 1 - \beta^4$. Condition on both $\set{S}_{\mathrm{size}}^g$ and $\set{S}_{\mathrm{score}}^g$. If $p_g T \leq 1$, then $\sum_{t=1}^T c_g(t) \leq \bar{n} \leq 6 + 5\ln(\nicefrac{1}{\beta})$. Otherwise, using \eqref{eq:bound-greedy-count} gives
\begin{align}
\sum_{\tau=1}^T c_g(\tau) &\leq \frac{2\bar{n}}{\sqrt{p_g T}\varepsilon_g}  + \frac{4\sqrt{\bar{n}\ln(\nicefrac{1}{\beta})}}{\varepsilon_g} + \frac{17\ln(\nicefrac{1}{\beta})}{\varepsilon_g^2} +\frac{2\bar{n}\gamma_g}{\varepsilon_g}\notag\\
&= \frac{2p_g T + 10\left(\sqrt{p_g T \ln(\nicefrac{1}{\beta})} + \ln(\nicefrac{1}{\beta})\right)}{\sqrt{p_g T}\varepsilon_g} + \frac{4\sqrt{\left(p_g T + 5\left(\sqrt{p_g T \ln(\nicefrac{1}{\beta})}+\ln(\nicefrac{1}{\beta})\right)\right)\ln(\nicefrac{1}{\beta})}}{\varepsilon_g} \notag\\
&\mspace{32mu}+17\ln(\nicefrac{1}{\beta})/\varepsilon_g^2 +\frac{2(p_g T+5\sqrt{p_g T} + \ln(1/\beta))\gamma_g}{\varepsilon_g}\notag\\
&\hspace{-0.7in}\leq \frac{2\sqrt{p_g T}}{\varepsilon_g} + \frac{20\ln(\nicefrac{1}{\beta})}{\varepsilon_g} + \frac{4\sqrt{p_g T\ln(\nicefrac{1}{\beta})}}{\varepsilon_g} + \frac{12\left(p_g T \ln^3\frac{1}{\beta}\right)^{0.25}}{\varepsilon_g} + \frac{29\ln(\nicefrac{1}{\beta})}{\varepsilon_g^2} + \frac{(12p_g T+2\ln(1/\beta))\gamma_g}{\varepsilon_g}\notag\\
&\leq \frac{18\sqrt{p_g T \ln(\nicefrac{1}{\beta})}}{\varepsilon_g} + \frac{51\ln(\nicefrac{1}{\beta})}{\varepsilon_g^2}+\frac{\gamma_gp_g T}{\varepsilon_g} \label{eq:group-greedy-bound}
\end{align}
where the second to last inequality uses the assumption that $p_g T \geq 1$, $\gamma_g \leq 1$, and the fact that $\sqrt{a + b} \leq \sqrt{a} + \sqrt{b}$ for any $a,b \geq 0$. Combining the two scenarios (whether $p_g T \leq 1$ or not), for every $g$ with $\gamma_g \leq \frac{\varepsilon_g}{2}$, we have $\sum_{\tau=1}^T c_g(\tau) \leq \frac{18\sqrt{p_g T \ln(\nicefrac{1}{\beta})}}{\varepsilon_g} + \frac{51\ln(\nicefrac{1}{\beta})}{\varepsilon_g^2}+\frac{12\gamma_gp_g T}{\varepsilon_g}$.

For a group $g$ with $\gamma_g > \frac{\varepsilon_g}{2}$, we also define the event $\set{S}_{\mathrm{size}}^g = \{N(g,T) \leq \bar{n} = p_g T + 5(\sqrt{p_g T} + \ln\frac{1}{\beta})\}$. By Lemma~\ref{lem:chernoff}, we again have with probability at least $1 - \beta^4$, $N(g,T) \leq \bar{n} \leq 6(p_g T + \ln\frac{1}{\beta})$. In addition, we have a trivial upper bound of $\sum_{\tau=1}^T c_g(\tau) \leq N(g,T)$, which is bounded by $6(p_g T + \ln\frac{1}{\beta})$ conditioned on $\set{S}_{\mathrm{size}}^g$.

Finally, let us define $\set{S} = \cap_{g \in \Sgro} \left(\set{S}_{\mathrm{size}}^g \cap \set{S}_{\mathrm{score}}^g\right)$. Then by union bound, we have $\Pr\{\set{S}\} \geq 1 - \sum_{g \in \Sgro} (1 - \Pr\{\set{S}_{\mathrm{size}}^g\} + 1 - \Pr\{\set{S}_{\mathrm{score}}^g\}) \geq 1 - G(2T+1)\beta^4$. In addition, condition on $\set{S}$, the total number of greedy assignments is upper bounded by

\begin{align*}
\sum_{g \in \Sgro} \sum_{\tau=1}^T c_g(\tau) &\leq \sum_{g\colon \gamma_g \leq \frac{\varepsilon_g}{2}} \left(\frac{18\sqrt{p_g T \ln(\nicefrac{1}{\beta})}}{\varepsilon_g} + \frac{51\ln(\nicefrac{1}{\beta})}{\varepsilon_g^2}+\frac{12\gamma_gp_g T}{\varepsilon_g}\right) + \sum_{g \colon \gamma_g > \frac{\varepsilon_g}{2}} 6\left(p_g T + \ln\frac{1}{\beta}\right) \\
&\leq \sum_{g \in \Sgro} \left(\frac{18\sqrt{p_g T \ln(\nicefrac{1}{\beta})}}{\varepsilon} + \frac{51\ln(\nicefrac{1}{\beta})}{\varepsilon}\right)+\sum_{g \colon \gamma_g \leq \frac{\varepsilon_g}{2}}\frac{12\gamma_gp_g T}{\varepsilon_g} + \sum_{g \colon \gamma_g > \frac{\varepsilon_g}{2}} 6p_g T  \\
&\leq \frac{18\sqrt{GT\ln(\nicefrac{1}{\beta})}}{\varepsilon}+\frac{51G\ln(\nicefrac{1}{\beta})}{\varepsilon^2}+C(\bolds{\gamma}).
\end{align*}

The first inequality is shown in \eqref{eq:group-greedy-bound}; the second inequality follows from $\varepsilon = \min_g \varepsilon_g \leq 1$; the third inequality is by the Cauchy-Schwarz inequality $\sum_{g \in \Sgro} \sqrt{p_g} \leq \sqrt{G\left(\sum_{g \in \Sgro} p_g\right)}=\sqrt{G}$.
\end{proof}

\begin{proof}[Proof of Lemma~\ref{lem:cons-temp}]
Recall the definition of $\widehat{\Temp}$ as the first period $t$ that $\hat{a}_j(t)$ is at least the initial capacity. That is, $\widehat{\Temp} = \min\{t\colon \exists j \in \Sloc, \hat{a}_j(t) \geq s_j\}$. Using union bound over the events in Lemma~\ref{lem:bound-step-bp} and Lemma~\ref{lem:bound-step-greedy}, there is an event $\set{S}$ with $\Pr\{\set{S}\} \geq 1 - MT\beta^4 - (2T+1)G\beta^4 \geq 1 - (2T+1)(M+G)\beta^4$, such that under $\set{S}$, we have $\sum_{g \in \Sgro} \sum_{\tau=1}^T c_g(\tau) \leq \frac{18\sqrt{GT}\ln(\nicefrac{1}{\beta})}{\varepsilon} + \frac{51G\ln(\nicefrac{1}{\beta})}{\varepsilon^2}+C(\bolds{\gamma})$ and that for every $g \in \Sgro, t \in [T]$, we have $\sum_{\tau=1}^t b_j(\tau) \leq \hat{s}_j T-(T-t)/U + \sqrt{2t\ln(\nicefrac{1}{\beta})}$. As a result, under $\set{S}$, for every $t \in [T], j\in \Sloc$, 
\begin{align*}
\hat{a}_j(t) \leq \sum_{\tau=1}^t \left(b_j(\tau)+c_{g(\tau)}(\tau)\right) &\leq \sum_{\tau=1}^t b_j(\tau) + \sum_{g \in \Sgro} \sum_{\tau=1}^T c_g(\tau) \\
&\hspace{-1in}\leq \hat{s}_jt + \sqrt{2t\ln(\nicefrac{1}{\beta})} + \frac{18\sqrt{GT}\ln(\nicefrac{1}{\beta})}{\varepsilon} + \frac{51G\ln(\nicefrac{1}{\beta})}{\varepsilon^2}+C(\bolds{\gamma}).
\end{align*}
Recall $\dcbp =\frac{1}{\hat{s}_{\min}}\left(\frac{20\sqrt{GT}\ln(\nicefrac{1}{\beta})}{\varepsilon}+\frac{51G\ln(\nicefrac{1}{\beta})}{\varepsilon^2}+C(\bolds{\gamma})\right) > 0$. We have $\forall j \in \Sloc$,
\begin{align*}
\hat{a}_j(T-\dcbp) &< \hat{s}_j T - \hat{s}_j \dcbp + \sqrt{2T\ln(\nicefrac{1}{\beta})} + \frac{18\sqrt{GT}\ln(\nicefrac{1}{\beta})}{\varepsilon} + \frac{51G\ln(\nicefrac{1}{\beta})}{\varepsilon^2}+C(\gamma) \\
&\hspace{-0.3in}\leq s_j - \frac{20\sqrt{GT}\ln(\nicefrac{1}{\beta})}{\varepsilon} - \frac{51G\ln(\nicefrac{1}{\beta})}{\varepsilon^2}-C(\bolds{\gamma}) + \frac{20\sqrt{GT}\ln(\nicefrac{1}{\beta})}{\varepsilon} + \frac{51G\ln(\nicefrac{1}{\beta})}{\varepsilon^2} + C(\bolds{\gamma})\\
&= s_j
\end{align*}
where the second inequality uses $\sqrt{2T\ln(\nicefrac{1}{\beta})} \leq \frac{2\sqrt{T}\ln(\nicefrac{1}{\beta})}{\varepsilon}$ and the definition of $\dcbp$.
Therefore, for every location $j$, we have $\hat{a}_j(T-\dcbp) < s_j - \frac{7G\ln(\nicefrac{1}{\beta})}{\varepsilon^2} \leq f_j(1)$. This implies that $\widehat{\Temp} \geq T-\dcbp$ under $\set{S}$. We then complete the proof by using Lemma~\ref{lem:low-temp} that $\Temp \geq \widehat{\Temp}$ and that $\Pr\{\set{S}\} \geq 1 - (2T+1)(M+G)\beta^4$.
\end{proof}

\subsection{Upper bound on global regret for general confidence bounds (Lemma~\ref{lem:cons-global-conf})}\label{app:proof-cons-global-conf}
\begin{proof}[Proof of Lemma~\ref{lem:cons-global-conf}]
The global average score can be bounded by the score obtained by $\Temp$, i.e.,  $\frac{1}{T}\sum_{t=1}^T w_{t,J^{\alg}(t)} \geq \frac{1}{T}\sum_{t=1}^{\Temp} w_{t,J^{\alg}(t)}.$ For the first $\Temp$ cases, with $J^{\alg}(t) = J^{\CONS}(t)$, \textsc{CBP} assigns either to $J^{\BP}(t)$ or to the greedy assignment $J^{\GR}(t)$. As a result, $$w_{t,J^{\alg}(t)} \geq \min\left(w_{t,J^{\BP}(t)},w_{t,J^{\GR}(t)}\right) = w_{t,J^{\BP}}(t).$$ Then, the global objective of \textsc{CBP} is at least $\frac{1}{T}\sum_{t=1}^{\Temp} w_{t,J^{\BP}(t)} \geq \frac{1}{T}\sum_{t=1}^T w_{t,J^{\BP}(t)} - \nicefrac{(T-\Temp)}{T}$ where the last inequality comes from $w_{t,j} \in [0,1]$. With independent cases, by Hoeffding inequality (Fact~\ref{fact:hoeffding}): \[
\Pr\left\{\sum_{t=1}^T w_{t,J^{\BP}(t)} < \sum_{t\in[T]}\expect{w_{t,J^{\BP}(t)}}-\sqrt{2T\ln(1/\beta)}\right\} \leq \beta^4.
\]
By Lemma~\ref{lem:bp-property-obj}, we have $\sum_{t\in[T]}\expect{w_{t,J^{\BP}(t)}} \geq \expect{TO^\star}$. As a result, with probability at least $1-\beta^4$, we have $\frac{1}{T}\sum_{t=1}^T w_{t,J^{\BP}(t)} \geq \expect{O^\star}-\sqrt{\frac{2\ln(1/\beta)}{T}}$. Lemma~\ref{lem:cons-temp} shows that $\Temp \geq T-\Delta^{\CONS}$ with probability at least $1-(2T+1)(M+G)\beta^4$. By union bound, the probability that the previous two events both occur is at least $1 - (2T+1)(M+G)\beta^4 - \beta^4 \geq 1-(5T+3)(M+G+1)\beta^4$. Under these two events, we can lower bound the global average score 
\[
\frac{1}{T}\sum_{t=1}^T w_{t,J^{\BP}(t)} - \frac{(T-\Temp)}{T} \geq \expect{O^\star}-\sqrt{\frac{2\ln(1/\beta)}{T}}-\frac{\Delta^{\CONS}}{T}.
\]
To get the desired result, focus on $\delta < 1$ (the result holds trivially for $\delta \geq 1$) and recall that $T_0 = \frac{12(M+G)}{\varepsilon^2}.$ As $\beta = \left(\frac{\delta}{12(M+G)T}\right)^{1/4}$ and $T \geq T_0$, we have $\beta \leq e^{-0.5}$ which implies $\sqrt{2\ln(1/\beta)} < 2\ln(1/\beta)$. As a result, with probability at least $1-(5T+5)(M+G)\beta^4$, we have $\set{R}_{\frule}^{\alg} \leq \frac{\Delta^{\CONS}}{T}+\frac{2\ln(1/\beta)}{\sqrt{T}}$. Moreover, for this $\beta$ and $T\geq T_0$, we have
$20\ln(1/\beta) \leq 6\ln\left(12(M+G)T/\delta\right) \leq 10\ln\left(T/\delta\right)$. Since $\sqrt{\frac{2\ln(1/\beta)}{T}}+\frac{\Delta^{\CONS}}{T} \leq \frac{20\ln(1/\beta)}{\hat{s}_{\min}\varepsilon}\sqrt{\frac{G}{T}}\left(1+\frac{3\sqrt{G}}{\varepsilon\sqrt{T}}\right)+\frac{ C(\bolds{\gamma})}{\hat{s}_{\min} T}$, with probability at least $1 - \delta\frac{(5T+5)(M+G)}{12(M+G)T} \geq 1 - \delta$, $\alg$ has global regret
\begin{equation}\label{eq:bound-regret-global}
\set{R}_{\frule}^{\alg} \leq \frac{20\ln(1/\beta)}{\hat{s}_{\min}\varepsilon}\sqrt{\frac{G}{T}}\left(1+\frac{3\sqrt{G}}{\varepsilon\sqrt{T}}\right)+\frac{ C(\bolds{\gamma})}{\hat{s}_{\min} T}t\leq \frac{10\ln(T/\delta)}{\hat{s}_{\min}\varepsilon}\sqrt{\frac{G}{T}}\left(1+\frac{3\sqrt{G}}{\varepsilon\sqrt{T}}\right)+\frac{C(\bolds{\gamma})}{\hat{s}_{\min}T}.
\end{equation}
Noting that $\frac{3\sqrt{G}}{\varepsilon\sqrt{T}} \leq 1$, this completes the proof.
\end{proof}

\subsection{Simplifying the Number of Greedy Step Expression (Lemma~\ref{lem:bound-cgamma})}\label{app:bound-cgamma}

\begin{proof}
We first upper bound $C(\bolds{\gamma})$. For a group $g$ with $\gamma_g > \varepsilon_g / 2$, we have $16\chi\sqrt{\frac{2\ln(1/\beta)}{p_g T}} > \frac{\varepsilon_g}{2}$, which implies $p_g T \leq \frac{2^{11}\chi\ln(1/\beta)}{\varepsilon^2}$. As a result,
\begin{align}
C(\bolds{\gamma}) &= \frac{12}{\varepsilon}\sum_{g\colon \gamma_g \leq \frac{\varepsilon_g}{2}}\gamma_gp_g T + 6\sum_{g\colon\gamma_g > \frac{\varepsilon_g}{2}}p_g T \nonumber\\
&\leq \frac{12}{\varepsilon}\sum_{g \in \Sgro}16\chi\sqrt{2\ln(1/\beta)p_g T} + \frac{2^{13}\chi G\ln(1/\beta)}{\varepsilon}\tag{By definition of $\gamma_g$}\\
&\leq \frac{192\chi\ln(T/\delta)\sqrt{GT}\left(1 + 43\sqrt{G/T}\right)}{\varepsilon}. \tag{$\sqrt{2\ln(1/\beta)} \leq 2\ln(1/\beta) \leq \ln(T/\delta)$ if $T \geq 1849(M+G)$} \\
&\leq \frac{384\chi\ln(T/\delta)\sqrt{GT}}{\varepsilon}.
\end{align}
As a result,
\begin{align}
\dcbp &= \frac{1}{\hat{s}_{\min}}\left(\frac{20\sqrt{GT}\ln(1/\beta)}{\varepsilon}+\frac{51G\ln(1/\beta)}{\varepsilon^2}+C(\bolds{\gamma})\right) \nonumber\\
&\leq \frac{1}{\hat{s}_{\min}}\left(\frac{10\ln(T/\delta)\sqrt{GT}}{\varepsilon}\left(1+\frac{3\sqrt{G/T}}{\varepsilon}\right)+\frac{240\chi \ln(T/\delta)\sqrt{GT}}{\varepsilon}\right) \tag{$2\ln(1/\beta)\leq\ln(T/\delta)$} \nonumber\\
&\leq \frac{1}{\hat{s}_{\min}}\left(\frac{20\ln(T/\delta)\sqrt{GT}}{\varepsilon}+\frac{384\chi  \ln(T/\delta)\sqrt{GT}}{\varepsilon}\right) \leq \frac{404\chi \ln(T/\delta)\sqrt{GT}}{\hat{s}_{\min}\varepsilon} \tag{$T \geq 1849(M+G)$ and $\chi \geq 1$}
\end{align}
\end{proof}

\subsection{Predict-to-meet condition ensures low g-regret (Lemma~\ref{lem:fairness-bpstep})}\label{app:proof-fairness-bpstep}
We first show that under the setting of parameters in \eqref{eq:parameter-setting}, $\Psi(g,t,\beta)$ is a valid lower bound (with high probability) on future fairness score surplus for group $g$ assuming all remaining arrivals after case $t$ obtain greedy selections.
\begin{lemma}\label{lem:valid-lowerbound}
There exists an event $\set{S}_{\mathrm{low}}$ with probability at least $1 - 3GT\beta^4$ such that: $\forall g \in \Sgro, t \in [T]$, we have $\Psi(g,t,\beta) \leq \sum_{\tau \in \Sarr(g,T) \setminus \Sarr(g,t)} \left(w_{\tau,J^{\GR}(\tau)} - \expect{O_g}-\gamma_g\right)$.
\end{lemma}
\begin{proof}
Fix a group $g$. If $\gamma_g > \varepsilon_g / 2$, we immediately have \[\Psi(g,t,\beta) = -T \leq \sum_{\tau \in \Sarr(g,T) \setminus \Sarr(g,t)} \left(w_{\tau,J^{\GR}(\tau)} - \expect{O_g}-\gamma_g\right)\] since $\expect{O_g} + \gamma_g \leq 1$ by assumptions on $\gamma_g$. We thus focus on groups with $\gamma_g \leq \varepsilon_g/2$. Since the number of cases $T$ is known, greedy scores of arrivals from group $g$ counting backward from the end of the horizon can be viewed as a sequence of i.i.d. random variables. By Hoeffding's Inequality (Fact~\ref{fact:hoeffding}) and union bound over $1,\ldots,T$, we have with probability at least $1 - T\beta^4$ that for every $0\leq n \leq T$, the sum of greedy scores for the last $n$ arrivals from group $g$ is lower bounded by $ n\expect{w_{1,J^{\GR}(1)} \mid g(1) = g} - \sqrt{2n\ln(\nicefrac{1}{\beta})}$. Denote this event by $\set{S}^g_{\mathrm{low}}$. Note that for every $t \in [T]$, the number of group $g$ arrivals between $t + 1$ and $T$ is at most $T$. For ease of notation, let us write $N^f(g,t) = N(g,T)-N(g,t)$ for the number of "future" cases of group $g$ after case $t$. Then under $\set{S}^g_{\mathrm{low}}$, for every $t \in[T]$, we have
\begin{equation}\label{eq:future-greedy-score-group-g}
\sum_{\tau=t+1}^{T} \indic{g(\tau)=g}w_{\tau,J^{\GR}(\tau)} \geq N^f(g,t)\expect{w_{1,J^{\GR}(1)} \mid g(1) = g}-\sqrt{2N^f(g,t)\ln(\nicefrac{1}{\beta})}.
\end{equation}
Note that by definition $\expect{w_{1,J^{\GR}(1)} \mid g(1) = g} = \expect{O_g} + \varepsilon_g\geq \expect{O_g} + \gamma_g + \varepsilon_g / 2$. Under $\set{S}^g_{\mathrm{low}}$, \eqref{eq:future-greedy-score-group-g} implies
\begin{equation}\label{eq:future-greedy-score-group-g-second}
\sum_{\tau=t+1}^{T} \indic{g(\tau)=g}\left(w_{\tau,J^{\GR}(\tau)}-\expect{O_g}-\gamma_g\right) 
\geq \frac{N^f(g,t)\varepsilon_g}{2} - \sqrt{2N^f(g,t)\ln(\nicefrac{1}{\beta})}.
\end{equation}
We next lower bound the right hand side. Let $f(n) = \frac{n\varepsilon_g}{2}-\sqrt{2n\ln(\nicefrac{1}{\beta})}.$ We know $f(0) = 0$. Taking derivative of $f(n)$ gives $f'(n) = \varepsilon_g/2 - \sqrt{\frac{\ln(\nicefrac{1}{\beta})}{2n}}$. Setting the derivative to zero gives the unique stationary point $n^\star = \frac{\ln(\nicefrac{1}{\beta})}{\varepsilon_g^2}$. We then know that $\forall n \geq 0, f(n) \geq -\frac{\ln(\nicefrac{1}{\beta})}{\varepsilon_g}$. Since $N^f(g,t) \geq 0$, we have $\sum_{\tau=t+1}^{T} \indic{g(\tau)=g}\left(w_{\tau,J^{\GR}(\tau)}-\expect{O_g}-\gamma_g\right) \geq -\frac{\ln(\nicefrac{1}{\beta})}{\varepsilon_g} = \Psi(g,t,\beta)$. We then finish the proof by using the union bound over all groups and observing that the event $\set{S}_{\mathrm{low}}$ defined by $\{\cap_{g \in \Sgro} \set{S}_{\mathrm{low}}^g\}$ occurs with probability at least $1 - GT\beta^4$.
\end{proof}
We next finish the proof of Lemma~\ref{lem:fairness-bpstep}.
\begin{proof}[Proof of Lemma~\ref{lem:fairness-bpstep}]
For a group $g$, let us define an event $\set{S}_g$ under which the number of group $g$ arrivals in $[1,T]$ is close to its expectation (from below) and that in $[1,T-\dcbp]$ is close to its expectation from above. In particular, $\set{S}_g$ is the intersection of the event that $
N(g,T) \geq p_g T-3\sqrt{p_g T \ln(\frac{1}{\beta})}$ and the event that $N(g,T)-N(g,T-\dcbp) \leq p_g\dcbp+5\sqrt{\max\left( p_g\dcbp,\ln(\frac{1}{\beta})\right)\ln(\frac{1}{\beta})}.$
Then by concentration bound (Lemma~\ref{lem:chernoff}), we have $\Pr\{\set{S}_g\} \geq 1 - 2\beta^4$. Let us define $\Sbp$ by
\begin{equation}\label{eq:def-sbp}
\Sbp=\{\Temp \geq T-\dcbp\} \cap \set{S}_{\mathrm{low}} \cap \left(\cap_{g \in \Sgro}\set{S}_g\right)
\end{equation}
where $\set{S}_{\mathrm{low}}$ is defined in Lemma~\ref{lem:valid-lowerbound}. By Lemma~\ref{lem:cons-temp} and Lemma~\ref{lem:valid-lowerbound}, we have $\Pr\{\Temp \geq T-\dcbp\} \geq 1-(2T+1)(M+G)\beta^4$ and $\Pr\{\set{S}_{\mathrm{low}}\} \geq 1 - GT\beta^4$. Using union bound gives $\Pr\{\Sbp\} \geq 1-(3T+1)(M+G)\beta^4-2G\beta^4 \geq 1-(3T+3)(M+G)\beta^4.$ We next show that condition on $\Sbp$, we have the desired properties in Lemma~\ref{lem:fairness-bpstep}. We first have $\Temp \geq T-\dcbp$ since $\Sbp \subseteq \{\Temp \geq T-\dcbp\}$. 

Let us then fix a group $g \in \Sgro$. Suppose that condition~\eqref{eq:cond-predict} is satisfied for case $t$ and the case is of group $g$. If there are multiple such cases, we take $t$ to be the label of the last one. Our goal is to show that condition on $\Sbp$ \eqref{eq:def-sbp} and assuming the existence of $t$, we have $\alpha_g \geq \expect{O_g} - \frac{12\dcbp}{T}$. By condition~\eqref{eq:cond-predict}, we have $V_g[t-1] - \expect{O_g} + \Psi(g,t,\beta) \geq \Cex$. Then by the definition of $V_g[t-1]$, the fact that $\Sbp \subseteq \set{S}_{\mathrm{low}}$,  Lemma~\ref{lem:valid-lowerbound} and \eqref{eq:parameter-setting}, we have
\begin{equation}\label{eq:lower-bound-total-score}
\sum_{\tau \in \Sarr(g,t-1)}\left(w_{\tau,J^{\alg}(\tau)}-\expect{O_g}-\gamma_g\right)-\expect{O_g}-\gamma_g + \sum_{\tau = t + 1,g(\tau)=g}^T \left(w_{\tau,J^{\GR}}(\tau)-\expect{O_g}-\gamma_g\right) \geq 6\ln(\nicefrac{1}{\beta}),
\end{equation}
which by reorganizing terms implies \[\sum_{\tau \in \Sarr(g,t)}w_{\tau,J^{\alg}(\tau)} + \sum_{\tau = t + 1}^T \indic{g(\tau)=g} w_{\tau,J^{\GR}}(\tau) - 6\ln(\nicefrac{1}{\beta})\geq N(g,T)\left(\expect{O_g}+\gamma_g\right).\]
Recall that $\alpha_g = \frac{1}{N(g,T)}\sum_{\tau \in \Sarr(g,T)}w_{\tau,J^{\alg}(\tau)}$. Using the definition of $\alpha_g$ and \eqref{eq:lower-bound-total-score} gives
\begin{equation}\label{eq:lower-bound-avg-score}
\alpha_g \geq \expect{O_g}+\gamma_g - \frac{1}{N(g,T)}\left(\sum_{\tau=t+1}^T \indic{g(\tau)=g}\left(w_{\tau,J^{\GR}(\tau)}-w_{\tau,J^{\alg}(\tau)}\right) - 6\ln(\nicefrac{1}{\beta})\right).
\end{equation}
Note that we have $N(g,T) \geq 1$ due to the existence of $t$. It remains to show that the second term on the right hand side can be upper bounded by $\frac{12\dcbp}{T}$. Notice that $\Temp \geq T-\dcbp$ and $t$ is the last case of group $g$ for whom condition~\eqref{eq:cond-predict} is satisfied. We next consider two scenarios. On the one hand, if $t \geq T-\dcbp$, we must have $\sum_{\tau=t+1}^T \indic{g(\tau)=g}\left(w_{\tau,J^{\GR}(\tau)}-w_{\tau,J^{\alg}(\tau)}\right) \leq N(g,T)-N(g,T-\dcbp)$ since scores are in $[0,1]$. On the other hand, suppose that $t < T-\dcbp$. Under $\Sbp$, we have $\Temp \geq T-\dcbp$ and thus $J^{\alg}(\tau) = J^{\CONS}(\tau)$ for every $\tau \in [t,T-\dcbp]$. In addition, since $t$ is the last case of group $g$ such that condition~\eqref{eq:cond-predict} is satisfied, we have $J^{\alg}(\tau) = J^{\CONS}(\tau) = J^{\GR}(\tau)$ for every $\tau \in [t + 1,T-\dcbp]$ with $g(\tau) = g$. As a result, under the scenario that $t < T-\dcbp$, we again have
\begin{align*}
\sum_{\tau=t+1}^T \indic{g(\tau)=g}\left(w_{\tau,J^{\GR}(\tau)}-w_{\tau,J^{\alg}(\tau)}\right) &= \sum_{\tau=T-\dcbp + 1}^T \indic{g(\tau)=g}\left(w_{\tau,J^{\GR}(\tau)}-w_{\tau,J^{\alg}(\tau)}\right) \\
&\leq N(g,T) - N(g,T-\dcbp).
\end{align*}
Summarizing the two scenarios and with \eqref{eq:lower-bound-avg-score}, we have $\alpha_g \geq \expect{O_g}+\gamma_g - \frac{N(g,T)-N(g,T-\dcbp) - 6\ln(\nicefrac{1}{\beta})}{N(g,T)}$.
Recall the definition of $\set{S}_g$ and the fact that $\Sbp \subseteq \set{S}_g$ by \eqref{eq:def-sbp}. We have that $N(g,T) \geq p_gT-3\sqrt{p_gT\ln(\nicefrac{1}{\beta})}$ and that~$N(g,T)-N(g,T-\dcbp) \leq p_g\dcbp+5\sqrt{\max\left(p_g\dcbp,\ln(\nicefrac{1}{\beta})\right)\ln(\nicefrac{1}{\beta})}$. Consider two scenarios on the size of group $g$. First, if $p_g\dcbp \leq \ln(\nicefrac{1}{\beta})$, we have that $N(g,T)-N(g,T-\dcbp) \leq \ln(\nicefrac{1}{\beta}) + 5\ln(\nicefrac{1}{\beta}) \leq 6\ln(\nicefrac{1}{\beta})$. Therefore, for a group $g$ with $p_g\dcbp \leq \ln(\nicefrac{1}{\beta})$, 
\[
\alpha_g \geq \expect{O_g}+\gamma_g - \frac{N(g,T)-N(g,T-\dcbp) - 6\ln(\nicefrac{1}{\beta})}{N(g,T)} \geq \expect{O_g}+\gamma_g - \frac{6\ln(\nicefrac{1}{\beta})-6\ln(\nicefrac{1}{\beta})}{N(g,T)} = \expect{O_g}+\gamma_g.
\]
Otherwise, suppose that $p_g\dcbp \geq \ln(\nicefrac{1}{\beta})$. We then have $N(g,T)-N(g,T-\dcbp) \leq 6p_g\dcbp$. In addition,
\[N(g,T) \geq p_gT - 3\sqrt{p_g T \ln(\nicefrac{1}{\beta})} \geq p_g T - 3\sqrt{p_g T \left(p_g \dcbp\right)} \geq p_g T - 3\sqrt{p_g T \left(\frac{1}{36}p_g T\right)} = \frac{1}{2}p_g T,
\]
where the last inequality is by assumption that $T\geq 36\dcbp$. As a result, when $p_g\dcbp \geq \ln(\nicefrac{1}{\beta})$, we have $\alpha_g \geq \expect{O_g}+\gamma_g - \frac{N(g,T)-N(g,T-\dcbp)}{N(g,T)} \geq \frac{6p_g\dcbp}{0.5p_gT} = \frac{12\dcbp}{T}$. Summarizing the above two scenarios shows that condition on $\Sbp$, for every group $g$, if there is an arrival satisfying condition~\eqref{eq:cond-predict}, we have that $\alpha_g\geq \min\left(\expect{O_g}+\gamma_g,\expect{O_g}+\gamma_g-\frac{12\dcbp}{T}\right) = \expect{O_g}+\gamma_g-\frac{12\dcbp}{T}$.
\end{proof}

\subsection{Guarantee of ex-post $g-$regret for large groups (Lemma~\ref{lem:fair-case-1})}\label{app:proof-fair-case-1}
\begin{proof}
Fix a group $g$ with $p_g \geq \frac{81\ln(\nicefrac{1}{\beta})}{\varepsilon_g^2 T}$  and $\gamma_g \leq \varepsilon_g / 2$. Our proof involves the following steps: 1) we show that there are sufficiently many arrivals of group $g$ before $T-\dcbp$; 2) we show that under this event, the total greedy score is so high that \textsc{Conservative Bid Price Control} must assign at least one arrival of group $g$ before case $T-\dcbp$ to $J^{\BP}$; 3) we apply Lemma~\ref{lem:fairness-bpstep} for the final result. 

We first show there are sufficient arrivals of group $g$ before $T-\dcbp$. Let $n^\star = \frac{32\ln(\nicefrac{1}{\beta})}{\varepsilon_g^2}$. Define $\set{S}_1 = \{N(g,T-\dcbp) \geq n^\star + 1\}$. Note that by assumption, $T \geq 36\dcbp$ and thus we know that 
$
\expect{N(g,T-\dcbp)} = p_g T - p_g\dcbp\geq \frac{35}{36}p_g T \geq \ln(\nicefrac{1}{\beta}).
$
Using the first probability bound of Lemma~\ref{lem:chernoff}, with probability at least $1-\beta^4$, we have
$N(g,T-\dcbp) \geq p_g T - p_g\dcbp-3\sqrt{p_gT\ln(\nicefrac{1}{\beta})}.$
Since $p_g T \geq 81\ln(\nicefrac{1}{\beta})\geq 81$, we have $3\sqrt{p_g T \ln(\nicefrac{1}{\beta})} \leq 
3{\nicefrac{(p_g T)^2}{81}} = \frac{1}{3}p_g T$. Therefore, with probability at least $1-\beta^4$, $N(g,T-\dcbp)$ is lower bounded by
\[
p_g T - p_g \dcbp-3\sqrt{p_gT\ln(\nicefrac{1}{\beta})} \geq \frac{35}{36}p_g T - \frac{1}{3}p_g T \geq \frac{34}{36}p_g T - \frac{1}{3}p_g T + 1 = \frac{11}{18}p_g T + 1 \geq \frac{32\ln(\nicefrac{1}{\beta})}{\varepsilon_g^2} + 1.
\]
Therefore, $\Pr\{\set{S}_1\} \geq 1 - \beta^4$. 

We next define $\set{S}_2$ as the event that the total greedy score of the first $n^\star$ arrivals of group $g$ is at least $n^\star \expect{w_{1,J^{\GR}(1)} \mid g(1)=g} - \sqrt{2n^\star \ln(\nicefrac{1}{\beta})}$. Since scores of arrivals are i.i.d., using Hoeffding's inequality (Fact~\ref{fact:hoeffding}) gives $\Pr\{\set{S}_2\} \geq 1 - \beta^4$. Recall event $\Sbp$ in Lemma~\ref{lem:fairness-bpstep}. Under $\Sbp$, we have $\Temp \geq T-\dcbp$ and that if a group $g$ receives an assignment $J^{\BP}$ during $[1,\Temp],$ we have $\alpha_g \geq \expect{O_g}+\gamma_g-\frac{12\dcbp}{T}$. Condition on $\set{S}_1 \cap \set{S}_2 \cap \Sbp$. It remains to show that there exists a case $t \leq \Temp$, such that $g(t) = g$ and condition~\eqref{eq:cond-predict} holds (so group $g$ receives an assignment $J^{\BP}$.) Note that conditioned on $\Sbp$, we have $\Temp \geq T-\dcbp$. In addition, by $\set{S}_1$, the number of group $g$ arrivals in the first $T-\dcbp$ periods is at least $n^\star + 1$. Consider two scenarios. First, assume that for the first $n^\star$ group $g$ arrivals, condition~\eqref{eq:cond-predict} holds at least once. Then we know there must exist $t' \leq T-\dcbp \leq \Temp$ with $g(t') = g$ satisfying condition~\eqref{eq:cond-predict}. Otherwise, the first $n^\star$ group $g$ cases all receive their greedy assignments. Let $t'$ be the label of the $n^\star + 1$th cases of group $g$. We next show condition~\eqref{eq:cond-predict} holds for case $t'$. We first have $t' \leq T-\dcbp$ since $N(g,T-\dcbp) \geq n^\star + 1$. The definition of $V_g[t'-1]$ gives $V_g[t'-1]=\sum_{\tau < t',~g(\tau)=g} (w_{\tau,J^{\alg}(\tau)} - \expect{O_g}-\gamma_g)$. We then have 
\begin{align*}
 V_g[t'-1] &= \sum_{\tau \in \Sarr(g,t'-1)} w_{\tau,J^{\GR}(\tau)} - n^\star (\expect{O_g}+\gamma_g) \\
 &\geq n^\star \expect{w_{1,J^{\GR}(1)}\mid g(1)=g}-\sqrt{2n^\star \ln(\nicefrac{1}{\beta})}-n^\star(\expect{O_g}+\gamma_g) \geq \frac{n^\star\epsilon_g}{2} - \sqrt{2n^\star \ln(\nicefrac{1}{\beta})},
\end{align*}
where the first equality is because the first $n^\star$ cases of group $g$ receive greedy assignments; the second inequality is by event $\set{S}_2$; the third inequality is by the slackness property (Definition~\ref{def:slackness}) and the assumption that $\gamma_g \leq \varepsilon_g / 2$. By the definition of $n^\star = \frac{32\ln(\nicefrac{1}{\beta})}{\varepsilon_g^2}$, under the parameter setting in \eqref{eq:parameter-setting}, we have
\begin{align*}
V_g[t'-1] - \expect{O_g} - \gamma_g + \Psi(g,t',\beta) - \mathrm{Buf}(\beta) &\geq
V_g[t'-1] - \left(1+\frac{\ln(\nicefrac{1}{\beta})}{\varepsilon_g}+6\ln(\nicefrac{1}{\beta})\right)\\
&\geq \frac{16\ln(\nicefrac{1}{\beta})}{\varepsilon_g}- 6\ln(\nicefrac{1}{\beta})-\frac{8\ln(\nicefrac{1}{\beta})}{\varepsilon_g}-\frac{\ln(\nicefrac{1}{\beta})+1}{\varepsilon_g} \\
&\geq \frac{\ln(\nicefrac{1}{\beta})}{\varepsilon}\left(16-6-8-2\right) \geq 0
\end{align*}
where for the third inequality we use the assumption that $\beta \leq e^{-1}$. Therefore, for case $t'$, condition~\eqref{eq:cond-predict} holds. Condition on $\set{S}_1 \cap \set{S}_2 \cap \Sbp$, we have that group $g$ receives an assignment $J^{\BP}$ during $[1,\Temp]$ and $\alpha_g \geq \expect{O_g}+\gamma_g-\frac{12\dcbp}{T}$ by Lemma~\ref{lem:fairness-bpstep}. We then finish the proof using $\Pr\{\Sbp\} \geq 1-(5T+1)(M+G)\beta^4$ and noticing that \[\Pr\{S_1 \cap S_2 \cap \Sbp\} \geq 1 - (5T+1)(M+G+1)\beta^4 - 2\beta^4 \geq 1 - (5T+3)(M+G+1)\beta^4.\]
\end{proof}
\subsection{Guarantee of g-regret for small groups (Lemma~\ref{lem:fair-case-2})}\label{app:proof-fair-case-2}
\begin{proof}
Recall event $\Sbp$ defined before. This event happens with probability at least $1-(5T+3)(M+G+1)\beta^4$ by Lemma~\ref{lem:fairness-bpstep}.

Condition on $\Sbp$. If a case of group $g$ in the first $\Temp$ periods receives an assignment $J^{\BP}$, by the definition of $\Sbp$ before Lemma~\ref{lem:fairness-bpstep}, we have $\alpha_g \geq \expect{O_g}+\gamma_g - \frac{12\dcbp}{T}$, proving the desired result. 
We thus only need to consider the scenario that all cases of group $g$ in the first $\Temp$ periods receive their greedy assignments. Without loss of generality, let us assume $N(g,T) \geq 1$ since otherwise $\alpha_g = O_g = 0$. By the ex-post feasibility of the fairness rule (see Definition~\ref{def:fairness_feasible}), we must have that $O_g$ is upper bounded by the average greedy score given by $\frac{1}{N(g,T)}\sum_{t \in \Sarr(g,T)} w_{t,J^{\GR}(t)}$. Given the condition that all arrivals of group $g$ in the first $\Temp$ periods receive greedy assignments under \textsc{Conservative Bid Price Control}, we have
\[
O_g - \alpha_g \leq \frac{\sum_{t \in \Sarr(g,T)} w_{t,J^{\GR}(t)}-\sum_{t \in \Sarr(g,T)}w_{t,J^{\alg}(t)}}{N(g,T)}.\] The right hand side is equal to $\frac{\sum_{t \in \Sarr(g,T)\setminus \Sarr(g,\Temp)} \left(w_{t,J^{\GR}(t)} - w_{t,J^{\alg}(t)}\right)}{N(g,T)}$ because all cases of group $g$ in the first $\Temp$ periods receive their greedy assignments. Under $\Sbp$, we know $\Temp \geq T-\dcbp$ and thus \[\sum_{t \in \Sarr(g,T)\setminus \Sarr(g,\Temp)} \left(w_{t,J^{\GR}(t)} - w_{t,J^{\alg}(t)}\right) \leq \sum_{t \in \Sarr(g,T)\setminus \Sarr(g,T-\dcbp)} \left(w_{t,J^{\GR}(t)} - w_{t,J^{\alg}(t)}\right),\]
which implies $O_g-\alpha_g \leq \frac{N(g,T) - N(g,T-\dcbp)}{N(g,T)}$.

Define event $\set{S}_1$ by $\{N(g,T)-N(g,T-\dcbp) = 0\}$. That is, there is no group $g$ arrival among the last $\dcbp$ arrivals. Then by union bound $\Pr\{\set{S}_1\} \geq 1 - p_g\dcbp$. In addition, let us define event $\set{S}_g$ (the same as in the proof of Lemma~\ref{lem:fairness-bpstep}) by the intersection of the event that $N(g,T) \geq p_g T-3\sqrt{p_g T \ln(\nicefrac{1}{\beta})}$ and the event that $N(g,T)-N(g,T-\dcbp) \leq p_g\dcbp+5\sqrt{\max\left( p_g\dcbp,\ln(\nicefrac{1}{\beta})\right)\ln(\nicefrac{1}{\beta})}$.
Using Lemma~\ref{lem:chernoff} and union bound, we have $\Pr\{\set{S}_g\} \geq 1 - 2\beta^4$. Therefore, $\Pr\{\Sbp \cap \set{S}_1 \cap \set{S}_g\} \geq 1-(5T+3)(M+G+1)\beta^4 - 2\beta^4 - p_g\dcbp \geq 1 - (5T+5)(M+G)\beta^4-p_g\dcbp$.  
Condition on $\Sbp \cap \set{S}_1 \cap \set{S}_g$. It remains to show that  $\alpha_g \geq \min(\expect{O_g}+\gamma_g,O_g)-\frac{12\dcbp}{T}$. 

Moreover, under $\set{S}_1$, there is no group-$g$ arrival after $T - \dcbp$, i.e., $\Sarr(g,T)\setminus \Sarr(g,T-\dcbp) = \emptyset$. As a result, $O_g - \alpha_g \leq 0$. 

Putting together the proof, we know that condition on $\Sbp \cap \set{S}_1 \cap \set{S}_g$, for a group $g$, there are two scenarios. First, an arrival before $\Temp$ receives an assignment $J^{\BP}$ and then $\alpha_g \geq \expect{O_g}+\gamma_g-\frac{12\dcbp}{T}$. Second, all arrivals before $\Temp$ receive greedy assignments. Under this scenario, we show $\alpha_g \geq O_g$ since there is no group-$g$ arrival after $\Temp$. As a result, for a group $g$, we have $\alpha_g \geq \min\left(\expect{O_g}+\gamma_g,O_g\right) - \frac{12\dcbp}{T}$ with probability lower bounded by~$\Pr\{\Sbp \cap \set{S}_1 \cap \set{S}_g\} \geq 1-(5T+5)(M+G)\beta^4-p_g\dcbp$.
\end{proof}

\section{Concentration inequalities}\label{app:conc}
% !TEX root = main.tex
\subsection{Hoeffding's and Bernstein's Inequalities}
In the proof, we frequently use Hoeffding's Inequality simplified from \cite{boucheron2013concentration}.
\begin{fact}[Hoeffding's Inequality]\label{fact:hoeffding}
Given $n$ independent random variables $X_i$ taking values in $[0,1]$ almost surely. Let $X = \sum_{i=1}^n X_i$. Then for any $x > 0$,
\begin{equation}\label{eq:chernoff-bound}
    \Pr\{X < \expect{X} - x\} \leq e^{-2x^2/n};~\Pr\{X > \expect{X} + x\} \leq e^{-2x^2/n}.
\end{equation}
\end{fact}
Another useful inequality is the Bernstein's Inequality. The following version is by Corollary 2.11 from \cite{boucheron2013concentration}.
\begin{fact}[Bernstein's Inequality]\label{fact:bernstein}
Given $n$ independent random variables $X_i$ such that $X_i \leq b$ almost surely and $\sum_{i=1}^n \expect{X_i^2} \leq \nu$. Let $X = \sum_{i=1}^n X_i$. Then for all $x > 0$,
\[
\Pr\{X > \expect{X}+x\} \leq \exp\left(\frac{-x^2}{2(\nu + bx/3)}\right).
\]
\end{fact}
An implication of the above result is the following lemma. 
\begin{lemma}\label{lem:bernstein-implied}
Given $n$ independent random variables $X_i$ such that $X_i \geq 0$ almost surely and $\sum_{i=1}^n \expect{X_i^2} \leq \nu$. Let $X = \sum_{i=1}^n X_i$. Then for all $x > 0$, 
\[
\Pr\{X < \expect{X}-x\} \leq \exp\left(\frac{-x^2}{2\nu}\right).
\]
\end{lemma}
\begin{proof}
Take $Y_i = -X_i$. Then $\sum_{i=1}^n \expect{Y_i^2} \leq \nu$ and $Y_i \leq 0$ almost surely. By Fact~\ref{fact:bernstein}, for any $x > 0$, we have $\Pr\{\sum_{i=1}^n Y_i > \expect{\sum_{i=1}^n Y_i} + x\} \leq \exp\left(\frac{-x^2}{2\nu}\right)$. Rearranging terms give the desired result.
\end{proof}

We finally recall the Chernoff bound of Bernoulli random variables and an implied version of it. We summarize the result as follows.
\begin{lemma}\label{lem:chernoff}
Let $X$ be the sum of $n$ independent Bernoulli random variables. Then, for any $x > 0$, we have
\[
\Pr\{X < \expect{X}-x\} \leq \exp\left(-x^2/(2\expect{X})\right);~\Pr\{X>\expect{X}+x\} \leq \exp\left(-\frac{x^2}{2(\expect{X}+x/3)}\right).
\]
In particular, the second inequality implies that for any $\beta > 0$, we have \[\Pr\left\{X > \expect{X}+5\sqrt{\max(\expect{X},\ln\frac{1}{\beta})\ln\frac{1}{\beta}}\right\} \leq \beta^4.\]
\end{lemma}
\begin{proof}
The first two probability bounds are implied by the Bernstein's Inequality (Fact~\ref{fact:bernstein}) since $X$ is a sum of Bernoulli random variables. We prove the third bound here. By the second probability bound, we have
\begin{align*}
&\mspace{32mu}\Pr\left\{X > \expect{X}+5\sqrt{\max(\expect{X},\ln\frac{1}{\beta})\ln\frac{1}{\beta}}\right\} \\
&\leq \exp\left(-\frac{25\max(\expect{X},\ln\frac{1}{\beta})\ln\frac{1}{\beta}}{2\expect{X} + 4\sqrt{\max(\expect{X},\ln\frac{1}{\beta})\ln\frac{1}{\beta}}}\right) \\
&\leq \exp\left(-\frac{25\max(\expect{X},\ln\frac{1}{\beta})\ln\frac{1}{\beta}}{6\max(\expect{X},\ln\frac{1}{\beta})}\right) \leq \exp\left(-4\ln\frac{1}{\beta}\right) = \beta^4.
\end{align*}
\end{proof}
\subsection{McDiarmid's Inequality}\label{sec:mcdiarmid}
We use a Bernstein type of McDiarmid's Inequality to show concentration of single-valued multivariate functions. Our notations below follow those from \cite[Section 3.2]{mcdiarmid1998concentration}. Consider $n$ independent random variables $\bolds{X} = (X_1,\ldots,X_n)$ with $X_k \in A_k$ and a real-valued function $f(x_1,\ldots,x_n)$ defined on $\prod_{k \leq n} A_k$. For $k \leq n$, define a function $h_k(x_1,\ldots,x_k) = \expectsub{\bolds{X}}{f(\bolds{X}) | X_i = x_i, i \leq k} - \expectsub{\bolds{X}}{f(\bolds{X}) | X_i = x_i, i \leq k - 1}$. We can accordingly define the variance \[\var_k(x_1,\ldots,x_{k-1}) = \var_{Y\sim X_k}\left(h_k(x_1,\ldots,x_{k-1},Y)\right).\] Define the maximum positive deviation by $\mathrm{maxdev}^+ = \sup_{\bolds{x} \in \prod_{k \leq n} A_k} \max_k h_k(x_1,\ldots,x_{k})$. For any $\bolds{x} \in \prod_{k} A_k$, define the sum of variance by $V(\bolds{x}) = \sum_{k\leq n} \var_k(x_1,\ldots,x_{k-1})$. The following result is from  \cite[Theorem~3.8]{mcdiarmid1998concentration}.
\begin{fact}\label{fact:mcdiarmid-bernstein}
Let $\bolds{X} = (X_1,\ldots,X_n)$ be a sequence of independent random variables with $X_k \in A_k$ and  $f(x_1,\ldots,x_n)$ is a bounded real-valued function over $\prod A_k$. Let $b = \mathrm{maxdev}^+$ and $\hat{v} = \sup_{\bolds{x} \in \prod_k A_k} V(\bolds{x})$, both of which are assumed to be finite. Then for any $d > 0$,
\[
\Pr\{f(\bolds{X}) - \expect{f(\bolds{X})} \geq d\} \leq \exp\left(-\frac{d^2}{2(\hat{v} + bd/3)}\right). 
\]
\end{fact}
\end{document}